\documentclass[%
superscriptaddress,
twocolumn,
nofootinbib,
]{revtex4-2}

\usepackage{graphicx}
\usepackage{comment}

\usepackage{dcolumn}
\usepackage{bm,bbm}
\usepackage{physics}
\usepackage{amsmath}
\usepackage{amsthm}
\usepackage{amsfonts}
\usepackage[colorlinks=true,citecolor=blue,linkcolor=magenta]{hyperref}
\usepackage{xcolor}
\usepackage{tikz}
\graphicspath{{Images/}}

\usepackage[export]{adjustbox}

\usepackage{algorithm2e}

\usepackage{algorithm}
\usepackage{algorithmicx}
\usepackage{algpseudocode} 

\newtheorem{theorem}{Theorem}
\newtheorem{corollary}{Corollary} 
\newtheorem{lemma}{Lemma}
\newtheorem{definition}{Definition}

\newtheorem{proposition}{Proposition}
\newtheorem{assumption}{Assumption}

\newtheorem{theoremsup}{Supplemental Theorem}

\newcommand{\poly}{\operatorname{poly}}

\newcommand{\Ebb}{\mathbb{E}}

\newcommand{\AC}{\mathcal{A}}

\newcommand{\DC}{\mathcal{D}}
\newcommand{\EC}{\mathcal{E}}

\newcommand{\MC}{\mathcal{M}}

\newcommand{\OC}{\mathcal{O}}
\newcommand{\PC}{\mathcal{P}}
\newcommand{\LC}{f}

\newcommand{\SC}{\mathcal{S}}

\newcommand{\Var}{{\rm Var}}

\renewcommand{\geq}{\geqslant}
\renewcommand{\leq}{\leqslant}

\DeclareMathOperator*{\argmax}{arg\,max}

\renewcommand{\vec}[1]{\boldsymbol{#1}}

\newcommand{\al}{\alpha }

\newcommand{\om}{\omega }

\newcommand{\thv}{\vec{\alpha}}
 
\newcommand{\alv}{\vec{\alpha}}

\graphicspath{{Images/}}

\newcommand{\be}{\begin{equation}}
\newcommand{\ee}{\end{equation}}

\newcommand{\1}{\mathbbm{1}}

\renewcommand{\vec}[1]{\boldsymbol{#1}}

\usepackage{mathtools}

\newcommand{\fsur}{\ensuremath{\widetilde{f}}}

\usepackage{amssymb}

\newcommand{\lf}{\LC(\boldsymbol{\alpha})}
\newcommand{\uni}{\boldsymbol{\mathcal{D}}_{\nparams}} 
\newcommand{\vol}{\boldsymbol{\mathcal{V}}_{\nparams}} 
\newcommand{\gradlk}{\partial_{i_1,..,i_k} \LC (\boldsymbol{\alpha}^*)}
\newcommand{\dellk}{\Delta_{i_1,..,i_k}(\vec{\alpha},\vec{\alpha}^*)}

\renewcommand{\alv}{\vec{\alpha}}

\newcommand{\omv}{\vec{\omega}}

\newcommand{\idx}[2]{\omv_{(#1, #2)}}

\newcommand{\dlv}{\vec{\delta}}
\newcommand{\nparams}{\ensuremath{m} }
\newcommand{\expansionorder}{\ensuremath{\kappa} }

\usepackage[normalem]{ulem}

\usepackage{cancel}

\definecolor{new}{RGB}{70, 142, 171}

\usepackage{minitoc}
\usepackage{tocloft}

\usepackage{ulem}

\newcommand{\cen}{\al^*}
\newcommand{\cenv}{\alv^*}
\newcommand{\apath}{\upsilon}
\newcommand{\Apath}{\Upsilon}
\renewcommand{\bold}[1]{\mathbf{#1}}

\newcommand{\supt}{Supplemental Theorem}

\makeatletter 
\renewcommand\onecolumngrid{
\do@columngrid{one}{\@ne}
\def\set@footnotewidth{\onecolumngrid}
\def\footnoterule{\kern-6pt\hrule width 1.5in\kern6pt}
}
\renewcommand\twocolumngrid{
        \def\footnoterule{
        \dimen@\skip\footins\divide\dimen@\thr@@
        \kern-\dimen@\hrule width.5in\kern\dimen@}
        \do@columngrid{mlt}{\tw@}
}
\makeatother

\begin{document}

\doparttoc 
\faketableofcontents 
\part{}

\title{Efficient quantum-enhanced classical simulation for patches of quantum landscapes}

\date{\today}

\author{Sacha Lerch}
\thanks{The first three authors contributed equally to this work.}
\affiliation{Institute of Physics, Ecole Polytechnique F\'{e}d\'{e}rale de Lausanne (EPFL), CH-1015 Lausanne, Switzerland}
\affiliation{Centre for Quantum Science and Engineering, \'Ecole Polytechnique F\'{e}d\'{e}rale de Lausanne (EPFL),   Lausanne, Switzerland}

\author{Ricard Puig}
\thanks{The first three authors contributed equally to this work.}
\affiliation{Institute of Physics, Ecole Polytechnique F\'{e}d\'{e}rale de Lausanne (EPFL), CH-1015 Lausanne, Switzerland}
\affiliation{Centre for Quantum Science and Engineering, \'Ecole Polytechnique F\'{e}d\'{e}rale de Lausanne (EPFL),   Lausanne, Switzerland}

\author{Manuel S. Rudolph}
\thanks{The first three authors contributed equally to this work.}
\affiliation{Institute of Physics, Ecole Polytechnique F\'{e}d\'{e}rale de Lausanne (EPFL), CH-1015 Lausanne, Switzerland}
\affiliation{Centre for Quantum Science and Engineering, \'Ecole Polytechnique F\'{e}d\'{e}rale de Lausanne (EPFL),   Lausanne, Switzerland}

\author{\\Armando Angrisani}
\affiliation{Institute of Physics, Ecole Polytechnique F\'{e}d\'{e}rale de Lausanne (EPFL), CH-1015 Lausanne, Switzerland}
\affiliation{Centre for Quantum Science and Engineering, \'Ecole Polytechnique F\'{e}d\'{e}rale de Lausanne (EPFL),   Lausanne, Switzerland}

\author{Tyson Jones}
\affiliation{Institute of Physics, Ecole Polytechnique F\'{e}d\'{e}rale de Lausanne (EPFL), CH-1015 Lausanne, Switzerland}
\affiliation{Centre for Quantum Science and Engineering, \'Ecole Polytechnique F\'{e}d\'{e}rale de Lausanne (EPFL),   Lausanne, Switzerland}

\author{M. Cerezo}
\affiliation{Los Alamos National Laboratory, Los Alamos, NM 87545, USA}
\affiliation{Quantum Science Center, Oak Ridge, TN 37931, USA}

\author{Supanut Thanasilp}
\affiliation{Institute of Physics, Ecole Polytechnique F\'{e}d\'{e}rale de Lausanne (EPFL), CH-1015 Lausanne, Switzerland}
\affiliation{Centre for Quantum Science and Engineering, \'Ecole Polytechnique F\'{e}d\'{e}rale de Lausanne (EPFL),   Lausanne, Switzerland}
\affiliation{Chula Intelligent and Complex Systems Center of Excellence, Department of Physics, Faculty of Science, Chulalongkorn University, Bangkok, Thailand, 10330}
\affiliation{Siam Quantum Square, Faculty of Science, Chulalongkorn University, Bangkok, Thailand}

\author{Zo\"{e} Holmes}
\affiliation{Institute of Physics, Ecole Polytechnique F\'{e}d\'{e}rale de Lausanne (EPFL), CH-1015 Lausanne, Switzerland}
\affiliation{Centre for Quantum Science and Engineering, \'Ecole Polytechnique F\'{e}d\'{e}rale de Lausanne (EPFL),   Lausanne, Switzerland}

\begin{abstract}
Understanding the capabilities of classical simulation methods is key to identifying where quantum computers are advantageous. Not only does this ensure that quantum computers are used only where necessary, but also one can potentially identify subroutines that can be offloaded onto a classical device. In this work, we show that it is always possible to generate a \textit{classical surrogate} of a sub-region (dubbed a ``patch") of an expectation landscape produced by a parameterized quantum circuit. That is, we provide a quantum-enhanced classical algorithm which, after simple measurements on a quantum device, allows one to classically simulate approximate expectation values of a subregion of a landscape. We provide time and sample complexity guarantees for a range of families of circuits of interest, and further numerically demonstrate our simulation algorithms on an exactly verifiable simulation of a Hamiltonian variational ansatz and long-time dynamics simulation on a 127-qubit heavy-hex topology.

\end{abstract}

\maketitle

\section{Introduction}
In a well-spirited effort to utilize quantum devices to the best of their abilities, the common approach could be summarized as: \textit{Do as little as possible on the quantum computer by offloading work onto classical machines}. However, this does not relieve quantum computers from the necessary and sufficient conditions to run quantum circuits that are both i) classically hard and ii) practically useful. With classical algorithms still outperforming current quantum computers for dynamical simulation~\cite{kim2023evidence, rudolph2023classical, tindall2023efficient, beguvsic2024fast, torre2023dissipative, liao2023simulation, thomson2024unravelling}, recent developments in variational quantum computing~\cite{cerezo2023does} and no clear trajectory for making use of near-term quantum devices, we are left to wonder: \textit{should we offload more work to classical computers?}

In this work, we provide theoretical guarantees and practical implementations of fully classical and quantum-enhanced classical simulations of \textit{quantum landscape patches}. 
That is, for simulating a small region in the full expectation landscape spanned by the gate angles of a quantum circuit.
We consider a \textit{spectrum of classicality}, where quantum computers are used either i) not at all, ii) to provide measurements of the initial states, or iii) to estimate the most important contributions to a quantum landscape. In the former two cases, our results can be viewed as providing guarantees for the simulation of low-magic or near-Clifford circuits~\cite{dowling2024magic}. In the latter two cases, the quantum computer is effectively tasked with collecting data to build a so-called \textit{classical surrogate}~\cite{schreiber2022classical, jerbi2023shadows, landman2022classically, angrisani2023learning, fontana2023classical, rudolph2023classical} of the patch of a quantum landscape. As such, our work establishes a hybrid quantum-classical simulation framework that attempts to realize the original motivation of utilizing quantum resources only for the smallest possible but crucial component. 

Our results are applicable to variational quantum computing where the landscape corresponds to a loss function which one seeks to minimize to solve a given task~\cite{cerezo2020variationalreview, bharti2022noisy}. In this context, prior studies of classical surrogates include those for; shallow hardware efficient ans\"{a}tze~\cite{basheer2022alternating, jerbi2023power, basheer2023ansatz, shaffer2023surrogate}, circuits with polynomial dynamical Lie Algebra~\cite{somma2005quantum, somma2006efficient, galitski2011quantum, goh2023lie,anschuetz2022efficient}, unstructured parameterized quantum circuits~\cite{angrisani2024classically, rudolph2023classical}, quantum convolutional neural networks~\cite{bermejo2024quantum, angrisani2024classically}, quantum kernel-based circuits ~\cite{smith2023faster}, noisy circuits~\cite{fontana2023classical, mele2024noise} and quantum machine learning~\cite{schreiber2022classical, jerbi2023shadows, sweke2023potential, landman2022classically}. However, those proposals are to surrogate the entire landscape for relatively restricted families of circuits.  Surrogating more general circuits in restricted regions, except for some largely heuristic works~\cite{koczor2020quantum, gustafson2024surrogate}, has so far not been studied. In this paper, we show that it is possible to create a classical surrogate for previously proposed reduced angle range initialization strategies for avoiding barren plateaus~\cite{zhang2022escaping,park2023hamiltonian, wang2023trainability, park2024hardware, shi2024avoiding, puig2024variational, mhiri2024constrained,puig2026warm,mhiri2025unifying,lerch2026iqp}.

We emphasize that although these surrogates superficially have the flavor of dequantizations, they do not in general imply the irrelevance of quantum computers. In the generic setting, one still requires a quantum device to generate the data from which the surrogate is constructed, except in those cases where the relevant data can themselves be obtained by efficient classical simulation.
In this sense, the patch-surrogate approach enables a training paradigm of a similar spirit to Ref.~\cite{koczor2020quantum}, in which quantum resources are devoted solely to the construction and updating of the surrogates. We investigate this setting and compare its performance with that of the standard gradient-descent approach. Although our preliminary results do not yield definitive conclusions as there exist regimes in which each method outperforms the other, this observation points to a promising new direction for exploration.  In particular, surrogates may offer significant practical benefits given the effect of noise drift and in the context queuing times for cloud-based hardware. 

Beyond variational quantum computing, our results have implications in numerous domains including quantum sensing and quantum simulation. 
For quantum sensing, one could use our surrogated landscape for inference-based  techniques~\cite{huerta2022inference,ijaz2024more}. In the case of dynamical simulation, the landscape is that of physical expectation values spanned by the system parameters and can be used to analyze hardware data or extend coherence times. Our numerical simulations, and in particular the long-time simulation of a 127-qubit transverse field Ising model on a heavy-hex topology, demonstrate our algorithms' utility beyond the regimes in which we have strict analytic guarantees.

\begin{figure*}
    \centering
    \begin{tikzpicture}
    \pgftext{\includegraphics[width=.85\textwidth,right]{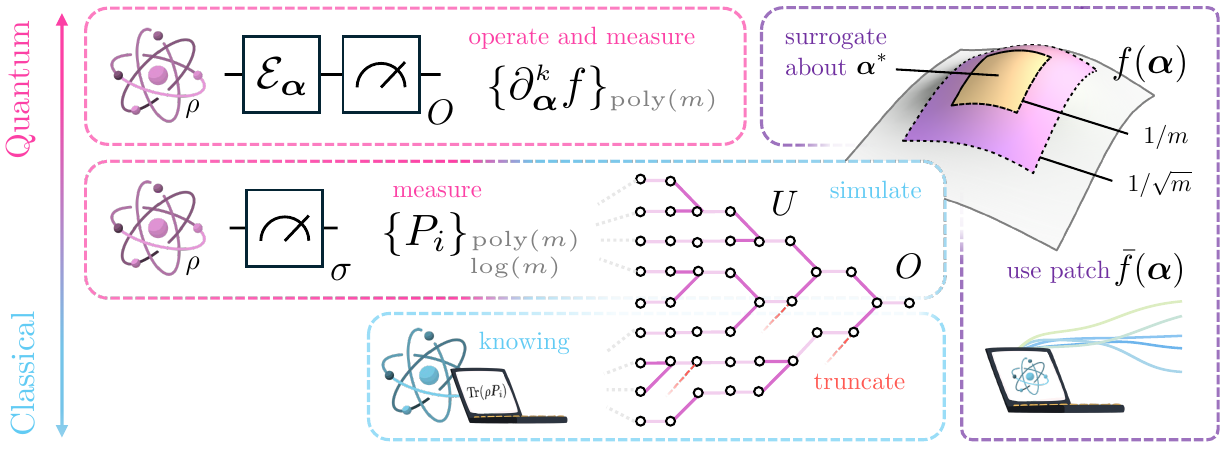}\hspace{2cm}} 
    \node at (-8.5,1.8) {$a)$ Theorem~\ref{thm:surrogate-average-general}};

    \node at (-8.5,.25) {$b1)$ Theorem~\ref{thm:paulisurrogate-average}};
    \node at (-8.5,-.25) {$b2)$ Theorem~\ref{thm:shadowsurrogate-average}};

    \node at (-8.5,-1.9) {$c)$ Theorem~\ref{thm:uncorrfull}};

    \end{tikzpicture}

    \caption{
        The process of surrogating the expectation function $f(\vec{\alpha})$ obtained as the expectation value of an observable $O$, measured over a state $\rho$ evolved under the $m$-parameter circuit $U(\vec{\alpha})$.
        Quantum measurements and/or classical simulations via truncated Pauli propagation are used to surrogate a patch $\overline{f}(\vec{\alpha})$ which can be used in lieu of further quantum operations for applications such as VQA.
    }
    \label{fig:overview_plot}
\end{figure*}

\section{Background}

We consider the task of simulating expectation values of the form 
\begin{equation}\label{eq:expectation}
    f(\vec{\alpha})=\Tr[\mathcal{E}_{\vec{\alpha}}(\rho) O]\,,
\end{equation}
where $\rho$ is an $n$-qubit input state, $O$ is a Hermitian operator, and $\mathcal{E}_{\vec{\alpha}}$ is a quantum channel parameterized by an \nparams\hphantom\negthickspace-component vector of real parameters $\vec{\alpha}$. 
By varying the choice in $\rho$, $O$ and $\mathcal{E}_{\vec{\alpha}}$ this expression can be used to capture observables and losses for a wide range of variational quantum algorithms, real-time simulations, and quantum machine learning tasks~\cite{cerezo2020variational, cirstoiu2020variational}.

In the following, we will briefly provide some background on the notions of classical simulation and classical surrogatability. We will then introduce a classical algorithm for classical-quantum surrogation. In particular, these  techniques constitute a new variant of LOWESA, introduced in Refs.~\cite{fontana2023classical,rudolph2023classical}, and more generally belongs to the family of Pauli propagation algorithms~\cite{rall2019simulation, aharonov2023polynomial, nemkov2023fourier, beguvsic2024fast, beguvsic2023simulating, fontana2023classical, shao2023simulating, rudolph2023classical, schuster2024polynomial, angrisani2024classically, gonzalez2024pauli, bermejo2024quantum, rudolph2026thermal, rudolph2025pauli}.

\subsection{Classical simulation and surrogation}

By now, it is well established that a large class of parameterized quantum circuits exhibit the so-called \textit{barren plateau} phenomenon where their parameter landscapes concentrate exponentially with system size~\cite{larocca2024review,mcclean2018barren}. On a barren plateau, the probability that any randomly chosen parameter vector $\vec{\alpha}$ corresponds to an observable expectation value which deviates significantly from its average vanishes exponentially with the number of qubits $n$ upon which the circuit acts~\cite{arrasmith2021equivalence}. This is widely believed to prohibit the training with random initializations for a broad family of variational quantum
algorithms~\cite{mcclean2018barren, ragone2023unified,fontana2023theadjoint, arrasmith2021equivalence, marrero2020entanglement,sharma2020trainability,patti2020entanglement,wang2020noise, larocca2021diagnosing, holmes2021connecting, cerezo2020cost,holmes2021barren,martin2022barren, crognaletti2024estimates, anschuetz2024unified,srimahajariyapong2025connecting, lerch2026iqp} and quantum machine learning models~\cite{rudolph2023trainability,kieferova2021quantum,tangpanitanon2020expressibility,thanaslip2021subtleties,thanasilp2022exponential, letcher2023tight, chang2024latent,xiong2023fundamental,xiong2025role,aghaei2026pitfalls}. Nevertheless, it leaves open the possibility that small sub-regions of the landscape exhibit substantial landscape variances.

These realizations have prompted recent efforts into variance lower bounds for sub-regions of quantum landscapes. Early work investigated small angle initializations for the hardware efficient variational ansatz~\cite{zhang2022escaping,shi2024avoiding,park2024hardware,cao2024exploiting, xin2024improve} and Hamiltonian variational ansatz~\cite{park2023hamiltonian, mele2022avoiding}. Subsequent work derived bounds for iterative learning schemes in the context of variational quantum simulation~\cite{puig2024variational}, as well as for ground state preparation~\cite{puig2026warm}. These approaches have recently been unified in Ref.~\cite{mhiri2025unifying} and further extended to generative modeling with a nonlinear loss in Ref.~\cite{lerch2026iqp}. Common to all these works is the observation that it is possible to guarantee landscape variances that vanish at worst polynomially in the number of trainable parameters in special sub-regions that have a radius that shrinks polynomially in the number of parameters. 

In parallel to this line of work, it has been observed that there is a strong correlation between being able to prove that a variational quantum landscape is barren plateau free (i.e,  proving that the variance of expectation values vanishes, at worst, polynomially) and the ability to create a classical surrogate for that landscape~\cite{cerezo2023does}. 
Here we demonstrate that in all the reduced parameter regimes where it is possible to provide polynomial (in $n$ or $\nparams$) lower bounds on the variance of the expectation value, it is possible to construct a classical surrogate of the expectation value landscape. 

To pin down what we mean by a classical surrogate in this context, let us start by defining what it means to simulate an expectation value. We will say that an algorithm can classically simulate an expectation value $f(\vec{\alpha})$ with high probability if a polynomial time classical algorithm can implement a function 
$\overline{f}(\vec{\alpha})$ approximating the expectation value up to error $\epsilon$, i.e., 
\begin{equation}\label{eq:error}
     \Delta f(\vec{\alpha}) := \left| \overline{f}(\vec{\alpha})- f(\vec{\alpha}) \right| \leq \epsilon\, ,
\end{equation}
with probability $1-\delta$ for random $\alv \sim\PC$ where $\PC$ is some distribution of interest and where $\epsilon \in \Theta(1)$ and $\delta \in \Theta(1)$ are arbitrarily small constants. 
We will also be interested in computing the function for \textit{any} parameter settings, i.e., the \textit{worst-case} error. Thus, we will also consider a stronger version where an algorithm can compute the expectation function if Eq.~\eqref{eq:error} holds for all $\alv $ in a given landscape region.

In contrast, preparing a \textit{classical surrogate} involves using a quantum computer for an initial (polynomial-time) data acquisition phase. Since this data can be utilized by subsequent classical simulation, the combination is more powerful than a purely classical algorithm. More concretely, we say it is possible to efficiently ``classically surrogate" a quantum expectation function $f(\vec{\alpha})$ over the distribution $\vec{\alpha}\sim\PC$ if, using data from an initial polynomial-time data acquisition phase on a quantum computer, it is subsequently possible to classically simulate expectation values over $\PC$.

Crucially, the practical benefits of classically surrogating a quantum landscape depend on the family of channels that need to be run on the quantum computer 
in the non-adaptive polynomial-time data acquisition phase. For our most general results on the surrogatability of quantum landscapes, we will allow arbitrary polynomial depth and width circuits to be run. In this case, shown in Fig.~\ref{fig:overview_plot}a), the main advantage in constructing the classical surrogate landscape is that after an initial overhead, any subsequent algorithm can be run fully classically. Such a scheme transforms a potentially adaptive hybrid quantum-classical algorithm into a non-adaptive one, opening up the possibility of using the power of classical optimization tools such as automatic
differentiation, while leaving the total quantum resources requirements  unchanged. However, we will also consider more specialized cases where one only needs to perform Pauli measurements on the  initial state. In this case, shown in Fig.~\ref{fig:overview_plot}b), the quantum resource requirements are substantially reduced as compared to running the original quantum channel to evaluate different points on the expectation landscape. This quantum resource reduction is achieved by offloading much of the work of simulating the circuit to a classical simulation algorithm.

It is possible to consider many different possible distributions $\mathcal{P}$ of parameters over sub-regions of a quantum landscape here. Here, in line with the small-angle initialization literature~\cite{puig2024variational, mhiri2024constrained, park2023hamiltonian, park2024hardware}, we will focus on uniform distributions over hypercubes of width $2r$. More concretely, let us define
\begin{equation}\label{eq:hypercube}
	\vol(\cenv, r) := \{ \vec{\alpha} \} \; \; \text{such that} \; \; \alpha_i \in [\cen_i -r , \cen_i + r ]  \, \, \forall \, \, i ,
\end{equation}
as the hypercube of parameter space centered around the point $\cenv$, and 
\begin{equation}\label{eq:uniformdist}
    \uni(\cenv, r) := \text{Unif}[\vol(\cenv, r)]
\end{equation}
as a uniform distribution over the hypercube $\vol(\cenv, r)$. We expect that our conclusions straightforwardly generalize to other reduced parameter range distributions including Gaussian distributions~\cite{zhang2022escaping, shi2024avoiding} or to allow for independent perturbation radii for each parameter.

\subsection{Pauli Propagation}\label{sec:PP}

There exist numerous classical techniques to simulate quantum circuits and evaluate the expectation function in Eq.~\eqref{eq:expectation}.
In this work, we employ Pauli propagation methods~\cite{nemkov2023fourier, beguvsic2024fast, beguvsic2023simulating, fontana2023classical, rudolph2023classical, angrisani2024classically,bermejo2024quantum, rudolph2025pauli, miller2025simulation, angrisani2025simulating}. 
In this context, it is most natural (and most efficient) to consider quantum circuits written as parameterized perturbations from Clifford circuits. That is, unitaries of the form
\begin{equation}\label{eq:CliffordVQAcircuit}
	U_{\rm CP}\left( \alv \right) = \prod_{i=1}^\nparams C_i U_i(\alpha_i)\,,
\end{equation}
where $\left\{ C_i \right\}_{i=1}^\nparams$ are a set of fixed but arbitrary Clifford unitaries and $\left\{ U_i(\alpha_i) = e^{-i\alpha_i P_i/2} \right\}_{i=1}^\nparams$ are parameterized rotations generated by Pauli operators $\{ P_i \}_{i=1}^\nparams$.
This architecture can be universal even when $U_i$ are single-qubit phase rotations, though we allow $U_i$ to be Pauli rotations acting on any number of qubits. In our analytics we will consider both the case where the parameters $\alpha_i$ are uncorrelated and cases where they are correlated. 
We further emphasise however that this particular circuit form is inessential to our most general complexity bounds in Section~\ref{sec:general_case} which can, for example, also quantify expectation landscapes for parameterized channels, and for circuit which are not near-Clifford. 

Pauli propagation methods typically work in the Heisenberg picture where the initial operator (here the observable $O$) is often sparse in the Pauli basis. In this basis, the operator is back-propagated through the circuit and finally overlapped with the initial state $\rho$. In particular, since the adjoint action of any Clifford operation on any Pauli observable $P_j$ simply produces another Pauli $P_j^\prime$, 
\begin{equation}
    C_i^\dagger \, P_j \, C_i = s_{ij} P_j^\prime, \;\;\;\; s_{ij} \in \{-1,1\},
\end{equation}
which we will notate as $P_j \xrightarrow[]{C_i} \pm P_j^\prime$ .
These transformations can either be computed in $\OC(n^2)$ time using the Gottesman-Knill theorem~\cite{aaronson2004improved} or in $\OC(1)$ time for any \textit{a priori} known $C_i$ acting on few qubits by use of look-up tables.
Applying a non-Clifford Pauli rotation $e^{-i\alpha_i P_i/2}$ to $P_j$ leaves the operator invariant if $[P_i,P_j]=0$, but otherwise produces a real-valued weighted sum of two Pauli operators,
\begin{equation}
    P_j \xrightarrow[]{e^{-i\alpha_i P_i/2}} \cos(\alpha_i)P_j +\sin(\alpha_i) P_j^{\prime\prime}\, ,
\end{equation}
where $P_j^{\prime\prime} = i[P_i, P_j]/2$ is a Pauli operator. We note that the sign of the sine term is flipped in the Schrödinger picture. By iterating this procedure, one ends up with a collection of Pauli operators and corresponding trigonometric prefactors.

After Pauli propagation, the expectation function in Eq.~\eqref{eq:expectation} takes the form
\begin{equation}\label{eq:path_expectation}
    f\left(\alv\right) = \sum_{P} c_{P}\left(\alv\right) \Tr[\rho P] \,,
\end{equation}
where $c_{P}$ are the coefficients of the back-propagated Pauli operators $P$. These coefficients capture both the initial weight of each of the relevant Paulis in the target observable $O$ and the sine and cosine coefficients that have been picked up during the propagation. 

One approach, embracing the concept of expectation function surrogates~\cite{fontana2023classical,rudolph2023classical,bermejo2024quantum}, is not to compute the coefficients $c_{P}$ of the Pauli paths from scratch every time the parameters $\bm\alpha$ are changed. Instead, one can symbolically or graphically represent the coefficients for each final Pauli $P$. When it becomes clear which $P$ contribute significantly to the expectation in Eq.~\eqref{eq:path_expectation}, only their corresponding coefficients $c_{P}$ can be re-evaluated. The approach has a substantial initial time and especially memory overhead, but once performed, enables rapid re-evaluation of the expectation value $f(\alv)$ for any $\alv$. The implementation used in this work can be seen as an improved variant of the original proposal of LOWESA. We give a clarifying illustration of the technique in Fig.~\ref{fig:example_pp_qft}. 

In its raw form, the number of Pauli operators in the backpropagated observable (what we call the ``\textit{Pauli paths}"), or equivalently the number of terms in the sum in Eq.~\eqref{eq:path_expectation}, grows exponentially with the number $m$ of Pauli rotations in the circuit. To make the algorithm efficient for larger $m$ values, one must therefore apply truncation schemes. Truncating higher weight Pauli operators has been shown to be efficient for both noisy circuits~\cite{aharonov2023polynomial, fontana2023classical, schuster2024polynomial} and random circuits~\cite{angrisani2024classically}. Here we will provide guarantees for a \textit{small-angle} truncation scheme for patches of expectation landscapes corresponding to perturbations away from a Clifford circuit. We note that this approach is related to the techniques numerically demonstrated in Ref.~\cite{beguvsic2023simulating}.

\begin{figure}
    \centering
    \begin{tikzpicture}
    \pgftext{\includegraphics[width=\columnwidth]{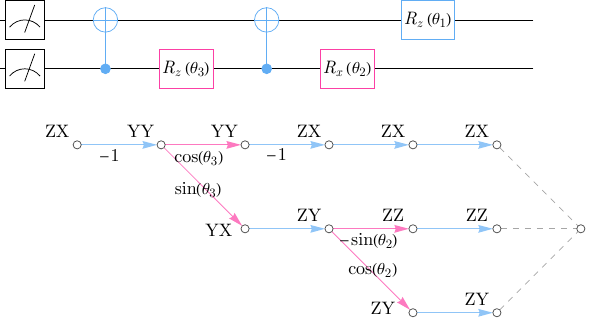}}
    \node at (3.8,-.35) {$\Tr(\rho \, \circ)$};
    \draw[draw=white,fill=white] (3,1.3) rectangle ++ (.5,1);
    \node at (3.5,1.7) {$\rho$};
    \end{tikzpicture}
    \caption{
        A visualisation of Pauli propagation computing the expectation value of observable $O = ZX$ resulting from the pictured two-qubit circuit upon input state $\rho$. Evaluation proceeds left-to-right, simulating the circuit in reverse; graph edges indicate the Pauli strings produced by the gate aligned above them, and their labels indicate the induced coefficient; elements of the gate's Pauli transfer matrix. 
        The Clifford CNOT gates merely multiply Pauli strings by $\pm 1$ (blue edges), while the non-Clifford rotations $R_P(\theta) = e^{-i \theta P/2 }$ can sometimes produce \textit{two} Pauli strings (pink edges). Notice the state $\rho$ need not be explicitly instantiated; instead, only its trace with respect to the final Pauli strings are evaluated.
    }
    \label{fig:example_pp_qft}
\end{figure}

\section{General patches}\label{sec:general_case}

\subsection{Error guarantees}

Our first result is a general bound which shows that it is always possible to create a classical surrogate for small sub-regions of a quantum landscape generated by a general parametrized quantum channel. Concretely, we show that for any function that depends on ${\nparams}$ independent parameters, it is possible to create a surrogate of a patch of width $\sim 1/\sqrt{\nparams}$, provided that some assumptions on the magnitude of the higher order derivatives are satisfied. This result is captured by Theorem~\ref{thm:surrogate-average-general} below, which we prove and state more formally in Appendix~\ref{app:generalsurrogatetheorem} (Supplemental Theorem~\ref{coro:surrogate_quantum_quantum}). 

\begin{theorem}[General Surrogation Guarantee, Informal]\label{thm:surrogate-average-general}
Consider an expectation function $f(\vec{\alpha})$ of the form in Eq.~\eqref{eq:expectation}, where the measurement operator satisfies  $\|O\|_{\infty} \in \OC(1)$, and where the parameters of the quantum channel are taken from a hypercube region $\vol(\alv^*,r)$ around an arbitrary point $\alv^*$ of width $2r$.
Assume that the partial derivatives within this region scale at most exponentially in the derivative order $k$, i.e.,
\begin{align}
    |\partial_{i_1, i_2, ... i_k}\LC(\alv)| \leq  \gamma^k\norm{O}_\infty \;\; \;,\;\; \forall \alv \in \vol(\alv^*, r)\;,
\end{align}
for some constant $\gamma \in \mathcal{O}(1)$. Furthermore, assume that the partial derivatives for $k \in \mathcal{O}(1)$ can be estimated efficiently on a quantum computer. 
Then, it is possible to efficiently classically surrogate $f(\vec{\alpha})$ up to an arbitrary small error $\epsilon \in \Theta(1)$ with high probability for random $\alv\sim\uni(\vec{\alpha}^*, r)$ for  
\begin{equation}
    r\in\OC\left( \frac{1}{\sqrt{\nparams}}\right) \, ,
\end{equation}
and for all $\alv$ in $\vol(\vec{\alpha}^*, r)$ for 
\begin{equation}
    r\in\OC\left( \frac{1}{\nparams}\right) \, ,
\end{equation}
where in both cases there is an arbitrarily small but constant failure probability from shot noise in the data collection phase. 

\end{theorem}

The proof of Theorem~\ref{thm:surrogate-average-general} relies on expressing the expectation value $f(\alv)$ with Taylor's expansion around the point of interest $\alv^*$ and then showing that with some mild assumptions regarding the partial derivatives (see Assumptions~\ref{assumption:bounded-derivatives}-\ref{assumption:loss-evaluations-expectations}  in Appendix~\ref{app:generalsurrogatetheorem}), a low-order truncation of the expansion gives a sufficiently good surrogate for a small region\footnote{ We note that surrogation beyond $r \in \mathcal{O}(1/\sqrt{m})$ is also achievable with low error, by extending the truncation order $\kappa$ to be polynomial. This, in turn, increases the computational resources required to be super-polynomial or even exponential with $m$. Detailed derivations for various truncations are provided in Appendix~\ref{sec:extension-larger-patch}. Furthermore, these techniques could be adapted to other geometries. Although a rigorous generalization requires more technical work, Fig.~\ref{fig:training-distance-patch} offers evidence that the results of Theorem~\ref{thm:surrogate-average-general} remain valid.}. These assumptions are in fact met by most parameterized quantum circuits considered currently by the community (see Appendix~\ref{sec:guarantee-assumptions-theorem1}) and could also be applied to parameterized channels. In particular, it includes not just the uncorrelated Pauli rotation ans\"{a}tze in Eq.~\eqref{eq:CliffordVQAcircuit} but also arbitrarily deep circuits where the ${\nparams}$ independent parameters can appear in multiple gates. To name just one example, the theorem also holds for gates using Givens Rotations, which are crucial elementary building blocks for quantum chemistry~\cite{arrazola2022universal}.

It is noteworthy that our average case surrogacy guarantee aligns with the growing body of prior work on absence of barren plateau proofs for patches of a quantum expectation function landscape that scale as $1/\sqrt{\nparams}$. In particular, it has been shown in Refs.~\cite{puig2024variational, mhiri2024constrained} that a region with $r \in \OC( 1/\sqrt{\nparams})$ around some minimum $\alv^*$ has substantial gradients
\begin{align}
\Var_{\alv\sim\uni(\vec{\alpha}^*, r)}\left[ f(\alv)\right] \in \Omega\left( \frac{1}{\poly(n)}\right) \; .
\end{align}
Theorem~\ref{thm:surrogate-average-general} hence implies that, in all of the sub-regions where it is possible to prove absence of barren plateaus, it is not necessary to run a variational quantum algorithm adaptively on a quantum computer but rather one could first quantumly prepare a surrogate which is then optimised fully classically. 

Beyond the average case surrogatability guarantee, we additionally prove the worst-case error incurred by our algorithm in a smaller patch of width $\OC( 1/m)$. Inside this smaller patch, there exists no parameter vector $\bm\alpha$ which is classically simulated with an error worse than the constant, but arbitrarily small, chosen error. As discussed further in the Appendices below  Theorem~\ref{thm:time-complexity-worst-PP-classical}, this implies that in the context of dynamical simulation Trotter errors can be reduced with only a polynomial overhead.

Theorem~\ref{thm:surrogate-average-general} leaves open the nature of the measurements required on the quantum computer, and in general these might require running highly non-trivial circuits. Namely, one will generally need to run circuits to evaluate the loss in order to compute gradients via the generalized parameter-shift rule~\cite{wierichs2022general}.
Concretely, for general circuits where the patch is expanded around an arbitrary point $\bm\alpha^*$, the quantum device can be used to estimate the function $f(\vec{\alpha})$ at different choices in $\vec{\alpha}$ and the surrogate is generated by reconstructing a multi-variate Taylor expansion of $f(\vec{\alpha})$. In this case, the surrogate bears similarities to applying higher-order gradient optimization methods because generating the surrogate requires effectively measuring higher order derivatives of the loss. In this sense the distinction between running a variational quantum algorithm explicitly (by say updating by one step per iteration using a higher order gradient descent method) and generating the surrogate (by first measuring higher order gradients and then training classically for a few steps) can become blurred. 

Nonetheless, Theorem~\ref{thm:surrogate-average-general} suggests
it may be worthwhile to forego experimentally re-evaluating a quantum circuit with varying parameters \textit{at each step}, and instead spend the quantum resources preparing a surrogate. The constructed surrogate can then be trained entirely classically for a few iterations within the patch. Nevertheless, for a whole training process to remain faithful, the surrogate will need to be reconstructed once the training trajectory leaves the patch.

In Appendix~\ref{sec:iterative-training-taylor-surrogate}, we present a preliminary exploration of this aspect. Specifically, we investigate the performance of an iterative surrogate training strategy and compare its quantum resource requirements—namely, the number of measurement shots—with those of standard gradient descent. We emphasize that, in this case, the surrogate is constructed directly at the level of the loss gradient, in contrast to the approaches discussed in the main text, where the loss function itself is surrogated. In addition, a shot-noise analysis of this gradient-based surrogate, including a practical approximation of the optimal shot-allocation strategy for reducing shot-noise, is provided in Appendix~\ref{sec:analytical-guarantee-second-order}.
In the next subsection we will explore one key component of this approach -- namely, we investigate numerically how far it is possible to train in practice using the data obtained to build a surrogate.

\subsection{Numerical tests}

In this section we numerically investigate the maximum distance we can move in the parameter space before the surrogate becomes unusable. That is, the maximum distance such that training on the surrogate continues to decrease the true loss.

For concreteness, we consider a surrogate of an expectation loss function of the form of Eq.~\eqref{eq:expectation}, where the observable is a Heisenberg Hamiltonian
\begin{equation}\label{eq:heisenberg_obs}
    O = \frac{1}{3(n-1)} \sum_{i=1}^{n-1} (X_i X_{i+1} + Y_i Y_{i+1} + Z_i Z_{i+1}).
\end{equation}
The initial state is the zero state, $\rho=\ketbra{0}^{\otimes n}$, and $\EC_{\alv}(\rho)$ is a variational ansatz with $m=3nL$ parameters composed of \( L \)-layers of single-qubit rotations followed by two-qubit \( ZZ \) gates of the form
\begin{equation}\label{eq:ansatz_taylor}
    U(\alv) = \prod_{l=1}^L \left( \prod_{i=1}^n R_{ZZ}(\alpha_{i+2n,l}) \prod_{i=1}^n R_Z(\alpha_{i+n,l}) R_X(\alpha_{i,l}) \right)\,.
\end{equation}
We compute the surrogate of the gradient $\nabla\tilde{f}(\alv)$ around the randomly chosen initial point $\alv^*$ and use it to train the algorithm until $f(\alv)$, the test loss, stops decreasing. Let $\alv_0$ denote the corresponding parameters at that step. Finally, we compute $\|\alv^* - \alv_0\|_{x}$, where $x\in\{\infty, 2,1\}$, with respect to the number of parameters in the circuit $m$.

The results are presented in Fig.~\ref{fig:training-distance-patch}, where the diamonds are the numerical values of $\|\alv^* - \alv_0\|_{x}$ for $x = 1,2, \infty$  after averaging for 20 runs with different initial points (with the maximum and minimum shaded), the solid lines are a fit of those points and the dashed lines represent the scalings we would expect from the bound. Namely, we recall that according to Theorem~\ref{thm:surrogate-average-general}, $\|\alv^* - \alv_0\|_{\infty}\in\order{m^{-0.5}}$, $\|\alv^* - \alv_0\|_{2}\in\order{1}$, $\|\alv^* - \alv_0\|_{1}\in\order{m^{0.5}}$, compared to the $\|\alv^* - \alv_0\|_{\infty}\in\order{m^{-0.15}}$, $\|\alv^* - \alv_0\|_{2}\in\order{m^{0.2}}$, $\|\alv^* - \alv_0\|_{1}\in\order{m^{0.7}}$, found in the experiment. We thus see that the experimental results have, roughly, between a $m^{0.2}$ and $m^{0.3}$ better scaling (i.e., the solid lines are angled above the dashed lines). This implies that the optimization process can use the surrogate for significantly longer distances than our bounds predict.

\begin{figure}
    \centering
    \includegraphics[width=0.99\linewidth]{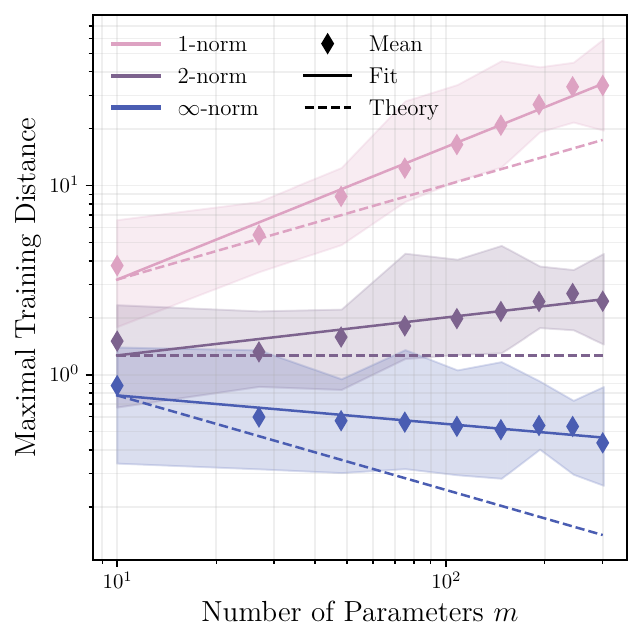}
    \caption{\textbf{Maximal training distance as function of the number of parameters.} We consider a loss of the form of Eq.~\eqref{eq:expectation}, where the observable is a Heisenberg Observable described in Eq.~\eqref{eq:heisenberg_obs}, the initial state is $\rho = \ketbra{0}^{\otimes n}$ and $\EC_{\alv}(\cdot)$ is a variational ansatz with $m = 3nL$ parameters described in Eq.~\eqref{eq:ansatz_taylor}. We compute the surrogate of the gradient around a randomly chosen initial point $\alv^*$ and use it to train the algorithm until the true loss $f(\alv)$ stops decreasing and starts to increase, at $\alv_0$. We compute the distance $\|\alv^*-\alv_0\|_{x}$, where $x\in\{\infty,2,1\}$, and plot it with respect to the number of parameters. We repeat the process 20 times and show the average in diamonds (with the maximum and minimum shaded), plot a fit as a solid line, and show the theoretical scalings according to Theorem~\ref{thm:surrogate-average-general} as a reference. The fitted lines correspond to the following functions $\|\alv^*-\alv_0\|_\infty =1.10 m^{-0.15} $, $\|\alv^*-\alv_0\|_2 = 0.80m^{0.20}$,$\|\alv^*-\alv_0\|_1 = 0.63m^{0.70}$. 
    }
    \label{fig:training-distance-patch}
\end{figure}

In Appendix~\ref{sec:iterative-training-taylor-surrogate}, we further provide a preliminary comparison between the quantum resources required for an iterative surrogate approach as compared to the standard hybrid quantum classical approach. As there are many different hyper-parameters for this benchmarking problem, it is hard to fairly draw definitive conclusions. In particular, the standard gradient-descent approach is operationally simpler, as it requires tuning fewer hyper-parameters than the iterative surrogate strategy, which additionally introduces choices such as the frequency of the surrogate reconstruction. Moreover, gradients estimated directly on quantum hardware are unbiased, rendering the method more resilient to shot noise. Nevertheless, our numerics suggest that the bias introduced by the Taylor approximation can, in some regimes, be balanced by the increased statistical noise associated with repeated gradient estimation in the standard approach, although we do not claim this as a general result. 

As such, we present both cases where the surrogate outperforms a standard variational quantum algorithm and vice versa. We therefore suggest that this Taylor surrogate should be viewed as a potential alternative to vanilla VQAs, rather than a replacement, with potential practical advantages given the problems associated with noise drift and/or queuing for cloud access.

\section{Near-Clifford patches}\label{sec:CliffordPert}

We will now specialise to particular families of problems for which the requirements of the quantum computer can be substantially reduced as compared to the Taylor surrogate approach. In particular, in this section, we consider a patch of a quantum landscape that is centered around a Clifford circuit.
We utilize the Pauli propagation approach described in Section~\ref{sec:PP} with a truncation method specialized for this near-Clifford region.

\subsection{Error guarantees}

Concretely, we consider the circuit composed of alternating layers of Clifford gates and uncorrelated Pauli rotations defined in Eq.~\eqref{eq:CliffordVQAcircuit}. We then consider a hypercube of width $2r$ around the all zero vector $\vec{0}$. In this region, the angles $\alpha_i$ are small and so we can leverage the fact that sine contributions to the coefficients $c_{P}$ in Eq.~\eqref{eq:path_expectation} are much smaller than cosine contributions. We use this observation to prune Pauli paths with many sine coefficients, similar to the approach in Ref.~\cite{beguvsic2023simulating}. Then, we show that this `small-angle' truncation strategy achieves an exponentially decreasing average error as a function of the sine expansion order $\expansionorder$ (i.e., the number of sine terms kept).

\subsubsection{Classical time complexity}
\label{sec:clifford_perturb_fully_classical}

We begin by presenting a theorem (derived in Appendix~\ref{apx:runtim}), which studies the time complexity of the classical simulation algorithm.

\setcounter{theorem}{1}
\begin{theorem}[Time complexity of small-angle Pauli propagation, Informal]\label{thm:uncorrfull}
Consider an expectation function $f(\vec{\alpha})$ of the form of Eq.~\eqref{eq:expectation} with an observable $O= \sum_{P\in\mathcal{P}_n} a_P P$, where at most  $N_{\rm Paulis} \in \mathcal{O}(\poly(n))$ coefficients $a_P$ are non zero, with $\norm{\vec{a}}_1 \in \OC(1)$, and a parametrized circuit of the form of Eq.~\eqref{eq:CliffordVQAcircuit} with $\nparams$ independent parameters. If the expectation values $\Tr[\rho P]$ of the initial state $\rho$ with any Pauli operator $P$ can be efficiently computed, a runtime of 
\begin{equation}
    t \in \OC\left(N_{\mathrm{Paulis}} \cdot \nparams^{\log{\left(\frac{1}{\epsilon^2\delta}\right)}}\right)\;,
\end{equation}
suffices to classically simulate $f(\alv)$ up to an error $\epsilon$ with probability at least $1-\delta$ for random $\alv\sim\uni(\vec{0}, r)$ with 
\begin{equation}
    r \in \OC\left(  \frac{1}{\sqrt{\nparams}} \right) \, .
\end{equation}
Similarly, with a runtime of
\begin{equation}
    t \in \OC\left(N_{\mathrm{Paulis}} \cdot \nparams^{\log{\left(\frac{1}{\epsilon}\right)}}\right)\;.
\end{equation}
we can classically simulate $f(\alv)$ up to an error $\epsilon$
for all $\vec{\alpha}$ in the hypercube $\vol(\vec{0}, r)$ with 
\begin{equation}
    r \in \OC\left(  \frac{1}{\nparams} \right) \, 
\end{equation}
\end{theorem}
A more formal version of the theorem is stated in Supplemental Theorems~\ref{thm:time-complexity-PP-classical-avg} and~\ref{thm:time-complexity-worst-PP-classical}.

Thus we see that Pauli propagation can be used to efficiently simulate Pauli expectation values for arbitrary circuits of the form Eq.~\eqref{eq:CliffordVQAcircuit} with high probability for all $\bm\alpha$ in a hypercube of width $\OC(1/\sqrt{\nparams})$ and \textit{for all} $\bm\alpha$ in a smaller hypercube of width $\OC( 1/\nparams)$. Notice that these results mirror the general result in Theorem~\ref{thm:surrogate-average-general}, but for a family of circuits that is itself classically tractable. It follows that as long as we have a classical representation of the initial state $\rho$ (analytically or, e.g., as a tensor network) from which we can efficiently compute expectation values of Paulis, then it is possible to classically simulate VQAs with high probability in the small angle regimes for which trainability guarantees have been derived~\cite{zhang2022escaping,park2023hamiltonian, shi2024avoiding,park2024hardware}. Moreover, our numerics in Section~\ref{sec:exact_numerics} suggest that the above bounds are somewhat loose and improved resource scalings can be obtained in practice. 

In Appendix~\ref{apx:average_correlated}, we further provide a bound for the average error of an observable with a circuit where all angles are correlated (i.e., where all the rotation gates are parametrized by the same $\alpha$). Intriguingly, we find that the average-case simulation of such maximally correlated angles is asymptotically as hard as the worst-case as a function of ${\nparams}$.
It is not clear whether the equal scaling is mostly due to our proof techniques, or whether simulating cases such as real-time dynamics of quantum systems is about as hard as the worst-case.

An interesting corollary of our worst-case analysis is that we can provably efficiently simulate quantum dynamics of Hamiltonians with $\OC(n)$ Pauli terms and $\OC(1)$-sized coefficients up to time $t = \OC\left(\frac{1}{n}\right)$ on any geometry. This is because, for $L$ Trotter layers with time step $dt$, we have $m = \Theta(nL)$ parameters and $dt = \OC\left(\frac{1}{m}\right)$ is the largest guaranteed time step. We thus have $dt = \OC\left(\frac{1}{nL}\right)$, and with $t = Ldt$ we find $t = \OC\left(\frac{1}{n}\right)$ is guaranteed to be efficient.

Given our significantly more favorable average-case errors for uncorrelated angles  it stands to reason that any independence of parameters interpolates between our two complexity bounds. In Section~\ref{sec:sexynumerics} we present a large-scale (127 qubits) numerical demonstration of Pauli propagation for a quantum simulation problem with a \textit{time-dependent} Hamiltonian which supports this hypothesis.

\subsubsection{Sample complexity}
\label{sec:clifford_perturb_quantum_enhanced}

In the previous section, we assumed that we could efficiently estimate  the overlap terms $\Tr[\rho P]$ in Eq.~\eqref{eq:path_expectation}. If these can be efficiently computed classically for any Pauli operator $P$, then the time complexity computed in the previous section fully quantifies the complexity of our small-angle patch simulation. If, instead, we do not have a classical representation of the state $\rho$,  then we need a quantum device to measure the expectation value $\Tr[\rho P]$. Here we quantify the number of measurements required, i.e., the sample complexity, in the case where $\rho$ is non-classically simulable. As we will see, this sampling step is efficient in the sense that it requires only a polynomial number of Pauli measurements on quantum hardware. 

There are several different tomographic methods that could be used to measure the necessary $\Tr[\rho P]$. The simplest is to take the total shot budget $N_s$ and divide them equally between each of the Pauli terms that need to be measured. In this case, analogously to Theorem~\ref{thm:uncorrfull}, the number of terms that need to be measured scales polynomially with the number of parameters for a hypercube region with $r \in \OC(1/\sqrt{\nparams})$. More concretely, the following theorem holds. 

\begin{theorem}[Surrogation via Direct Pauli Measurements, Informal]\label{thm:paulisurrogate-average}
Consider an expectation $f(\vec{\alpha})$ of the form of Eq.~\eqref{eq:expectation} with an observable $O= \sum_{P\in\mathcal{P}_n} a_P P$ containing at most $N_{\rm Paulis} \in \mathcal{O}(\poly(n))$ Pauli terms with $\norm{\vec{a}}_1 \in \OC(1)$ and a parametrized circuit of the form of Eq.~\eqref{eq:CliffordVQAcircuit} with $\nparams$ independent parameters. A shot budget of 
\begin{equation}
   N_s \in \OC\left(N_{\rm Paulis}\; m^{\log\left(\frac{1}{\epsilon^2\delta}\right)}\right) 
\end{equation}
suffices to classically surrogate $f(\alv)$ with an error $\Delta f(\alv)$ of at most $\epsilon$ with probability at least $1-\delta$ for any random $\alv\sim\uni(\vec{0}, r)$.

\end{theorem}
A more formal version of the theorem together with the proof is provided in Appendix~\ref{sec:even-allocation-PP-surogate-guarantee-quantum} in the form of \supt~\ref{thm:avg-guarantee-PP-uncorr-quantum-even-allocation}. 

While the above strategy is simple, in general using the same number of shots to measure each of the Pauli terms is suboptimal. Indeed, if  some of the Paulis have larger coefficients than others (on average over the patch of interest) it is advantageous to use more shots to measure the terms that contribute more to the final expectation value (see warm-up example in Appendix~\ref{sec:warm-up} comparing different shot allocation strategies). In fact, it is precisely the exponentially decreasing contribution per sine-expansion order that we leverage for our efficient simulation guarantees. Hence, one should allocate the shots to each term in proportion to the average coefficient of that Pauli over the patch. In this case, the sample complexity becomes effectively independent of $N_{Paulis}$ and instead scales with the 1-norm of the averaged coefficients. That is, we show that the number of shots scales as 
\begin{equation}
    N_s \in \mathcal{O}(\norm{ \vec{c} }_{1,\mathrm{avg}}^2)\,,
\end{equation}
where we call $\norm{\boldsymbol{c}}_{1,\mathrm{avg}}:= \sum_{P\in\mathcal{P}_n}\sqrt{\mathbb{E}_{\boldsymbol{\alpha} \sim \mathcal{D}} [c_P(\alv)^2]}$ an ``effective'' 1-norm. We describe this quantity in more detail in Appendix~\ref{app:estimating-observables-parametrized} and Appendix~\ref{sec:truncated-propagated-observable}. The surrogation guarantees with the effective $1$-norm allocation strategy are further provided in Appendix~\ref{sec:uncorrelated-parameters-PP-surrogate-guarantee-quantum} (average-case) and Appendix~\ref{sec:worst-case-PP-surrogate-guarantee-quantum} (worst-case).

We note that a similar strategy can be employed if one is not interested in surrogating a whole patch but rather evaluating a single expectation value for a given $\alv_0$ without running $U(\alv_0)$ on the quantum computer. This is often true in the context of dynamical simulation, where one seeks $f(\alv)$ at only a handful of parameter values. In this case, by allocating the total shot budget to Pauli $P$ in proportion $|a_P|$ one obtains a sample complexity that scales as $N_s \propto \norm{ \vec{a} }_{1}^2$. 

\medskip

Another popular family of tomographic procedures are \textit{classical shadows}~\cite{huang2020predicting,elben2022randomized,west2024real}. Their variant for Pauli expectation values, which we here call Pauli classical shadows, enjoy tight accuracy guarantees but require more restricted families of circuits for optimal efficiency. Here we suppose that the Clifford layers consist exclusively of single-qubit Clifford gates and the Pauli rotation layers are $\OC(1)-$local. Now, each of the back-propagated Pauli operators have Pauli weight at most $\OC(\expansionorder)$,  where $\expansionorder$ is the sine-truncation order. Thus, for $\expansionorder \in \OC(\log(n))$, i.e., the regime in which small-angle Pauli Propagation enjoys polynomial run times, we can employ classical shadows.
This is captured by the following theorem. 

\begin{theorem}[Shadow Surrogation Guarantee, Informal]\label{thm:shadowsurrogate-average}
Consider an expectation function $f(\vec{\alpha})$ of the form of Eq.~\eqref{eq:expectation}, an observable $O= \sum_{P\in\mathcal{P}_n} a_P P$ composed of at most $N_{\rm Paulis} \in \mathcal{O}(\text{poly}(n))$ Pauli terms with $\norm{\vec{a}}_1 \in \OC(1)$, and a parameterized quantum circuit of the form of Eq.~\eqref{eq:CliffordVQAcircuit} with the additional constraints that $C_i$ are single-qubit Clifford operations and the Pauli rotations $U_i(\alpha_i)$ act on $\OC(1)$ qubits. A shot budget of 
\begin{align}
    N_s \in     \OC\left(\frac{\min\left( \log(n), \log(m \,N_{\mathrm{Paulis}})\right)}{\epsilon}\right),
\end{align}
random Pauli measurements suffices to achieve a mean square error at most $\epsilon$ with high probability over a uniformly sampled hypercube $\uni(\vec{0}, r)$ with $r \in \OC \left( \frac{1}{\sqrt{\nparams}}\right) $.
\end{theorem}
A more formal version of this theorem is provided in Appendix~\ref{sec:uncorrelated-parameters-PP-surrogate-guarantee-quantum}\ as \supt~\ref{thm:shadow-surrogation-guarantee-formal}. 
This tomographic procedure is well-studied and follows a ``measure first, ask questions later'' philosophy. Thus, one might wonder how this shot allocation strategy fairs against the seemingly more informed effective 1-norm allocation described above. The general scalings in Appendix~\ref{sec:uncorrelated-parameters-PP-surrogate-guarantee-quantum} do not allow for a clear comparison between the approaches, but we compare them numerically in Fig.~\ref{fig:error-bounds}. While circuits with multi-qubit Clifford gates would likely make Pauli classical shadows underperform, there is a significant upside in having collected a more general set of measurements that can be used to estimate other quantities with a similar error.

\begin{figure*}
    \centering
    \includegraphics[width=0.98\linewidth]{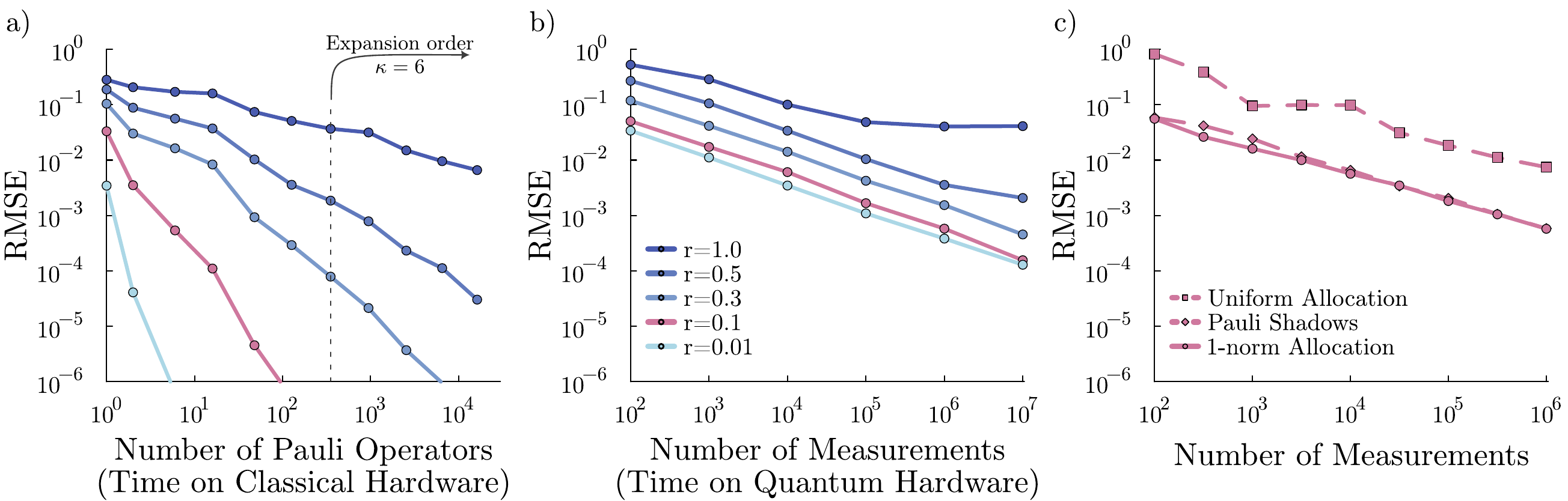}
    \caption{\textbf{Landscape and surrogation root mean square error (RMSE) as a function of classical and quantum resources.} We consider a 16-qubit Hamiltonian variational ansatz of a transverse field Ising (TFI) model on a 4x4 grid and a 1-local Pauli $Z$ expectation value near the middle of the grid. The initial state $\rho$ is the all-zero state evolved exactly by 4 Trotter layers of the TFI Hamiltonian at $dt=0.1$, and we surrogate the final 4 layers with free parameters. a) Truncation RMSE of our small-angle Pauli propagation surrogate for different $r$ as a function of the number of backpropagated Pauli operators. This number corresponds to increasing sine-expansion order from $\expansionorder=0$ to $\expansionorder=10$. b) Total RMSE for expectation values using measurements of the initial state $\rho$ for different $r$ as a function of the number of quantum measurements at $\expansionorder=6$. We use the effective 1-norm of the surrogate's coefficients described to allocate the measurements on $\rho$. c) Comparison of different measurement allocation approaches at $\expansionorder=6$ and $r=0.1$. The uniform allocation over the operators expectedly underperforms and Pauli classical shadows are in this case on-par with our effective 1-norm allocation. }
    \label{fig:error-bounds}
\end{figure*}

\subsection{Numerical resource and error scalings}\label{sec:exact_numerics}

In this section, we present numerical experiments that indicate that both the time complexity and error of Pauli Propagation methods equipped with small-angle truncation outperforms the guarantees provided by our bounds.

First we test the error incurred by the truncated Pauli propagation together and the shot noise in an exactly verifiable system. We focus on the case of the Hamiltonian Variational Ansatz (HVA)~\cite{wecker2015progress} for which reduced parameter range trainability guarantees were derived in Ref.~\cite{park2023hamiltonian}. This HVA family of circuits has the structure of the Trotterized time evolution under the Hamiltonian, but with mostly, or fully, free rotation angles. Here we choose a transverse field Ising (TFI) Hamiltonian 
\begin{equation}\label{eq:timeindependent-TFI}
    H_{\text{TFI}} = H_X + H_{ZZ}
\end{equation}
with
\begin{align}\label{eq:HX_HZZ}
    H_X = - \sum_{i}h_iX_i \,, \,\,\, 
    H_{ZZ} = - \sum_{\langle i,j\rangle}J_{ij}Z_iZ_j\,,
\end{align}
where $X, Z$ are Pauli operators, and $h_i$, $J_{ij}$ coefficients. The topology of the $ZZ$-interaction terms is that of nearest neighbors on a 4x4 grid in a 16-qubit system. 

To demonstrate the freedom in our approach, we choose the initial state $\rho$ as the exactly time-evolved all-zero state for 4 Trotter layers with $h_j = J_{ij} = 1$ and $dt=0.1$, and we surrogate the final 4 layers with the same HVA circuit ansatz but with independent parameters. The observable is a $Z$ observable in the middle of the grid. 

In Fig.~\ref{fig:error-bounds}a), we show the root mean square error (RMSE) of the truncated simulation as a function of the number of Pauli operators computed at increasing sine-truncation order $\expansionorder$, and for different angle ranges $r$. For example, the RMSE for $r=0.1$ goes below $10^{-6}$ at expansion order $\expansionorder=6$, which computes around $10^2$ Pauli operators. As a reference, the corresponding classical surrogation procedure takes a few milliseconds on a standard laptop with expectation evaluations being significantly faster. 
If the initial state is classically simulable this time complexity effectively quantifies the total simulation cost. If not, then we additionally need to account for the effect of shot noise. 

Moving onto the quantum resource requirements for non-classically simulable initial states, in Fig.~\ref{fig:error-bounds}b), we record the total error of the expectation value as a function of the number of measurements of $\rho$. We allocate the measurements according to the effective 1-norm allocation described before. The numerical results seem to indicate that the total error for $r\leq 0.5$ is dominated by shot noise even with $10^7$ measurements, as those uncertainties are orders of magnitudes larger than the classical approximation error for small $r$. This argument can be exemplified by the extreme case of $r=0$, where a single operator perfectly reconstructs the ``landscape''. Yet, the total error scales as $\sim \frac{1}{\sqrt{N_s}}$, where $N_s$ is the number of measurements. 

In Fig.~\ref{fig:error-bounds}c), we compare the effective 1-norm allocation to uniform allocation over all operators and to Pauli classical shadows. While uniform allocation understandably underperforms, Pauli classical shadows and the effective 1-norm allocation are on-par for this example. However, the setup considered here is among the best-case scenarios for Pauli classical shadows, with more global observables or patches around non-Identity unitaries deteriorating its relative performance.

These results are complemented by the plots shown in Fig.~\ref{fig:heisenberg_PP} in Appendix~\ref{apx:heisenberg_PP}. There we showcase a Heisenberg model HVA circuit structure instead of the TFI HVA in Fig.~\ref{fig:error-bounds}. The number of generated Pauli strings is larger, which, broadly speaking, increases the surrogation and measurement errors at the same classical or quantum resources, respectively. Due to the increased structure and number of commuting, overlapping terms in the Heisenberg Hamiltonian, the 1-norm estimation strategy underperforms without additional grouping of simultaneously measurable terms.
How to best pair commuting observables and allocate quantum computing resources for the simulation of quantum landscape patches remains an open problem for future work.

\subsection{Large-scale small-angle Pauli Propagation }\label{sec:sexynumerics}
In this section, we numerically demonstrate small-angle Pauli propagation on a large-scale problem. When picking a system to showcase the algorithm, there is a fine balance to be struck between picking a complex problem that is beyond the realm of exact simulation methods and yet one where we can be confident in the results obtained. To address this tension we will draw motivation from Ref.~\cite{miessen2024benchmarking}, where the authors used analytically known scalings for benchmarking. While Ref.~\cite{miessen2024benchmarking} takes this approach to benchmark a quantum simulation on a quantum device, we here apply the same logic to benchmark our approximate classical simulation at scale.

The system we consider is a time-dependent transverse field Ising (TFI) model with $127$ spins and on a heavy-hex entangling topology~\cite{kim2023evidence}. The time-dependent TFI Hamiltonian is related to the time-independent TFI Hamiltonian in Eq.~\eqref{eq:timeindependent-TFI} and reads
\begin{equation}\label{eq:TFI}
    H(t) = (1-g(t))H_X + g(t)H_{ZZ}\,.
\end{equation}
As before, we set the coefficients in $H_X$ and $H_{ZZ}$ (see Eq.~\eqref{eq:HX_HZZ}) to $h_i = J_{ij} = 1$. Importantly, the ramp function $g(t)$ is chosen such that it interpolates between the two Hamiltonians $H_X$ and $H_{ZZ}$ from time $t=0$ to the final time $t=t_f$, i.e., $g(0) = 0, g(t_f) = 1$. Given the initial state $\rho = |+\rangle\langle+|^{\otimes n}$, the evolution begins in the ground state of the transverse Hamiltonian. This initial state and other stabilizer states like the all-zero state $|0\rangle\langle0|$ allow for very fast and exact evaluation of the final overlaps $\Tr[\rho P]$ for any Pauli operator $P$.

The phenomenon we leverage to benchmark our classical surrogate algorithm is the \textit{Kibble-Zurek} (KZ) mechanism~\cite{PhysRevLett.95.105701, KIBBLE1980183,Zurek1985}. It predicts the density of defects in the evolved state relative to the ground state of the final Hamiltonian (here $H_{ZZ}$) as the system crosses a critical point with a finite annealing time $t_f$. Defects here are detected by non-zero values for quantities like $1-\langle Z_iZ_j\rangle$ between any two sites $i,j$. In the limit of infinitely slowly ramping between $H_X$ and $H_{ZZ}$, the state remains in the instantaneous ground state and there will be no defects. 
For a 1D TFI model, the density of defects $n_{\text{def}}$ has been shown to scale as $n_{\text{def}} \propto t_f^{-1/2}$. Interestingly, the authors in Ref.~\cite{miessen2024benchmarking} conjecture that experiments on the heavy-hex topology should follow a very similar scaling because the topology is closer to 1D than 2D. 

Here we provide numerical evidence for this conjecture by simulating the time evolution of the time-dependent TFI model in Eq.~\eqref{eq:TFI} with 127 spins on a heavy-hex topology. We use a smaller time step of $dt=0.3$, use the full heavy-hex lattice, and simulate to significantly longer times than possible with the quantum device in Ref.~\cite{miessen2024benchmarking}. In this simulation, we use a combination of \textit{weight} truncation~\cite{angrisani2024classically} and our sine-expansion order truncation. Concretely, we truncate Pauli operators with Pauli weight of $W > 5$ (i.e., more than 5 non-identity Paulis) and with a number of sines $\expansionorder > 21$. 
Constructing the classical surrogate for the $Z_{62}Z_{63}$ Pauli operator in the middle of the lattice with 50 Trotter layers took under 30 seconds using one CPU thread on a laptop. We then record the density of defects on these lattice sites via $n_{\text{def}} = 1-\langle Z_{62}Z_{63}\rangle$ on the surrogate landscape at points corresponding to different final evolution times and annealing ramps. This took under one second per evaluation. 

Our results can be seen in Fig.~\ref{fig:kz-scaling}. The linear ramp is defined via $g(t) = t/t_f$, the square ramp via  $g(t) = (t/t_f)^2$, and the tanh ramp via $g(t) = \frac12 \left[\tanh(-3\left(1-\frac{t}{t_f}\right) + 3\frac{t}{t_f})+ 1\right]$ which sweeps over the center part of a rescaled tanh function.
Not only do we observe that all three annealing ramps exhibit the analytically known scaling of the defect density, but we can verify the quality of our results by observing that the norm of the Heisenberg-evolved observable is very high between $84\%$ and $99\%$ its original 2-norm after 50 circuit layers. Thus we see that our Pauli propagation simulation algorithm vastly exceeds its theoretical guarantees and, arguably, the capabilities of most current quantum hardware.

It is evident that the classical simulations could be pushed significantly further, for example, to attempt exploring when the tree-like scaling on the heavy-hex geometry transitions into a 2D scaling. This is coincidentally the regime where the classical simulation should become substantially harder. Our work merely highlights that PP surrogates with their fast re-evaluation capabilities can be used to explore entire families of quantum protocols, particularly using our sine-truncation in the small-angle regime, for later deployment on quantum devices.

\begin{figure}
    \centering
    \includegraphics[width=0.99\linewidth]{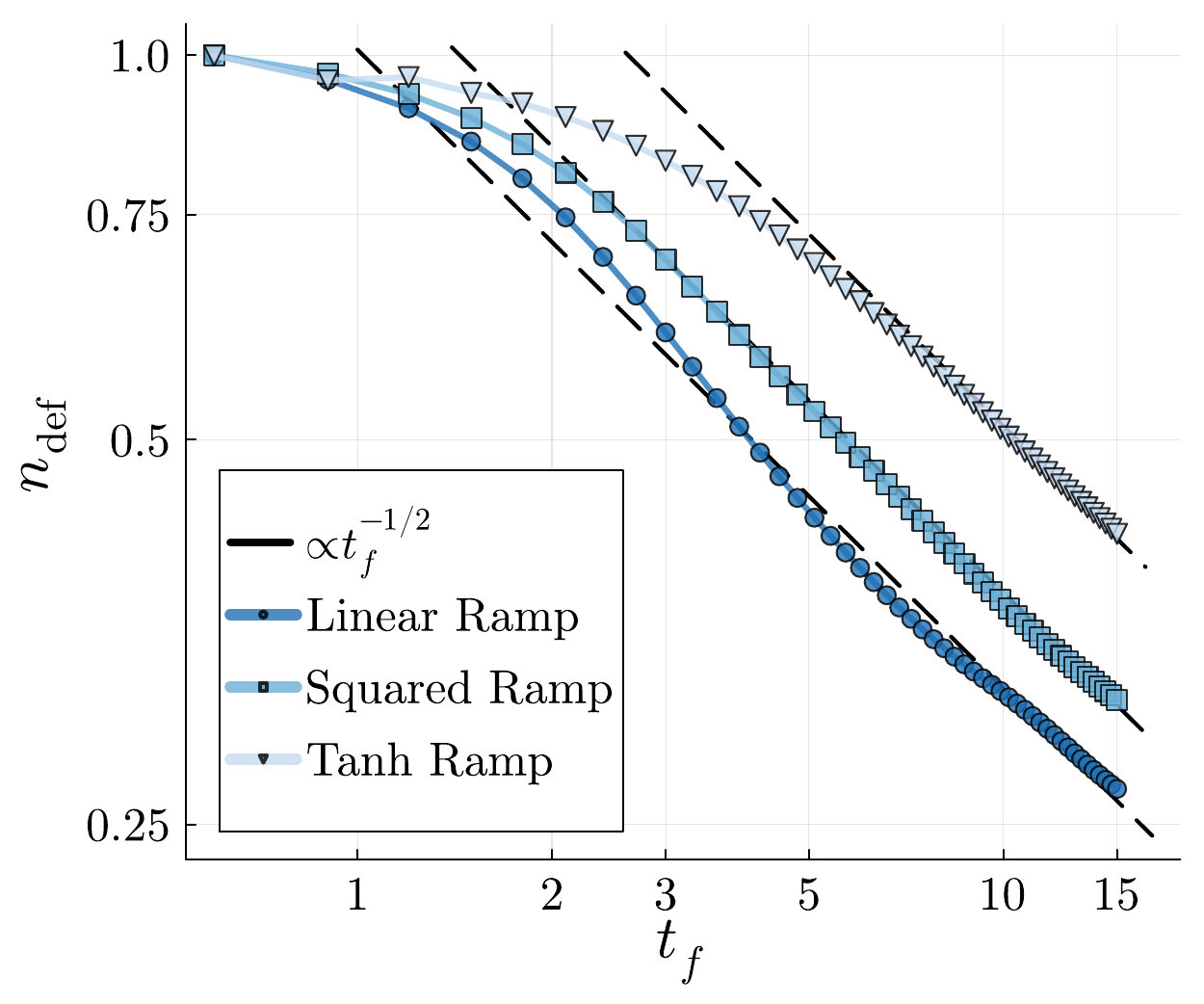}
    \caption{Testing the KZ scaling of defects $n_\text{def}$ in a time-dependent TFI model with 127 spins on a heavy-hex entangling topology. Using our Pauli propagation surrogate, we verify the conjectured scaling resembling $t_f^{-1/2}$ on the heavy-hex topology. The ramp function $g(t)$ in Eq.~\eqref{eq:TFI} is varied to confirm a similar scaling across various annealing protocols. The surrogation of the expectation landscape with 50 Trotter steps and $m=13550$ gates was done in under 30 seconds, with every expectation evaluation taking under one second.}
    \label{fig:kz-scaling}
\end{figure}

\section{Discussion}

This work described classically \textit{surrogating} a sub-region (dubbed a ``patch") of an expectation landscape produced by a parameterized quantum circuit. The surrogation procedure may require quantum resources but subsequent evaluations can be performed fully classically. This alleviates the load of quantum computers in a wide range of quantum algorithms, like VQAs, quantum machine learning, dynamical simulation, and even quantum sensing, by greatly reducing the necessary gate count and the number of measurements.

We presented three main scenarios demanding progressively more modest quantum resources. Section~\ref{sec:general_case} considered arbitrarily located patches in very general circuits acting on potentially unknown (but quantumly preparable) initial states. Here, the surrogation prescribes executing the deepest (inside this work) yet relatively undemanding, non-adaptive circuits and measurements. Section~\ref{sec:CliffordPert} considered patches centered around parameters for which the circuit is Clifford. This simplifies the task by enabling efficient classical simulation of the quantum circuit with our Pauli propagation algorithm. In this case, for states that cannot be classically simulated, but can be prepared on quantum hardware, our surrogates require only polynomially few Pauli measurements upon the initial state. And for classically tractable initial states $\rho$, whereby $\Tr[\rho P]$ with a Pauli string $P$ is efficiently classically evaluable, the surrogate can now be computed entirely classically.

We highlight that our scalings for efficient surrogatability and simulation align with prior work on the absence of barren plateau guarantees inside patches of width $\mathcal{O}(1/\sqrt{m})$~\cite{zhang2022escaping,park2023hamiltonian, wang2023trainability, park2024hardware, shi2024avoiding, puig2024variational, mhiri2024constrained}. Thus our findings here complement the conjecture connecting the provable absence of barren plateaus and efficient classical surrogatability~\cite{cerezo2023does}.

More positively, however, we have expanded the corpus of quantum-enhanced classical algorithms, consequently expanding the collection of problems that can be tackled with both classical and quantum resources. These hybrid approaches are emerging as a notable intermediary between fully classical and quantum settings, thus nuancing the threshold of practical quantum advantage~\cite{daley2022practical,herrmann2023quantum}. In particular, classical surrogation of a patch can augment a wide class of quantum algorithms, addressing shortcomings like that of insufficient coherence times and gate fidelities to reliably prepare ansatz states, but which can faithfully prepare and measure initial states. As quantum hardware improves, computations could be linked up \textit{in the middle} of the quantum circuit, as we have shown in Fig.~\ref{fig:error-bounds}, which in turn reduces the classical complexity. It may also alleviate the need for experimentally difficult parametrized unitaries that can be more straightforwardly applied classically. The surrogated patch also has utility as a \textit{symbolic} approximation to the parametrized landscape. In addition to being evaluated at particular points, it can be analytically studied and ergo integrated with, for example, symbolic and global minimization techniques~\cite{horst2000introduction}, as has been shown to improve VQA convergence~\cite{PhysRevResearch.4.023017}.

In general, we have shown that quantum preparation of a classical surrogate is often efficient and can be used to simplify an algorithm's prescribed quantum operations. 
We caution however that we have \textit{not} shown this process is necessarily faster than the original quantum algorithm running on compatible hardware. Constants and pre-factors may well mean, for example, that surrogation of a variational cost function can involve more total measurements and operations than direct quantum evaluation. This is indeed supported by our numerics in Appendix~\ref{sec:iterative-training-taylor-surrogate} where we find that Taylor surrogation outperforms our bounds, and may well have practical advantages, but has a broadly compatible performance to vanilla variational quantum algorithms. Even in the case of near-Clifford patches where the quantum resources required to surrogate are minimal, there could also end up being polynomial speed ups from running directly on quantum hardware when the initial state is quantum. Still, by adding these hybrid methods to our algorithmic arsenal and analyzing their performance, we hope to have \textit{patched} a hole in the community understanding of what is classically simulable, surrogatable, and worthwhile deploying to quantum computers.

\begin{acknowledgments}
RP acknowledges the support of the SNF Quantum Flagship Replacement Scheme (grant No. 215933). MSR acknowledges funding from the 2024 Google PhD Fellowship. AA and ZH acknowledge support from the Sandoz Family Foundation-Monique de Meuron program for Academic Promotion. TJ acknowledges support of the NCCR MARVEL, a National Centre of Competence in Research, funded by the Swiss National Science Foundation (grant number 205602). MC acknowledges supported by the Laboratory Directed Research and Development  program of Los Alamos National Laboratory (LANL) under project number 20230049DR, and by the Quantum Science Center (QSC), a National Quantum Information Science Research Center of the U.S. Department of Energy (DOE). This work was also supported by were supported by the U.S. DOE, Office of Science, Office of Advanced Scientific Computing Research, under Computational Partnerships program. ST acknowledges Exchange Faculty Travel Grant: TG168033 from Chulalongkorn University as well as funding from National Research Council of Thailand
(NRCT) [grant number N42A680126]. 
ST further acknowledges Thailand Science research and Innovation Fund Chulalongkorn University (IND\_FF\_69\_258\_2300\_062), and funding from The Office of the Permanent Secretary of the Ministry of Higher Education, Science, Research and Innovation.
This research has received funding support from the NSRF via the Program Management Unit for Human Resources \& Institutional Development, Research and Innovation [grant number B39G690073].

\end{acknowledgments}

\subsection*{Code availability}
The Pauli propagation simulations in this work were conducted with the \href{https://github.com/MSRudolph/PauliPropagation.jl}{PauliPropagation.jl} package publically available on GitHub. Code for other numerical experiments are available from the authors upon reasonable request.

\subsection*{Data availability}
The raw data for the numerical experiments are available from the authors upon reasonable request.

\bibliography{quantum,extrabib}

\onecolumngrid

\newpage
\appendix

\part{Appendix}
\parttoc 

\newpage

\section{Preliminaries}
\subsection{Notation}

    \noindent\begin{tabular}{ |p{2cm}|p{13.5cm}|  }
        \hline
        \multicolumn{2}{|c|}{\textbf{General notation table}} \\
        \hline
        Symbol & Definition  \\
        \hline
        $n$ & Number of qubits \\
        $\alv$ & A vector of variational real parameters and its $i^{\rm th}$ component \\
        $\alpha_i$ & $i^{\rm th}$ component of $\alv$ \\
        $\nparams$ & Number of variational parameters i.e., the size of $\alv$ \\
        $f(\alv)$ & The expectation value $\Tr[O \mathcal{E}_{\alv}(\rho)]$ with an observable $O$, an initial state $\rho$ and a parametrized quantum channel $\EC_{\alv}(\cdot)$ \\
        $\overline{f_{{\,\,}}}(\alv)$ & Approximate function that can be efficiently implemented.\\
        $\Delta f(\alv)$ & Difference between an approximate function and the original function $|\overline{f}(\alv) - f(\alv)|$ evaluated at $\alv$\\
        $\vol(\cenv, r)$ & Hypercube of width $2r$ around some fixed point $\cenv$ in the parameter space, $\vol(\cenv, r) := \{ \vec{\alpha} \} \; \; \text{such that} \; \; \alpha_i \in [\cen_i -r , \cen_i + r ]  \, , \, \forall \, \, i$ \\
        $ \uni(\cenv, r)$ & Uniform distribution of $\alv$ over $\vol(\cenv, r)$ also referred to as $\text{Unif}[\vol(\cenv, r)]$  \\
        $\fsur(\alv)$ & Surrogate of $f(\alv)$.\\
        $\widehat{f{{\,\,}}}(\alv)$ & Empirical etimator of $f(\alv)$. \\
        $\widehat{\fsur_{{\, \,}}}(\alv)$ & Empirical estimator of the surrogate.\\
        $f^{(P)}(\alv)$ & The subscript $(P)$ refers to $f^{(P)}(\alv) = \Tr[\rho U^\dagger (\alv) P U (\alv)]$.\\
        $\expansionorder$ & Order of expansion of a surrogate $\fsur(\alv)$ of $f(\alv)$.\\
        $N_s$ & Number of measurements.\\ 
        $\hat{o}$ & Empirical estimate of an observable with a state $\Tr[\rho O]$.\\
        $\hat{o}(\alv)$ & Empirical estimate of a back-propagated observable with a state $\Tr[\rho O(\alv)]$.\\
        $\mathcal{S}$ & Set of ``experimentally friendly'' observables. A typical choice is a set of Pauli i.e. $\SC\subseteq\PC_n$.\\
        $\Lambda$ & Upper-bound on the infinity norm of all the elements in $\mathcal{S}$, i.e. $\Lambda\geq \| P \|_\infty\; \forall P\in\mathcal{S}$.\\
        $\beta(\cdot)$ & $\beta: \mathcal{S}\to ]0,1]$ such that $\sum_{P\in\mathcal{S}}\beta(P) = 1$ and $\beta(P)>0\; \forall P\in \SC$.\\
        \hline
    \end{tabular}
\vspace{0.25 cm}

    \noindent\begin{tabular}{ |p{2cm}|p{13.5cm}|  }
        \hline
        \multicolumn{2}{|c|}{\textbf{General surrogate: notation table}} \\
        \hline
        Symbol & Definition  \\
        \hline
        $\gamma$& Constant such that the $k$-th partial derivatives are bounded as $|\partial_{i_1, i_2, ... i_k}\LC(\alv)| \leq  \gamma^k\norm{O}_\infty$\\
        $\mathcal{E}_{\alv}$ & Parametrised quantum channel.\\
        $\partial_{i_1,i_2,...,i_k}$ & Partial derivatives with respect to parameters $\{\alpha_{i_1},\alpha_{i_2},...,\alpha_{i_k}\}$ , i.e. $\partial_{i_1,i_2,...,i_k}:=\frac{\partial^k}{\prod_{j=1}^k\partial\alpha_{i_j}}$.\\
        $\partial^{\vec{k}}$& Partial derivatives of order $k_l$ with respect to each parameter $\alpha_l$ where $\vec{k}=(k_1,k_2,...,k_\nparams)$, i.e. $\partial^{\vec{k}}:=\prod_{l=1}^\nparams\frac{\partial^{k_l}}{\partial\alpha_{l}^{k_l}}$. \\
        $b_j$ & The set $\{ b_j\}$ is defined as  $\partial_{i_1, i_2, ..., i_k} f(\alv) = \sum_{j =1}^{N_{d} } b_j f(\alv_j) $.\\
        $N_d$ & Is the number of coefficients $b_j$.\\
        $b_0$ & $b_0 := \max_j |b_j|$.\\
        $\delta_i$ & $\delta_i = \alpha_i - \alpha_i^*$\\
        $\Delta_{i_1, ...,i_k}(\bm{x},\bm{x}^*) $ & $\Delta_{i_1, i_2, ...,i_k}(\bm{x},\bm{x}^*) = (x_{i_1} - x^*_{i_1})(x_{i_2} - x^*_{i_2}) ... (x_{i_k} - x^*_{i_k})$\\
        \hline
    \end{tabular}

   \noindent \begin{tabular}{ |p{2cm}|p{13.5cm}|  }
        \hline
        \multicolumn{2}{|c|}{\textbf{Pauli Propagation Surrogate: Notation table}} \\
        \hline
        Symbol & Definition  \\
        \hline
        $\mathcal{P}_n $ & $n$-qubit Pauli basis, i.e. $ \{I,X,Y,Z \}^{\otimes n}$.\\
        $\abs{P}$ & Weight of a Pauli $P\in\mathcal{P}_n$, i.e. number of qubits $P$ acts non-trivially on.\\
        $N_{\rm Paulis}$ & For $O= \sum_{P\in\PC_n} a_P P$, $N_{\rm Paulis}$ corresponds to the number of different $a_P\neq 0$.\\
        $a_P$ & Coefficients in which an observable is decomposed in the Pauli basis.\\
        $c_P(\alv)$ & Coefficients of propagated observable $O(\alv) = U(\alv)^\dagger O U(\alv)$ in the Pauli basis.\\
        ${O}(\alv)$ & Back-propagated observable. \\
        $\widetilde{O}(\alv)$ & Truncated version of the back-propagated observable. \\
        $O_{\omv}$ & Back-propagated observable along a certain path $\omv$.\\
        $\omv$ & Pauli-path.\\
        $\Phi_{\omv}(\alv)$ & Product of the coefficients obtained when propagating a Pauli $P$ through a Pauli path. \\
        $\phi_{\om_i}(\al_i)$ & Individual coefficients that are picked up after back-propagating a Pauli through one Pauli rotation. Check Eq.~\eqref{eq:pauli_path} for the possible values.\\
        $d_{\omv}$ & Overlap of the back-propagated Pauli along the path $\omv$ with the initial state $\Tr[\rho P_{\omv}]$.\\
        $n_{\rm pth}^{(\expansionorder)}$ & Number of paths found by the Pauli Propagation algorithm with truncation order $\expansionorder$.\\
        $T$ & Time that it takes to do each Pauli Propagation Step. \\
        \hline
    \end{tabular}

\subsection{Concentration of measure}
We recall some useful tail bounds which we will employ throughout the paper. For further details, see also\ \cite{dubhashi2009concentration} and references therein.
\begin{lemma}[Markov's inequality]
\label{lemma:markov}
    Let $X$ be a positive random variable (i.e. $X\geq 0$) and $\epsilon>0$. Then, we have

    \begin{equation}
        {\rm Pr}[X\geq \epsilon]\leq \frac{\Ebb[X]}{\epsilon}
    \end{equation}
\end{lemma}

\begin{lemma}[Hoeffding's inequality ~\cite{hoeffding1963probability}] \label{lemma:hoeffdingsinequality}
Given $\epsilon,\delta \in (0,1), c>0$, let $N_s = \log(2/\delta)c^2/(2\epsilon^2)$
Suppose $X_1,X_2,\dots,X_{N_s}$ are independent random variables taking values in $[a,b]$ such that $\abs{a-b}\leq c$ . Let $X$ denote their mean.
Then for any $\epsilon>0$,
\begin{align}
    \Pr[\abs{X-\mathbb{E}[X]}\geq \epsilon] \leq \delta.
\end{align}
\end{lemma}
\subsection{Central limit theorem and moments of a Normal distribution}
\label{sec:central_limit_th}
We present and explain the Central Limit theorem and discuss the moments of a random variable that follows a normal distribution. For more information we refer the reader to Ref.~\cite{billingsley1995probability}

\begin{theoremsup}[Lindeberg–Lévy: Central Limit theorem]\label{theorem:clt}
    Suppose that $\{\alpha_i\}_{i=1}^{\nparams}$ is an independent sequence of ${\nparams}$ random variables having the same distribution with mean $\Ebb[\alpha_i] = \mu$ and finite positive variance $\Var[\alpha_i] = \sigma_1^2$. If $S_{\nparams} = \sum_{i =  1}^{\nparams} \alpha_i$, then
    \begin{equation}
        X_{\nparams} = \frac{S_{\nparams} - {\nparams}\mu}{\sigma_1\sqrt{{\nparams}}} \underset{\nparams\to\infty}{\longrightarrow} \mathcal{N}(0,1)
    \end{equation}
    i.e. the random variable $X_{\nparams}$ tends to a normal distribution with mean $\mu_{\mathcal{N}} = 0$, and variance $\sigma_{\mathcal{N}} = 1$. 
\end{theoremsup}

Using \supt~\ref{theorem:clt} for a sequence of random variables $\{\alpha_i\}_{i=1}^{\nparams}$ with $\Ebb[\alpha_i] = 0$ we get that the probability distribution of the random variable $X_{\nparams}$ tends to behave as a normal distribution with variance $1$ and mean $0$, or in other words
\begin{equation}
        X_{\nparams} = \frac{S_{\nparams}}{\sigma_{\nparams}} \underset{\nparams\to\infty}{\longrightarrow} \mathcal{N}(0,1)
\end{equation}
where we have used that $\sigma_1 \sqrt{{\nparams}} = \sigma_{\nparams}$ as all the $\alpha_i$ are independent and have the same mean. 

The $k$-th moment of such distribution is defined as the expected value of the variable to the $k$-th power, i.e. $\Ebb[X_{\nparams}^k]$. This can be computed as follows~\cite{papoulis2002probability}
\begin{equation}
    \Ebb[X_{\nparams}^k] = \int_{-\infty}^\infty X_{\nparams}^k\frac{1}{\sqrt{2\pi}}e^{-\frac{X_{\nparams}^2}{2}} = 
    \begin{cases}
        0 \, {\rm if} \, k \, {\rm odd}\\
        (k-1)!! \, {\rm if} \, k \, {\rm even}
    \end{cases}
\end{equation}
where $a!!$ is the double factorial defined for $a\in \mathbb{N}$ as $a!!(a-1)!! = a!$, or in other words $a!! = a(a-2)(a-4)(a-6)...1$.

\subsection{Binomial inequalities}

In subsequent sections, we will make frequent use of the following binomial inequalities:
\begin{align}\label{eq:binom_ineq}
    \left(\frac{m}{k}\right)^k \leq \binom{m}{k} \leq \frac{m^k}{k!} < \left(\frac{em}{k}\right)^k \,,
\end{align}
where we assume $0\leq k\leq m$ and $e$ is the Euler number. 
In fact, the rightmost inequality is loose enough so that for $k \leq m$  we have
\begin{equation}\label{eq:cumbinom_ineq}
    \sum_{j=0}^k \binom{m}{j} \leq \left(\frac{em}{k}\right)^k \,.
\end{equation}
We also prove the following more general lemmas.

\begin{lemma}
For any $x\in[0,1]$ such that $k\leq mx$, we have
\begin{equation} \label{eq:binom_sum_first_powers_inequ}
    \sum_{j=0}^k\binom{m}{j}x^j \leq \left(\frac{emx}{k}\right)^k \;.
\end{equation}   
\end{lemma}

\begin{proof}
    \begin{equation}
        \sum_{j=0}^k\binom{m}{j}x^j \leq \sum_{j=0}^k \frac{(mx)^j}{j!}=\sum_{j=0}^k \left(\frac{mx}{k}\right)^j\frac{k^j}{j!} \leq \left(\frac{mx}{k}\right)^k \sum_{j=0}^k \frac{k^j}{j!} \leq \left(\frac{emx}{k}\right)^k\;,
    \end{equation}
        where we first use Eq.~\eqref{eq:binom_ineq}. The second inequality is obtained from the condition $k\leq xm$ i.e. $\left(\frac{xm}{k}\right)^j$ increases with $j$ so reaches its maximum for $j=k$. The last inequality is obtained by adding terms in the sum (up to $j=\infty$) and recognising the exponential series. 
\end{proof}
\begin{lemma}
\label{lem:bin-up}
 For any $x $ such that $x \leq k/m$ we have
\begin{equation}\label{eq:binom_sum_last_powers_ineq}
    \sum_{j=k}^m\binom{m}{j}x^j \leq \left(\frac{emx}{k}\right)^k \;.
\end{equation}   
\end{lemma}
\begin{proof}
The proof of Eq.~\eqref{eq:binom_sum_last_powers_ineq} is similar to Eq.~\eqref{eq:binom_sum_first_powers_inequ} and reads as follows
    \begin{align}
        \sum_{j=k}^m\binom{m}{j}x^j \leq \sum_{j=k}^m x^j \frac{m^j}{j!}
        = \sum_{j=k}^m  \left(\frac{xm}{k}\right)^j \frac{k^j}{j!}
        \leq \left(\frac{xm}{k}\right)^k \sum_{j=k}^m   \frac{k^j}{j!}
        \leq\left(\frac{emx}{k}\right)^k\;,
    \end{align}
    where we first use Eq.~\eqref{eq:binom_ineq}. The second inequality is obtained from the condition $xm\leq k$ i.e. $\left(\frac{xm}{k}\right)^j$ decreases with $j$ so reaches a maximum for $j=k$. The last inequality is obtained by adding terms in the sum (from $j=0$ to $j=\infty$) and recognizing the exponential series. 
\end{proof}

\subsection{Lagrange multipliers}
In this section we present the method of Lagrange multipliers to find the minima of a function subjected to some constraints. We refer the reader to Ref.~\cite{protter2012intermediate}. The optimisation algorithm works as follows. \\

\textbf{Lagrange multipliers method: optimisation algorithm.}\label{alg:lagrange}
\begin{enumerate}
    \item Let $f(\vec{x})$ to be the function to optimise.
    \item Let $g(\vec{x}) = 0$ to be the constrain function.
    \item Define the Lagrangian: $\mathcal{L}(\vec{x}, \lambda) = f(\vec{x}) + \lambda g(\vec{x})$. 
    \item Compute the derivatives: $\frac{\partial}{\partial x_i}\mathcal{L}(\vec{x},\lambda)$ and $\frac{\partial}{\partial \lambda}\mathcal{L}(\vec{x},\lambda)$.
    \item Solve the following system of equations
    $$      
                \begin{cases}
                   \frac{\partial}{\partial x_i}\mathcal{L}(\vec{x},\lambda) = 0\,,\\
                    \frac{\partial}{\partial \lambda}\mathcal{L}(\vec{x},\lambda) = 0\,.
                    \end{cases}
            $$
    \item Discard $\lambda^{(sol)}$.
    \item The points $\vec{x}^{\rm (sol)}$ are the solution to the problem.
\end{enumerate}

\subsection{Taylor remainder theorem}
We now introduce the Taylor Remainder Theorem, which will be used below to bound the error arising from making a $k$-order expansion in the expectation function. For further details on the Taylor reminder theorem we refer the reader to Ref.~\cite{taylorbook}.
\begin{theoremsup}[Taylor Reminder Theorem] \label{thm:reminder_taylor}
Consider a multivariate differentiable function $f(\vec{x})$ such that $f: \mathbb{R}^\nparams \rightarrow \mathbb{R}$ and some positive integer $\expansionorder$. The function $f(\vec{x})$ can be expanded around some fixed point $\vec{x}^*\in \mathbb{R}^\nparams$ as
\begin{align}
    f(\bm{x}) = \sum_{k = 0}^{\expansionorder} \, \sum_{i_1, i_2, ...,i_k = 1}^\nparams \frac{1}{k!} \partial_{i_1, i_2, ...,i_k}f(\bm{x}^*)\Delta_{i_1, i_2, ...,i_k}(\bm{x},\bm{x}^*) + \epsilon_{\expansionorder,\vec{x}^*}(\vec{x}) \;,
\end{align}
where we are denoting
\begin{align}
    &\partial_{i_1, i_2, ...,i_k}f(\bm{x}^*) := \left( \frac{\partial^k  f(\bm{x})}{\partial x_{i_1} \partial x_{i_2} ... \partial x_{i_k}} \right)\bigg|_{\vec{x} := \vec{x}^*}\,,\\
    &\Delta_{i_1, i_2, ...,i_k}(\bm{x},\bm{x}^*) = (x_{i_1} - x^*_{i_1})(x_{i_2} - x^*_{i_2}) ... (x_{i_k} - x^*_{i_k})\,.
\end{align}

Then, the remainder of this function $ \epsilon_{\expansionorder,\vec{x}^*}(\vec{x})$ can be expressed as follows
\begin{align}
     \epsilon_{\expansionorder,\vec{x}^*}(\vec{x}) = \sum_{i_1, i_2, ...,i_{\expansionorder+1} = 1}^\nparams \frac{1}{(\expansionorder+1)!} \partial_{i_1, i_2, ...,i_{\expansionorder+1}}f(\tilde{\bm{x}})\Delta_{i_1, i_2, ...,i_{\expansionorder+1}}(\bm{x},\bm{x}^*) \;,
\end{align}
with $\tilde{\bm{x}} = \tau \vec{x} + (1-\tau) \vec{x}^*$ for some $\tau \in [0,1]$. 
\end{theoremsup}

\section{Shot noise analysis for estimating observables with restricted measurements}
\label{app:estimating-observables}

In this section we analyse the effect of shot noise for different strategies of computing families of observables. We will use these results in the later appendices to compute the effect of shot noise for our different surrogation strategies. 

\medskip

Given an observable $O$ and a quantum state $\rho$, we consider the task of estimating the expectation values $\Tr[\rho O]$ with quantum experiments. In many practical settings, a direct measurement of the observable $O$ on the state $\rho$ might be experimentally unfeasible. Thus, we define a set $\mathcal{S}$ of ``experimentally friendly'' observables such that we can decompose $O$  as
\begin{align}
    O = \sum_{P\in \mathcal{S}} a_P P \, , 
\end{align}
where $|S|$ is of a manageable size and each $P$ is easily measurable. 
A typical choice of $ \mathcal{S}$ is a subset of the Pauli basis, i.e. $\mathcal{P}_n\coloneqq\{I,X,Y,Z\}^{\otimes n}$. While in Sections\ \ref{app:sine-cosine-surrogate-non-Clifford} and\ \ref{app:smallanglePP} we will precise consider such basis, we here make no assumption, and leave $\mathcal{S}$ as a general set of observables. 
By measuring each $P\in\mathcal{S}$ individually, we can estimate an approximate value of $\Tr[O\rho]$. We start by providing a general measurement protocol for estimating 
$\Tr[\rho O]$ up to a desired error $\epsilon$.

\begin{lemma}[Parametric estimator with restricted measurements]
\label{lemma:hoeffding-obs}
Let $O = \sum_{P\in \mathcal{S}} a_P P$, where $a_P$ are non-zero real coefficients and $\mathcal{S}$ is a set of observables, for which each element can be measured efficiently. Moreover, assume that $\norm{P}_\infty \leq \Lambda$ for all $P\in \mathcal{S}$. Let $\beta : \mathcal{S} \rightarrow ]0,1]$ be a probability distribution whose support is $\mathcal{S}$ (i.e., $\beta(\cdot)$ satisfies $\sum_{P\in \mathcal{S}} \beta(P) =1$ and $\beta(P)>0$ $\forall P\in\mathcal{S}$).  \begin{enumerate}
    \item There exists a randomized measurement protocol which uses $\Lambda^2/\epsilon^2$ measurements on different copies of an unknown state $\rho$, to output an estimate $\hat{o}$ satisfying
\begin{align}
    \sqrt{\Ebb[\left(\hat{o} - \Tr[O \rho] \right)^2]} <  \,   \epsilon \sqrt{\sum_{P\in\mathcal{S}}\frac{a_P^2}{\beta(P)}}.
\end{align}
    \item There exists a probabilistic measurement protocol which uses $2\Lambda^2 \log(2/\delta)/\epsilon^2$ measurements on different copies of an unknown state $\rho$, to output an estimate $\hat{o}$ satisfying
\begin{align}
    \abs{\hat{o} - \Tr[O \rho]} < \epsilon \cdot  \max_{P\in\mathcal{S} } \frac{\abs{a_{P}}}{\beta(P)},
\end{align}
with probability at least $1-\delta$.
\end{enumerate}

\end{lemma}
\begin{proof}
Given $N_s\geq 1$,  for $i = 1,2,\dots, N_s$, we perform the following steps:
\begin{enumerate}
    \item Sample an operator $P_i\in \SC$ with probability $\beta(P_i)$\,.
    \item Measure $P_i$ on (an unused copy of) the state $\rho$\,.
    \item Denote the measurement outcome as $\hat x_i$\,.
    \item Define the random variable
    \begin{align}
        \widehat X_i = \frac{a_{P_i}}{\beta(P_i)} \hat x_i. 
    \end{align}
\end{enumerate}
We notice that $ \widehat X_i$ is an unbiased estimator for $\Tr[\rho O]$:
\begin{align}
    \mathbb{E}  \widehat X_i = \sum_{P\in \SC} \beta(P) \frac{a_{P}}{\beta(P)} \Tr[P\rho] = \Tr[O\rho].
\end{align}

We also note that
\begin{align}
    \abs{\widehat{X_i}} \leq \Lambda \max_{P\in\SC} \frac{\abs{a_P}}{\beta(P)},
\end{align}
and therefore $\widehat{X_i}$ takes value in an interval of size at most $2\Lambda\max_{P\in\SC}\left\{ \frac{\abs{a_P}}{\beta(P)}\right\}$.
Let $\hat{o}$ be the mean of $\widehat{X}_1,\widehat{X}_2,\dots,\widehat{X}_{N_s}$. By Lemma~\ref{lemma:hoeffdingsinequality} (Hoeffding's inequality), we have that
\begin{align}
     \Pr\left[\abs{\hat{o} -  \Tr[O\rho]} < \epsilon\cdot  \Lambda \max_{P\in\SC}\frac{\abs{a_P}}{\beta(P)}\right] \geq 1 - \delta,
\end{align}
provided that $N_s \geq 2\log(2/\delta)/\epsilon^2$ which proves the second statement of the lemma.

Moreover, we have 
\begin{align}
    \Ebb[\left(\hat{o} - \Tr[O \rho] \right)^2] &= \Var[\hat{o}] \\
    &= \frac{1}{N_s^2}\sum_{i=1}^{N_s}\Var[\widehat{X_i}] \\
    &\leq \frac{1}{N_s^2}\sum_{i=1}^{N_s}\Ebb[\widehat{X_i}^2] \\
    & = \sum_{P\in\SC}\frac{a_P^2}{\beta(P)}\frac{1}{N_s^2}\sum_{i=1}^{N_s}\Ebb[\hat{x}_i^2] \\
    &\leq \sum_{P\in\SC}\frac{a_P^2\Lambda^2}{\beta(P) N_s} \;, \label{eq:general-variance-single-estimator-no-params}
\end{align}
where the last inequality is due to $|\hat{x}_i|\leq \Lambda$. This proves the first statement of the lemma by setting $N_s=\Lambda^2/\epsilon^2$.
\end{proof}
For any given values of the coefficients $a_P$, one can analytically compute the distribution $\beta(P)$ minimizing the function 
\begin{align}
    \max_{P\in\SC} \frac{\abs{a_{P}}}{\beta(P)}\;,
\end{align} 
and run the corresponding measurement protocol, as we demonstrate in the following section. Notice that we assumed $\beta(P)>0$ for $a_P\neq 0$. This condition will naturally be met with the following settings for $\beta$.

\subsection{Warm-up: estimating a single observable}
\label{sec:warm-up}
The protocol presented in Lemma~\ref{lemma:hoeffding-obs}  is parameterized by a distribution \(\beta\), which defines the probability of measuring each observable \( P\in \SC \) in a single round of the protocol. After \( N_s = \Omega(\log(1/\delta)/\epsilon^2) \) rounds, each observable \( P\in \SC \) is measured \( \beta(P) \times N_s \) times, in expectation. Intuitively, more measurement shots should be allocated to the observables \( P \) that have larger coefficients \( a_P \) (in absolute value). The following corollary formalizes this intuition.
\begin{corollary}
 Let $O = \sum_{P\in \SC} a_P P$ be an observable.
There exists a randomized measurement protocol which uses $2\norm{\boldsymbol{a}}_1^2\log(2/\delta)\Lambda^2/\epsilon^2$ measurements on different copies of an unknown state $\rho$, which outputs an estimate $\hat{o}$ satisfying
\begin{align}
    \abs{\hat{o} - \Tr[O \rho]} < \epsilon.
\end{align}   
\label{cor:warm-up}
\end{corollary}
\begin{proof}
    The corollary follows from  Lemma~\ref{lemma:hoeffding-obs}  by setting $\beta(\cdot)$ as
\begin{align}
    \beta(P) = \frac{|a_P|}{\sum_{P'\in \SC}|a_{P'}|} := \frac{|a_P|}{\norm{\boldsymbol{a}}_1}.
\end{align}
\end{proof}
We note that this shot allocation strategy is not novel and has been explored in a number of contexts including the variational quantum algorithms~\cite{arrasmith2020operator,rubin2018application}.
Crucially, this choice of \(\beta\) minimizes the loss function \( \mathcal{L}(\beta)=\max_P {|a_P|}/{\beta(P)}\), which can be verified using Lagrange multipliers (see Appendix~\ref{app:optimal-beta-lagrange}). One might wonder whether a simpler protocol, such as allocating an equal probability to all operators in $\SC$, could still achieve comparable performance. However, a uniform allocation of measurement shots becomes highly suboptimal when the observable \( O \) is concentrated around a few elements of $\SC$. In such cases, a uniform strategy would result in an exponentially larger sample complexity, as we demonstrate in the following discussion.

\bigskip

\noindent\textbf{Comparison with uniform allocation.} One of the simplest choices for $\beta(\cdot)$ is the uniform distribution over the set $\SC$, that is
\begin{align}
    \beta(P) = \begin{cases}
        \frac{1}{\abs{\SC}} & \text{ if $P\in \SC$,}
        \\ 0 & \text{otherwise},
    \end{cases}
\end{align}
which corresponds to measuring each observable (in $\SC$) $N_s/\abs{\SC}$ times in expectation. 
By Lemma~\ref{lemma:hoeffding-obs}, we have that $N_s=2\log(2/\delta)/\epsilon^2$ measurement shots are sufficient to ensure
\begin{equation}
    |\hat{o}_{\mathrm{uniform}}-\Tr[O\rho]|\leq \epsilon \abs{\SC} \max_P {\abs{c_{P}}}< \epsilon \Lambda \abs{\SC} \;, \label{eq:uniform} 
\end{equation}
with probability at least $1-\delta$.
Setting $\epsilon = \frac{\tilde{\epsilon}}{\Lambda\abs{\SC}}$, we obtain that $N_s=2\abs{\SC}^2\Lambda^2\log(2/\delta)/\tilde{\epsilon}^2$ many measurements are required to achieve an error $\tilde{\epsilon}$ with probability at least $1-\delta$.
To demonstrate the limitations of this measurement protocol, consider the following extreme example of an observable:
\[
    O_{\mathrm{peaked}} = Z^{\otimes n} + \tau \sum_{P \in \{I, X, Y, Z\}^{\otimes n} \setminus \{Z^{\otimes n}\}} \frac{P}{4^n - 1}.
\]
for some constant $\tau  = \Theta(1)$.
In the above expression, the observable is dominated by a single Pauli term \( Z^{\otimes n} \), while all other Pauli strings contribute only marginally with coefficients scaled by an exponentially small factor \(\tau/(4^n - 1)\).

In this case $\abs{\SC} = 4^n$ and $\Lambda=1$ (i.e. $\SC$ is the set of $n$ qubits Pauli), and thus the samples complexity for the uniform allocation scales as $\OC(16^n\log(1/\delta)/\tilde{\epsilon}^2)$.

On the other hand, we have that 
\begin{align}
    \norm{\boldsymbol{a}}_1 = 1 + \tau \in\Theta(1),
\end{align}
and therefore the sample complexity of the protocol described in Corollary~\ref{cor:warm-up} is $ \OC(\log(1/\delta)/\tilde\epsilon^2)$ for achieving error $\tilde{\epsilon}$ and success probability $1-\delta$. Crucially, the uniform allocation requires exponentially more measurements.

\subsection{Estimating parametrized observables}
\label{app:estimating-observables-parametrized}
Throughout this paper, we often encounter \emph{parametrized} observables $O(\boldsymbol\alpha)=\sum_{P\in \mathcal{S}}c_P{(\boldsymbol\alpha)} P$, where $\boldsymbol{\alpha}$ is sampled from a given distribution $\mathcal{D}$.
Here, we aim at estimating a function $\hat{o}(\boldsymbol{\alpha})$ minimizing the following \emph{mean square error}:
\begin{align}
    \Ebb_{\boldsymbol{\alpha} \sim \mathcal{D}} \left(\Tr[O(\boldsymbol\alpha) \rho] - \hat{o}(\boldsymbol{\alpha}) \right)^2.
\end{align}

In such case, designing an optimal measurement protocol is more challenging, since the vector $\vec{c}(\boldsymbol{\alpha})$ is now a random variable.
Consequently, the measurement distribution $\beta(\cdot)$ should depend on the distribution $\mathcal{D}$.

Broadly speaking, we show that, while the sample complexity for a single observable (i.e., a fixed value vector $\alv_0$) is dominated by the 1-norm $\norm{\vec{c}(\alv_0)}_1$, the sample complexity for the multiple observables cases can be described via a suitable ``effective 1-norm'', which we define in the following.
\begin{definition}[Effective 1-norms]
\label{def:effective-1-norms}
  Let $\boldsymbol{\alpha}$ be a parameter vector, and  let $O(\boldsymbol{\alpha} ) = \sum_{P\in \SC} c_P(\alv) P$ be a parametrized observable. 
      \begin{itemize}      
      \item  Let us assume that $\vec{\alpha}$ is sampled from a distribution $\mathcal{D}$. Then the associated average-case effective 1-norm is defined as
      \begin{align}
      \label{def:avg-case-effective-1-norm-mu}
           \norm{\boldsymbol{c}}_{1,\mathrm{avg}}:= \sum_{P\in \SC}\sqrt{\mathbb{E}_{\boldsymbol{\alpha} \sim \mathcal{D}} [c_P(\alv)^2]}.
      \end{align}
        \item Let us assume that $\vec{\alpha}$ lies in a set $\mathcal{A}$. Then the associated \emph{worst-case effective 1-norm} is defined as
      \begin{align}
      \label{def:worst-case-effective-1-norm-mu}
          \norm{\boldsymbol{c}}_{1,\mathrm{worst}}:=\sum_{P\in \SC} \max_{\boldsymbol{\alpha} \in \mathcal{A}} \abs{c_P(\alv)}.
      \end{align}
      \end{itemize}
\end{definition}

\begin{lemma}[Mean square error]
\label{lemma:multiple-observable-mean-squared-error}    
 Let us assume that $\vec{\alpha}$ is sampled from a distribution $\mathcal{D}$. 
There exists a randomized measurement protocol which uses $N_s=\Lambda^2\norm{\boldsymbol{c}}^2_{1,\mathrm{avg}}/\epsilon^2$ measurements on different copies of a state $\rho$, and outputs an estimate $\hat{o}(\alv)$ for each $\alv$ satisfying
    \begin{equation}
    \sqrt{\Ebb_{\boldsymbol{\alpha} \sim \mathcal{D},\MC} \left(\Tr[O(\boldsymbol\alpha) \rho] - \hat{o}(\boldsymbol{\alpha}) \right)^2} \leq \epsilon\;,
    \end{equation}
    where $\Ebb_{\MC}$ denotes an average over the measurement outcomes. 
\end{lemma}
\begin{proof}

Let $\beta(\cdot)$ be defined as
\begin{align}
    \beta(P) = \frac{\sqrt{\mathbb{E}_{\boldsymbol{\alpha} \sim \mathcal{D}} [c_P(\alv)^2]}}{\norm{\boldsymbol{c}}_{1,\mathrm{avg}}}.
\end{align}
 The measurement protocol here is similar to Lemma~\ref{lemma:hoeffding-obs}, but in the fourth step, we use the outcome of the third step, $\hat{x}_i$ to define the random variable 
 \begin{align}
     \widehat{X}_i(\alv)=\frac{c_{P_i}(\alv)}{\beta(P_i)}\hat{x}_i
 \end{align} 
 
 for each $\alv$. Recall that $\widehat{X}_i(\alv)$ is an unbiased estimator of $\Tr[\rho O(\alv)]$. Let 
 \begin{equation}
     \hat{o}(\alv) = \frac{1}{N_s}\sum_{i=1}^{N_s}\widehat{X}_i(\alv)\, .
 \end{equation}

 We have also seen in Eq.~\eqref{eq:general-variance-single-estimator-no-params} that
 \begin{align} \label{eq:variance_single_point_estimator_general_case}
     \Ebb_{\MC}[(\hat{o}(\alv)-\Tr[\rho O(\alv)])^2] \leq \sum_{P\in \SC}\frac{c_P(\alv)^2\Lambda^2}{\beta(P)N_s}\;,
 \end{align}
which is true for any $\alv$. Therefore, averaging over $\alv\sim\DC$ and inserting the definition $\beta(P)$ leads to the desired result i.e. 
  \begin{align}
     \Ebb_{\alv\sim\DC,\MC}[(\hat{o}(\alv)-\Tr[\rho O(\alv)])^2] &\leq \sum_{P\in \SC}\frac{\Ebb_{\alv\sim\DC}[c_P(\alv)^2]\Lambda^2}{\beta(P)N_s} \\
     &= \frac{\norm{\boldsymbol{c}}_{1,\mathrm{avg}}^2\Lambda^2}{N_s}\;.
 \end{align}

Finally, we complete the proof by setting $N_s=\frac{\Lambda^2\norm{\boldsymbol{c}}^2_{1,\mathrm{avg}} }{\epsilon^2}$.
\end{proof}

\begin{lemma}[Maximum error]
\label{lemma:multiple-observable-max-error} 
Let us assume that $\vec{\alpha}$ lies in the set $\mathcal{A}$.
There exists a randomized measurement protocol which uses $N_s=\frac{2\log(2/\delta)\Lambda^2\norm{\boldsymbol{c}}^2_{1,\mathrm{worst}}}{\epsilon^2}$ measurements on different copies of a state $\rho$, and output an estimate $\hat{o}(\alv)$ satisfying
    \begin{equation}
     \abs{\Tr[O(\boldsymbol\alpha) \rho] - \hat{o}(\boldsymbol{\alpha})}  < \epsilon\;,
    \end{equation}
    with probability $1-\delta$ (over the internal randomness of the measurement protocol).
\end{lemma}
\begin{proof}

    The measurement protocol is similar to Lemma~\ref{lemma:hoeffding-obs} with the sampling distribution $\beta(\cdot)$ defined as 
    \begin{align}
        \beta(P)=\frac{\max_{\alv\in\AC}|c_P(\alv)|}{\norm{\boldsymbol{c}}_{1,\mathrm{worst}}} \;.
    \end{align}
    Then, for any fixed $\alv\in\AC$ we can use Hoeffding's inequality (i.e., Lemma~\ref{lemma:hoeffdingsinequality}) as in Lemma~\ref{lemma:hoeffding-obs} to get
    \begin{align}
        \label{eq:worst-case-error-hoeffding-general-upperbound}
        \abs{\Tr[O(\boldsymbol\alpha) \rho] - \hat{o}(\boldsymbol{\alpha})} < \tilde{\epsilon}\cdot \Lambda \max_{P\in\SC}\frac{|c_P(\alv)|}{\beta(P)}\;,
    \end{align}
    with probability at least $1-\delta$ using $N_s=2\log(2/\delta)/\tilde{\epsilon}^2$. From previous definition of $\beta(P)$, we have $\max_P\frac{|c_P(\alv)|}{\beta(P)}\leq \norm{\boldsymbol{c}}_{1,\mathrm{worst}}$. Thus, substituting $ \tilde{\epsilon}=\frac{\epsilon}{\norm{\boldsymbol{c}}_{1,\mathrm{worst}}\Lambda}$  completes the proof.
\end{proof}

Notice that the sampling distribution, $\beta(.)$ used in Lemma~\ref{lemma:multiple-observable-max-error} leads to the same bound on the MSE as in Lemma~\ref{lemma:multiple-observable-mean-squared-error}, but with the average-case effective 1-norm being replaced by the worst-case one. 

\subsection{Optimal sampling distributions}
\label{app:optimal-beta-lagrange}
In order to find the distribution $\beta(P)$ that minimize the MSE in Eq.~\eqref{eq:variance_single_point_estimator_general_case}, we used Lagrange multipliers method (see Appendix~\ref{alg:lagrange}) to minimize the function

\begin{align}\label{eq:averagebound}
     \sum_{P\in \SC}\frac{\Ebb_{\alv}[c_P^2(\vec{\alpha})]}{\beta(P) }\;,
\end{align}
under the constraint $\sum_{P\in\SC}\beta(P) = 1$. The appropriate Lagrangian is given by
\begin{align}
    \mathcal{L} = \sum_{P\in\SC}\frac{\Ebb_{\alv}[c_P^2(\vec{\alpha})]}{\beta(P) } +\lambda\left(\sum_{P\in\SC}\beta(P)-1\right)\,,
\end{align}
where $\lambda$ is the Lagrangian multiplier. Next we compute the derivatives with respect to $\beta(P)$ and equate them to zero. The previous leads to  
\begin{equation}
    \frac{\partial\mathcal{L}}{\partial \beta(P)} = -\frac{\Ebb_{\alv}[c_P^2(\vec{\alpha})]}{\beta(P)^2 } + \lambda = 0\;.
\end{equation}

Then, the solution to the system of equations $\{\partial_{\beta(P)}\mathcal{L} = 0\}_{P\in\SC}$ with the normalization constraint is
\begin{equation}\label{eq:allocation_av}
\beta(P) = \frac{\sqrt{\Ebb_{\alv}[c_P^2(\vec{\alpha})]}}{\sum_{P\in\SC}\sqrt{\Ebb_{\alv}[c_P^2(\vec{\alpha})]}}=\frac{\sqrt{\Ebb_{\alv}[c_P^2(\vec{\alpha})]}}{\norm{\boldsymbol{c}}_{1,\mathrm{avg}}}\;.
\end{equation} 

The derivation of the optimal distribution $\beta(P)$ in Lemma~\ref{lemma:multiple-observable-max-error} is similar to Lemma~\ref{lemma:multiple-observable-mean-squared-error}, but here we need to minimize Eq.~\eqref{eq:worst-case-error-hoeffding-general-upperbound} for any choice of $\alv\in\AC$ which is equivalent to minimizing the function

\begin{equation}
    \max_P \max_{\alv} \frac{|c_P(\alv)|}{\beta(P)}\;.
\end{equation}

Let us define the vector $\vec{v}$ with entries $v_P = \max_{\alv} \frac{|c_P(\alv)|}{\beta(P)}$. So, we want to minimize $\norm{\vec{v}}_\infty$ under the constraint $\sum_{P\in\SC}\beta(P)=1$. Let us consider the $k$-norm of $\vec{v}$ and then we will take the limit $k\to\infty$. The Lagrangian is given by
\begin{equation}
    \mathcal{L}=\left(\sum_{P\in\SC}v_P^k\right)^{\frac{1}{k}}+\lambda\left(\sum_{P\in\SC}\beta(P)-1\right)\;,
\end{equation}
which is minimized if, for every $\beta(P)$, we have
\begin{equation}
    -\frac{(\max_{\alv}|c_P(\alv) |)^k}{\beta(P)^{k+1}}\left(\sum_{P\in\SC}v_P^k\right)^{\frac{1-k}{k}}+\lambda =0\;.
\end{equation}
The solution is given by
\begin{equation}
    \beta(P)=\frac{(\max_{\alv}|c_P(\alv) |)^{\frac{k}{k+1}}}{\sum_{P'\in\SC}(\max_{\alv}|c_{P'}(\alv) |)^{\frac{k}{k+1}}}\;.
\end{equation}
Finally, by taking the limit $k\to\infty$, we end up with 
\begin{equation}
\beta(P)=\frac{\max_{\alv}|c_P(\alv) |}{\sum_{P'\in\SC}\max_{\alv}|c_{P'}(\alv) |} =\frac{\max_{\alv}|c_P(\alv) |}{\norm{\vec{c}}_{1,{\rm worst}}} \;.
\end{equation}

\section{General surrogatability guarantee}\label{app:generalsurrogatetheorem}

\subsection{Summary of the key results}\label{app:generalsurrogatetheorem-1}

In this section we will start by summarising the main results of which our general surrogatability guarantee is comprised. Their proofs are then provided in subsequent sections of this appendix.

\subsubsection{Taylor surrogate construction}
We recall that we are interested in expectation functions of the general form
\begin{align}\label{eq:loss_app}
    \LC(\alv)=\Tr[\EC_{\alv}(\rho) O]\,,
\end{align}
where $\EC_{\alv}(\cdot)$ is an arbitrary parameterized quantum channel with ${\nparams}$ variational parameters $\alv$ acting on a $n-$qubit quantum state $\rho$ and $O$ is some Hermitian operator. We are interested in constructing the surrogate $\fsur_\expansionorder(\alv)$ in the hypercube region $\vol(\alv^*,r)$ i.e.,
\begin{equation}\label{eq:hypercube-region-appendix}
	\vol(\vec{\alpha}^*, r) := \{ \vec{\alpha} \} \; \; \text{such that} \; \; \alpha_i \in [\alpha_i^* -r , \alpha_i^* + r ]  \, \, \forall \, \, i \,.
\end{equation}
This expectation function can be expressed using a multivariate Taylor's expansion as
\begin{align}\label{eq:taylor_loss}
    \lf =  \sum_{k = 0}^{\infty} \frac{1}{k!}\sum_{i_1, i_2, ..,i_k=1}^{\nparams} \gradlk \dellk \;,
\end{align}
where we introduce two compressed notations for ease of notation
\begin{align}
    \gradlk & =  \left( \frac{\partial^k  \lf}{\partial \alpha_{i_1} \partial \alpha_{i_2} ... \partial \alpha_{i_k}} \right)\bigg|_{\vec{\alpha} = \vec{\alpha}^*}  \;, \\
    \dellk & = (\alpha_{i_1} - \alpha^*_{i_1})(\alpha_{i_2} - \alpha^*_{i_2}) ... (\alpha_{i_k} - \alpha^*_{i_k}) \;\;.\label{eq:def_deltath}
\end{align}

The $\expansionorder^{\rm th}$-order surrogate of the expectation function is made by simply truncating the infinite sum in  Eq.~\eqref{eq:taylor_loss} to  order $\expansionorder$ 
\begin{equation}\label{eq:surrogate_taylor}
    \fsur_\expansionorder(\alv) =  \sum_{k = 0}^{\expansionorder} \frac{1}{k!}\sum_{i_1, i_2, ..,i_k=1}^{\nparams} \gradlk \dellk \;.
\end{equation}
In order to construct the surrogate, we are required to obtain all partial derivatives of the expectation function up to order $\expansionorder$. These partial derivatives can be obtained by evaluating the expectation function $f(\alv)$ at different points on the landscape.

The following three key assumptions are made for the surrogatability guarantee in our proofs.
\begin{assumption}\label{assumption:bounded-derivatives}
The partial derivatives of $f(\alv)$ within the region are bounded as
\begin{align}\label{eq:assumption-bounded-derovatives}
     |\partial_{i_1, i_2, ... i_k}\LC(\alv)| \leq  \gamma^k\norm{O}_\infty \;\; \;,\;\; \forall \alv \in \vol(\alv^*, r)\;,
\end{align}
with some constant $\gamma \in \OC(1)$.
\end{assumption}

\begin{assumption}\label{assumption:bounded-loss-evaluations}
Computing each partial derivative (up to the order $\expansionorder$) requires at most $N_{d} $ evaluations of the expectation function which scales as
\begin{align}\label{eq:assumption-bounded-loss-evaluations}
    N_{d} \in \OC\left({\nparams}^\expansionorder\right) \;.
\end{align}
\end{assumption}

\begin{assumption}\label{assumption:loss-evaluations-expectations}
The partial derivative can be expressed as a linear combination of expectation functions evaluated at different parameters 
\begin{align}\label{eq:assumption-3-linear-combination}
    \partial_{i_1, i_2, ..., i_k} f(\alv) = \sum_{j =1}^{N_{d} } b_j f(\alv_j) \;,
\end{align}
where each $b_j$ is an associated coefficient to $f(\alv_j)$ such that $ b_{0} = \max_{j} |b_j| $ with some positive $b_{0}\in\OC(\poly(m))$, and can be efficiently computed classically. Furthermore, a set of $\{ \alv_j \}_{j=1}^{N_{d} }$ depends on $\alv$ and can be efficiently determined classically. Note that for each combination of the derivatives $\partial_{i_1,...,i_k}$ the dependence on $\{b_j\}_{j=1}^{N_d}$ and $\{ \alv_j \}_{j=1}^{N_{d} }$ can be different. However, we need the condition $ b_{0} = \max_{j} |b_j| $ to hold for all the derivatives. 
\end{assumption}
\noindent We note that two well-known approaches to compute gradients that satisfy Assumption~\ref{assumption:bounded-loss-evaluations} and Assumption~\ref{assumption:loss-evaluations-expectations} are parameter shift-rule and numerical finite difference. \\

In what follows, we present a series of formal statements with the ultimate aim of providing the surrogatability guarantee over the region. We present such result in the form of  \supt~\ref{coro:surrogate_quantum_quantum}, which is a formal version of Theorem~\ref{thm:surrogate-average-general} in the main text.

\subsubsection{Worst-case error}
The worst-case error of the $\expansionorder^{\rm th}$ order surrogate in the region $\vol(\alv^{*}, r)$ around the parameters $\alv^*$ is defined as
\begin{equation}
    \epsilon_{\expansionorder,\alv^*}^{{\rm WC}} := \max_{\alv\in\vol(\alv^{*}, r)}|\fsur_{\expansionorder}(\alv) - {\LC}(\alv)|
\end{equation}

The following proposition (proved below) tells how the worst-case error scales with respect to the surrogate order.
\begin{proposition}[Worst-case error of the surrogate]\label{th:wors_case_error_app}
Consider an expectation function $\LC(\alv)$ as defined in
Eq.~\eqref{eq:loss_app} with $\norm{O}_{\infty} \in\OC(1)$. Furthermore, consider the hypercube region $\vol(\alv^*,r)$ around some initial fixed point $\alv^*$ with the perturbation $r$ as defined in Eq.~\eqref{eq:hypercube-region-appendix} such that 
\begin{align}
    r \in \OC\left( \frac{1}{\nparams}\right) \;.
\end{align}
Assume that all the derivatives (up to order $k=\expansionorder+1$) are bounded as
\begin{align}\label{eq:assumption_gamma_wc_1}
     |\partial_{i_1, i_2, ... i_k}\LC(\alv)| \leq  \gamma^k\norm{O}_\infty \;\; \;,\;\; \forall \alv \in \vol(\alv^*, r)\;,
\end{align}
with $\gamma \in \OC(1)$ (i.e., Assumption~\ref{assumption:bounded-derivatives}).

Then, we have that $\fsur_{\expansionorder}(\alv)$ according to Eq.~\eqref{eq:surrogate_taylor}, with the $\expansionorder$ order of the truncation, is a classical surrogate of $\LC(\alv)$ over this regime with a worst-case error bounded by
\begin{align}
     \epsilon_{\expansionorder,\alv^*}^{{\rm WC}} \leq \frac{\norm{O}_\infty}{(\expansionorder+1)!}(\gamma r {\nparams})^{\expansionorder+1} \;.
\end{align}
In particular, an arbitrarily small but constant error $  \epsilon_{\expansionorder,\alv^*}^{{\rm WC}} \in \OC(1)$ can be achieved with $\expansionorder \in \OC(1)$.
\end{proposition}

\subsubsection{Mean square error}

The mean square error (MSE) of the $\expansionorder^{\rm th}$ order surrogate associated with an uniform distribution $\uni(\alv^*,r)$ over the hypercube region $\vol(\alv^*,r)$ is defined as
\begin{equation}
    \left(\epsilon_{\expansionorder,\alv^*}^{{\rm MSE}}\right)^2 = \Ebb_{\alv \sim \uni(\alv^*,r)}\bigg[\left(\fsur_{\expansionorder}(\alv) - {\LC}(\alv)\right)^2\bigg]\,.
\end{equation}

The next proposition (proved below) analytically shows the MSE scaling with respect to the surrogate order.
\begin{proposition}[Mean Square Error of the surrogate]\label{th:mse_surrogated}
Consider an expectation function $\LC(\alv)$ as defined in Eq.~\eqref{eq:loss_app} with $\norm{O}_{\infty} \in\OC(1)$.
Furthermore, consider the hypercube region $\vol(\alv^*,r)$ around some initial fixed point $\alv^*$ with the perturbation $r$ as defined in Eq.~\eqref{eq:hypercube-region-appendix} such that 
\begin{align}
    r \in \OC\left( \frac{1}{\sqrt{\nparams}}\right) \;.
\end{align}
Assume that all the derivatives (of order $k>\expansionorder$) are bounded as (i.e., Assumption~\ref{assumption:bounded-derivatives})
\begin{align}\label{eq:assumption_gamma_wc_2}
     |\partial_{i_1, i_2, ... i_k}\LC(\alv)| \leq  \gamma^k\norm{O}_\infty \;\; \;,\;\; \forall \alv \in \vol(\alv^*, r)\;,
\end{align}
with $\gamma \in \OC(1)$.

Then, we have that $\fsur_{\expansionorder}(\alv)$ according to Eq.~\eqref{eq:surrogate_taylor}, with the $\expansionorder$ order of the truncation, is a classical surrogate of $\LC(\alv)$ over this regime with the MSE associated with the uniform distribution $\uni(\alv^*,r)$ bounded by
\begin{align}
     \epsilon_{\expansionorder,\alv^*}^{{\rm MSE}} & \leq \left( \frac{2 \gamma^2 {\nparams} r^2}{3} \right)^{\frac{\expansionorder+1}{2}}\frac{\norm{O}_\infty}{\sqrt{(\expansionorder +1 )!}} e^{\frac{\gamma^2 {\nparams} r^2}{3} } \;.
\end{align}
In particular, for an arbitrarily small but constant error $  \epsilon_{\expansionorder,\alv^*}^{{\rm MSE}}  \in \OC(1)$ can be achieved for $\expansionorder \in \OC(1)$.
\end{proposition}

\subsubsection{Scaling of the number of expectation function evaluations with truncation order}
The classical surrogate can be constructed from the partial derivatives at $\alv^*$ according to Eq.~\eqref{eq:surrogate_taylor}. These partial derivatives $\partial_{i_1, ..., i_k} \LC(\alv^*)$ can be computed by evaluating the expectation function $f(\alv)$ at different parameter values. Let $N$ denote the total number of expectation function evaluations to determine the partial derivatives required to construct the $\expansionorder^{\rm th}$ order surrogate. The following proposition bounds $N$ in terms of the surrogate order. 

\begin{proposition}[Scaling of the number of terms]\label{prop:number_of_terms}
The classical surrogate can be constructed by evaluating the expectation functions at $N$ parameter settings where,
\begin{align}
    N \leq N_{d} \left[\frac{ e({\nparams}+\expansionorder-1)}{\expansionorder}\right]^\expansionorder \;,
\end{align}
with $N_{d} $ an upper bound on the number of expectation function evaluations required to compute each partial derivative (up to $\expansionorder+1$ order).
\end{proposition}

 The way we evaluate the expectation function to compute the partial derivatives has not been specified yet. In some certain cases, the expectation function evaluations might be computable classically. 
However, in general we will need access to quantum computers to estimate expectation function values and in turn the partial derivatives. Since only statistical estimates of the expectation function can be obtained from quantum computers through measurement shots, this leads to an additional error for the surrogate due to the statistical fluctuation.

\subsubsection{General surrogation guarantee}
\label{sec:general-surrogate-guarantee}
The following \supt~presents the final surrogatability guarantee taking into account the presence of shot noise. Here we assume that the expectation function can be estimated on a quantum computer by measuring in the eigenbasis of the observable $O$. That is, given an eigendecomposition $O = \sum_{j} \lambda_j |\lambda_j \rangle\langle \lambda_j |$ where $\lambda_j$ is an eigenvalue with an associated eigenstate $|\lambda_j \rangle$, an unbiased estimator of the expectation function after $N_s$ measurement shots can be expressed as
\begin{align}\label{eq:taylor-loss-statistical-estimate}
    \widehat{f}(\alv) = \frac{1}{N_s} \sum_{i=1}^{N_s} \hat{o}_i \;,  
\end{align}
where $\hat{o}_i$ is an outcome of $i^{\rm th}$ measurement which takes a value $\lambda_j$ with a probability $p_j = \Tr[\EC_{\alv}(\rho) |\lambda_j \rangle\langle \lambda_j |]$. For observables of complex many-body systems it may not be straightforward to measure in the eigenbasis of $O$. In this case, it is instead natural to expand $O$ as a sum of `easier' operators to measure (e.g., a sum of Paulis). While we do not explicitly present this case here, the generalization is straightforward. 

\begin{theoremsup}[General Surrogatability Guarantee, Formal]\label{coro:surrogate_quantum_quantum}

Consider an expectation function $\LC(\alv)$ as defined in Eq.~\eqref{eq:loss_app} with $\norm{O}_\infty\in\OC(1)$. Assume that the derivatives are bounded as
\begin{align}\label{eq:assumption_gamma_wc}
     |\partial_{i_1, i_2, ... i_k}\LC(\alv)| \leq  \norm{O}_\infty\gamma^k \;\; \;,\;\; \forall \alv \in \vol(\alv^*, r)\;,
\end{align}
with $\gamma \in \OC(1)$ (i.e., Assumption~\ref{assumption:bounded-derivatives}). We additionally assume that each partial derivative (up to the order $\expansionorder$) can be written as a linear combination of at most $N_{d} $ expectation function evaluations such that
\begin{align}
    N_{d} \in \OC\left({\nparams}^\expansionorder\right) \;
\end{align}
(i.e., Assumption~\ref{assumption:bounded-loss-evaluations}) and where the biggest coefficient in absolute value scales at most 
 as $b_0\in \OC(\poly(m))$ (i.e., Assumption~\ref{assumption:loss-evaluations-expectations}). Moreover, we assume that each of these expectation function evaluations can be estimated by measuring in the eigenbasis on a quantum computer as in Eq.~\eqref{eq:taylor-loss-statistical-estimate}. 
\begin{itemize}
    \item Worst-case error: when the perturbation $r$ scales as
    \begin{align}
        r \in\OC\left( \frac{1}{\nparams} \right)\;,
    \end{align}
    with a number of measurement shots $N_s$ and probability at least $1-\delta$ (for any $\delta \in \Theta(1)$) we can construct $\widehat{\fsur_\expansionorder}(\alv)$ such that,
    \begin{align}
        \left|\widehat{\fsur_\expansionorder}(\alv) - f(\alv) \right| &\leq \norm{O}_\infty\left(b_0N_d e^{r\nparams}  \sqrt{\frac{2 \log(2/\delta)}{N_s}}  + \frac{( \gamma r {\nparams})^{\expansionorder+1}}{(\expansionorder+1)!} \right)\;. 
    \end{align}
    Thus the worst-case error can be reduced to an arbitrarily small constant,  $\left|\widehat{\fsur_\expansionorder}(\alv) - f(\alv) \right| \in \Theta(1)$ by increasing the order $\expansionorder$ with $\expansionorder \in \OC(1)$ and the number of shots with $N_s \in \OC(\poly(\nparams))$.    
    \item Average case error: when the perturbation $r$ scales as
    \begin{align}
        r \in\OC\left(\frac{1}{\sqrt{{\nparams}} }\right) \;,
    \end{align}    
      with $N_s$ measurement shots, we can construct $\widehat{\fsur_\expansionorder}(\alv)$ such that,
    \begin{align}
        \label{eq:total-error-taylor-surrogate-mse-quantum} 
    \Ebb_{\vec{\alpha} ,\MC}\left[\left(\widehat{\fsur_{\expansionorder}}(\alv)-f(\alv)\right)^2\right] &\leq 2\left(\frac{b_0 N_d(\expansionorder+1)(r\nparams)^{\expansionorder}\norm{O}_\infty}{\sqrt{N_s} \expansionorder !}\right)^2 + 2\left( \frac{2 \gamma^2 {\nparams} r^2}{3} \right)^{\expansionorder+1}\frac{\norm{O}_\infty^2}{(\expansionorder +1 )!} e^{\frac{ 2 \gamma^2 {\nparams} r^2}{3} }\;,
    \end{align}
    Thus the mean square error can be arbitrarily reduced to any constant in $\Theta(1)$ by increasing the order $\expansionorder$ with $\expansionorder \in \OC(1)$ and the number of shots $N_s$ with  $N_s \in \OC(\poly(\nparams))$.  Then by applying Markov's inequality~\ref{lemma:markov} we find the informal version of the theorem.
\end{itemize}
\end{theoremsup}
\supt~\ref{coro:surrogate_quantum_quantum} is the formal version of Theorem~\ref{thm:surrogate-average-general} from the main text. Indeed, we can set both the truncation and the empirical mean square errors to be $(\epsilon/2)^2$. We have seen that with $r \in\OC\left( \frac{1}{\sqrt{{\nparams}}}\right)$, the mean square truncation error can be made arbitrarily small by increasing $\expansionorder\in\OC(1)$. Then, equating both errors from the expression in Eq.~\eqref{eq:total-error-taylor-surrogate-mse-quantum} leads to
\begin{equation}\label{eq:Nshots-taylor-general}
    N_s \in \OC(\nparams^\expansionorder b_0^2 N_d^2)=\OC(\poly(\nparams))\;,
\end{equation}
where we used Assumptions~\ref{assumption:bounded-loss-evaluations} and~\ref{assumption:loss-evaluations-expectations} for the scaling of both $b_0$ and $N_d$. Therefore, using polynomial resources, we can get an arbitrarily small mean square error. A similar analysis can be made the worst-case error guarantee with $r\in\OC(1/\nparams)$ to get 
\begin{equation}
    N_s\in\OC(N_d^2 b_0^2 \log(1/\delta))=\OC(\poly(\nparams))\;.
\end{equation}

We note that in the proof of \supt~\ref{coro:surrogate_quantum_quantum}, the surrogate is expressed as the linear combination of the expectation function evaluated at different parameters. In order to construct a statistical estimate of the surrogate, we individually estimate each term in the sum and also the total measurement shout budget is evenly distributed to all the terms. This strategy is sub-optimal in general since it could be the case that there are terms in the sum that commute with each other and hence they can be simultaneously estimated and/or the total shot budget could be distributed in a more informative manner. Nevertheless, we now show that even with this sub-optimal strategy one can achieve a surrogate that achieves an arbitrarily small constant error with high probability with a polynomial number of measurement shots.

\textit{Note:} In principle, the analysis can be generalized to allow for independent perturbation radii for each parameter. The intuition is simple: in the derivation of the Taylor remainder in Eq.~\eqref{eq:variance_sm}, the dominant contribution arises from the variance of the sum of the parameter variations, suggesting that the expression 
$mr^2$ could be replaced by $\sum_{i=1}^m r_i^2$ where $\delta_i\in[-r_i,r_i]$. Nevertheless, establishing this generalization with full rigor would require considerably more technical work and would add little conceptual value to the main results. For this reason, and to keep the presentation focused, we restrict the analysis to hypercubes, while noting that the extension follows the same intuition.

\subsubsection{Guarantee that our assumptions are met by common parameterized quantum circuits}
\label{sec:guarantee-assumptions-theorem1}
The following proposition establishes that a broad family of parameterized quantum circuits (indeed, essentially all currently considered by the community currently with independent parameters between circuit layers) satisfy the assumptions of our bounds. 

\begin{proposition}[Bounded derivatives for expectation value with parametrized unitary channel] \label{prop:derivatives-upperbound-unitary}
Let us assume a expectation function of the form of Eq.~\eqref{eq:expectation} i.e. $f(\alv)=\Tr[\rho U^\dagger(\alv) O U(\alv)]$, where $U(\alv)=\prod_{l=1}^\nparams V_l U_l(\alpha_l)$ such that $V_l$ are arbitrary non-parametrized unitaries and each $U_l(\alpha_l)=\exp(-iH_l\alpha_l)$ is a parametrised gate with generator $H_l$ that is hermitian. Then, the $k$-th order derivatives of the form $\frac{\partial^k f(\alv)}{\prod_{l=1}^\nparams \partial\alpha_l^{k_l}}:=\partial^{\vec{k}}f(\alv)$ where $\sum_{l=1}^\nparams k_l=k$ and  $k_l\in\mathbb{N}$ for each $l$ are bounded by
\begin{align}
    |\partial^{\vec{k}}f(\alv)| &\leq \norm{O}_\infty \prod_{l=1}^\nparams (2\norm{H_l}_\infty)^{k_l} \\
    &\leq \norm{O}_\infty \left(2\max_{l}\norm{H_l}_\infty\right)^k
\end{align}

\end{proposition}

\begin{corollary}[Surrogation Guarantee for Unitary Channel]\label{cor:surrogate-taylor-unitary}
Consider a generic expectation function $f(\vec{\alpha})$ as defined in Eq.~\eqref{eq:expectation} with ${\nparams}$ independent parameters such that the spectral norm of each generators as well as the initial observable one are $\OC(1)$, and assume that (up to) order-$\expansionorder$ partial derivatives can be estimated efficiently with quantum computers. 
It is possible to efficiently classically surrogate $f(\vec{\alpha})$ over any uniformly sampled hypercube $\uni(\vec{\alpha}^*, r)$ around an arbitrary point $\vec{\alpha}^*$, with
\begin{align}
     r \in \OC \left( \frac{1}{\sqrt{\nparams}}\right) \; .
\end{align}
\end{corollary}

In our general surrogate guarantee, Assumptions~\ref{assumption:bounded-loss-evaluations} and~\ref{assumption:loss-evaluations-expectations} can be argued from a practical perspective. Let us clarify the scope of this theorem by first considering the following simple example. Assume a circuit composed by uncorrelated Pauli rotations and any non-parametrized gates (not restricted to near-Clifford here). In this case, the parameter-shift rule, which satisfies both Assumptions~\ref{assumption:bounded-loss-evaluations} and~\ref{assumption:loss-evaluations-expectations}, can be readily applied. Indeed, any order of derivatives with respect to a single parameters involves two loss evaluations with coefficients magnitude $1/2$. From this, we can directly state that $b_0\in\OC(1)$. For $k$-th  order partial derivatives, the largest number of loss evaluations correspond to derivatives with respect to $k$ distinct parameters which leads to up-to $2^k$ loss evaluations in total (with coefficients magnitude $1/2^k$)~\cite{cerezo2020impact}. This means that truncating at order $\kappa$ involves at most $N_d\in\OC(2^\kappa)$ loss evaluation for each partial derivatives considered in the surrogate. Therefore, both assumptions are satisfied. For this case, we provide a numerical study in Appendix~\ref{sec:iterative-training-taylor-surrogate} where we use Taylor surrogate iteratively to approximate a ground state.

The standard parameter-shift rule applies to gates whose generators have a two distinct eigenvalues, including single-qubit rotations and multi-qubit Pauli rotations. More generally, we can use the generalized parameter shift rule~ \cite{wierichs2022general} to compute derivatives of generators with more than two disctinct eigenvalues. One example are Givens rotations whose generators are sums of commuting Pauli operators. These require four circuit evaluations.

\subsubsection{How the polynomial guarantees break down for larger patch sizes}
\label{sec:extension-larger-patch}
For completeness, we note that if we want the resource requirements to scale at most polynomially with the system size and number of variational parameters, the only region for which we~\footnote{We do not rule out that others using a different proof strategy could not prove guarantees in a larger region~\cite{cerezo2020impact}, although we think this unlikely.} could prove generic guarantees for is the aforementioned $r\in\OC(1/\sqrt{m})$. Indeed, this ensures that the terms we require to compute and the MSE $\epsilon$ remain constant. 

We can see this as follows. We start by recalling  that the number of terms we need to compute, given by Proposition~\ref{prop:number_of_terms}, scales as 
\begin{equation}\label{eq:resources_beyond}
    N\in\OC(m^{w\kappa})\,,
\end{equation}
where $\kappa$ is the truncation order and $w\in[1,2]$ is some constant, and the MSE given in Proposition~\ref{th:mse_surrogated} is proportional to
\begin{equation}
    \epsilon^2\leq \frac{\left(c \nparams r^2\right)}{(\kappa+1)!}^{\kappa+1}e^{c \nparams r^2}\;,
\end{equation}
where $c=2\gamma^2/3\in\OC(1)$ is a constant. Next we will proceed to show what happens when we want to extend beyond the $r\in\OC(1/\sqrt{m})$.

Let us first simplify the MSE  expression using Stirling approximation for the factorial: \begin{equation}
    k!\geq \sqrt{2\pi k}\left(\frac{k}{e}\right)^k\geq \left(\frac{k}{e}\right)^k\;.
\end{equation}
Therefore, the MSE $\epsilon$ can be bounded as follows
\begin{equation}
    \epsilon^2\leq \left(\frac{ce \nparams r^2}{\kappa+1}\right)^{\kappa+1}e^{c \nparams r^2}\;.
\end{equation}

 From the previous equation, we see that in order to make the error bound arbitrarily small, we would necessarily need $cemr^2<\kappa+1$. Let us fix $\kappa+1=Kcemr^2\in\OC(mr^2)$ where $K>1$ is a constant. Then we can rewrite the upper bound in the MSE from the previous equation as
 \begin{equation}
     \epsilon^2\leq \left(\frac{1}{K}\right)^{Kecmr^2}e^{cmr^2}=\left(\frac{e}{K^{Ke}}\right)^{cmr^2}\;.
 \end{equation}

Since we are interested in extending $r$ to larger region (beyond $\order{1/\sqrt{m}}$), we can assume $r\in\OC(1/m^\beta)$ with $0\leq \beta < 1/2$. In this case, we would have $\kappa\in\OC(mr^2)=\OC(m^{1-2\beta})$ where $1-2\beta>0$. Substituting back the relation between $\kappa$ and $m$ into the number of resources required in Eq.~\eqref{eq:resources_beyond}, the count scales as
\begin{equation}
    N\in\OC\left(m^{w m^{1-2\beta}}\right)\;,
\end{equation}
which scales super-exponentially for $\beta=0$, where $r\in\order{1}$. 

More positively, we can reduce to \textit{sub-exponential} (albeit still {super-polynomial}) if $0<\beta<1/2$. Note that if $r$ decreases logarithmically with $m$, the resources also scales \textit{super-polynomially}. \\

\paragraph*{Pauli Propagation.}
 The previous observations can be equivalently shown for Pauli propagation truncation case (see Appendix~\ref{app:smallanglePP}).
In this case, the maximal resource requirement (see Eq.~\eqref{eq:number-paths}) is proportional to
\begin{equation}\label{eq:worst-case-resource-Pauli-extension}
    N\leq\left(\frac{em}{\kappa}\right)^\kappa\;,
\end{equation}
whereas the error will be (see Supplemental Theorem~\ref{thm:mse-avg-supp}) of the following form 
\begin{equation}
    \epsilon^2\leq\left(\frac{emr^2}{\kappa}\right)^\kappa\;.
\end{equation}
To make the error arbitrarily small, we would need  at least $\kappa\in\Omega(mr^2)$. So assuming $r\in\OC(1/m^{\beta})$ with $0\leq \beta <1/2$, we have $\kappa\in\OC(m^{1-2\beta})$.
Since the resource scales exponentially with $\kappa\in\OC(m^{1-2\beta})$, it will be exponential with $m$ for the case $\beta=0$, i.e., constant $r$ and sub-exponential, but super-polynomial for $0<\beta<1/2$. Similarly, if $\kappa$ decreases logarithmically with $m$, the resource scaling will be super-polynomial (but sub-exponential in general).
Obviously, the case $r\sim 1$ makes no sense here since sine can be larger than cosine. Also, setting $\kappa=m$, means that we are not truncating anymore. 

Within the Pauli-propagation framework, the required number of terms is inherently problem-dependent. While the worst-case complexity is super-polynomial, Pauli propagation can still be practical in specific settings. Eq.~\eqref{eq:worst-case-resource-Pauli-extension} should thus be interpreted as an upper bound arising from a maximal-splitting scenario. 

\subsection{Absolute error: Proof of Proposition~\ref{th:wors_case_error_app}}\label{sec:absolute_error}\label{app:generalsurrogatetheorem-3}

\begin{proof} The worst-case error of the surrogate  $\fsur_{\expansionorder}(\alv)$ of the order $\expansionorder$ is written as
\begin{equation}
    \epsilon_{\expansionorder,\alv^*}^{{\rm WC}} := \max_{\alv\in\vol(\alv^{*}, r)}|\fsur_{\expansionorder}(\alv) - {\LC}(\alv)|\,,
\end{equation}
where $\LC_{\alv}$ is the expectation function, and $\vol(\alv^{*}, r)$ is the hypercube in the parameter space such that
\begin{align}
    \alpha_i \in [\alpha^*_i - r, \alpha^*_i + r ]\,.
\end{align}

By invoking the Taylor Reminder Theorem (i.e., \supt~\ref{thm:reminder_taylor}), we see that this is equivalent to bounding the reminder of the Taylor expansion. In particular, by denoting the parameters corresponding to the worst case error
\begin{align}
\alv' = \argmax_{\alv\in\vol(\alv^*, r)}|\fsur_{\expansionorder}(\alv) - {\LC}(\alv)| \;,    
\end{align}
 we can use \supt~\ref{thm:reminder_taylor} to explicitly write the worst-case error as
\begin{align}
      \epsilon_{\expansionorder,\alv^*}^{{\rm WC}} & = |\fsur_{\expansionorder}(\alv') - {\LC}(\alv')| \\
      & = \left| \frac{1}{(\expansionorder+1)!}\sum_{i_1, \dots , i_{\expansionorder+1}=1}^{\nparams} \Delta_{i_1, \dots , i_{\expansionorder+1}}(\alv', \alv^*)\partial_{i_1, ... i_{\expansionorder+1}}\LC(\bm{\vartheta}) \right| \;,
\end{align}
where $\bm{\vartheta} = \tau \alv' + (1-\tau)\alv^*$ for some $\tau\in[0,1]$.

By applying the triangle inequality we can further simplify this expression to obtain
\begin{align}
   \epsilon_{\expansionorder,\alv^*}^{{\rm WC}} & = \left| \frac{1}{(\expansionorder+1)!}\sum_{i_1, \dots , i_{\expansionorder+1}=1}^{\nparams} \Delta_{i_1, \dots , i_{\expansionorder+1}}(\alv', \alv^*)\partial_{i_1, ... i_{\expansionorder+1}}\LC(\bm{\vartheta}) \right|\\
    & \leq \frac{1}{(\expansionorder+1)!}\sum_{i_1, \dots , i_{\expansionorder+1}=1}^{\nparams} \left|\Delta_{i_1, \dots , i_{\expansionorder+1}}(\alv', \alv^*)\partial_{i_1, ... i_{\expansionorder+1}}\LC(\bm{\vartheta}) \right|\\
    & =  \frac{1}{(\expansionorder+1)!}\sum_{i_1, \dots , i_{\expansionorder+1}=1}^{\nparams} \left|\Delta_{i_1, \dots , i_{\expansionorder+1}}(\alv', \alv^*)\right|\left|\partial_{i_1, ... i_{\expansionorder+1}}\LC(\bm{\vartheta}) \right|\label{eq:wc_bound_aftertriangle}
\end{align}
where in the last equality we use that the absolute value of the product is the product of absolute values, i.e., $|x y| = |x||y|$.

Next, let us  focus on the terms $\partial_{i_1, ... i_{\expansionorder+1}}\LC(\bm{\vartheta}) $. Since the assumption in Eq~\eqref{eq:assumption_gamma_wc_1} of Proposition~\ref{th:wors_case_error_app} holds, the derivatives in the previous sum are bounded as
\begin{align}
     |\partial_{i_1, i_2, ... i_k}\LC(\alv)| \leq  \norm{O}_\infty\gamma^{k} \;\; \;,\;\; \forall \alv \in \vol(\alv^*, r)\;.
\end{align}
Combining the previous result with  Eq.~\eqref{eq:wc_bound_aftertriangle}, we can simplify the bound as
\begin{align}
\epsilon_{\expansionorder,\alv^*}^{{\rm WC}} & \leq \frac{\norm{O}_\infty\gamma^{\expansionorder+1}}{(\expansionorder+1)!}\sum_{i_1, \dots , i_{\expansionorder+1}=1}^{\nparams} \left|\Delta_{i_1, \dots , i_{\expansionorder+1}}(\alv, \alv^*)\right|\,.
\end{align}

Next we can recall the definition of $\Delta_{i_1, \dots , i_{\expansionorder+1}=1}(\alv', \alv^*)$ in Eq.~\eqref{eq:def_deltath}
\begin{align}
    \Delta_{i_1, \dots , i_{\expansionorder+1}}(\alv', \alv^*) & = (\alpha'_{i_1} - \alpha^*_{i_1})(\alpha'_{i_2} - \alpha^*_{i_2}) ... (\alpha'_{i_{\expansionorder+1}} - \alpha^*_{i_{\expansionorder+1}}) \\
    & = \delta'_{i_1} \delta'_{i_2} ... \delta'_{i_{\expansionorder+1}}  \;,
\end{align}
where we denote $\delta'_{k} = \alpha'_k - \alpha^*_k$. Therefore, we have
\begin{align}
    \epsilon_{\expansionorder,\alv^*}^{{\rm WC}}   &\leq  \frac{\norm{O}_\infty\gamma^{\expansionorder+1}}{(\expansionorder+1)!}\sum_{i_1, \dots , i_{\expansionorder+1}=1}^{\nparams} \left|\delta'_{i_1}...\delta'_{i_{\expansionorder+1}}\right|\\
    &\leq \frac{\norm{O}_\infty\gamma^{\expansionorder+1}}{(\expansionorder+1)!}\sum_{i_1, \dots , i_{\expansionorder+1}=1}^{\nparams} r^{\expansionorder+1}\\
    &= \frac{\norm{O}_\infty}{(\expansionorder+1)!}(\gamma r {\nparams})^{\expansionorder+1}\,,
\end{align}
where in the second inequality we used that $\max_{\alpha_i} |\alpha_i-\alpha_i^*| = r$ for all the $\alpha_i$ and to reach the last line we use $\sum_{i=1}^{\nparams} 1 = {\nparams}$ for all indices. Notice that $r\in\OC(1/\nparams)$, $\gamma\in\OC(1)$ and $\norm{O}_\infty\in\OC(1)$ implies $\gamma r\nparams\in\OC(1)$ and therefore, the error is $\OC(1)$ and decreases with $\expansionorder$, which completes the proof of the proposition.
\end{proof}

\subsection{Mean-square error: Proof of Proposition~\ref{th:mse_surrogated}} \label{appx:surrogate-taylor-average}\label{app:generalsurrogatetheorem-4}

\begin{proof}
The MSE is of the form
    \begin{align}
        \left(\epsilon_{\expansionorder,\alv^*}^{{\rm MSE}}\right)^2 & = \Ebb_{\alv \sim \uni(\alv^*, r_{\nparams}) }\bigg[\left(\fsur_{\expansionorder}(\alv) - {\LC}(\alv)\right)^2\bigg] \\
        & = \Ebb_{\alv \sim \uni(\alv^*, r_{\nparams}) }\left[\left(\sum_{k = \expansionorder+1}^{\infty}\frac{1}{k!}\sum_{i_1, \dots , i_k=1}^{\nparams} \Delta_{i_1, \dots , i_k}(\alv, \alv^*)\partial_{i_1, ... i_k}\LC(\alv^*)\right)^2\right]
    \end{align}
where $\uni(\alv^*, r_{\nparams})$ is a uniform distribution of $\alv$ in the hypercube region, i.e., $\vol(\alv^*, r_{\nparams}) = \{ \alv \}$ such that $\alpha_i \in [\alpha_i^* -r_{\nparams} , \alpha_i^* + r_{\nparams} ] $ for all $i$. Note that we explicitly add the subscript ${\nparams}$ in $r_{\nparams}$ indicating that it is associated with $\alv$ with ${\nparams}$ parameters. This will come in handy later in our proof. To reach the second equality, we use the definition of the Taylor expansion on the expectation function in Eq.~\eqref{eq:taylor_loss} and the definition of the surrogate in Eq.~\eqref{eq:surrogate_taylor} to see that $\fsur_{\expansionorder}(\alv) - {\LC}(\alv)$ is an infinite sum of terms.

To further proceed, we introduce a perturbation vector $\dlv := \alv - \alv^*$, whose $i$-th components is  $\delta_i = \alpha_i - \alpha^*_i$ for any index $i$. Then, we also denote  the average of $\dlv$ over the perturbed region $\Ebb_{\dlv} [\cdot] := \Ebb_{\dlv \sim \uni(\vec{0},r_{\nparams})} [\cdot]$ where $\uni(\vec{0},r_{\nparams})$ is a uniform distribution over the region $\vol(\vec{0},r_{\nparams}) = \{ \dlv \}$ with $\delta_i \in [-r_{\nparams}, r_{\nparams}]$. More specifically,  $\uni(\vec{0},r_{\nparams})$ is obtained by sampled independently each component $\delta_i$ over $[-r_{\nparams}, r_{\nparams}]$. Then, the MSE can be written as
\begin{align}
   \left(\epsilon_{\expansionorder,\alv^*}^{{\rm MSE}} \right)^2& =\Ebb_{\alv \sim  \uni(\alv^*, r)}\Bigg( \sum_{k = \expansionorder+1}^{\infty}\frac{1}{k!} \sum_{i_1, \dots , i_k=1}^{\nparams}  \Delta_{i_1, \dots , i_k}(\alv, \alv^*)\partial_{i_1, ... i_k}\LC(\alv^*)\Bigg)^2 \\
    & = \sum_{k,l = \expansionorder+1}^{\infty}\frac{1}{l! k!}\sum_{i_1, \dots , i_k=1}^{\nparams} \sum_{j_1, \dots , j_l=1}^{\nparams} \Ebb_{\alv \sim \uni(\alv^*, r)}\Bigg[\Delta_{i_1, \dots , i_k}(\alv, \alv^*) \Delta_{j_1, \dots , j_l}(\alv, \alv^*)\Bigg]\partial_{i_1, ... i_k}\LC(\alv^*)\partial_{j_1, ... j_k}\LC(\alv^*)\\
    & \leq \sum_{k,l = \expansionorder+1}^{\infty}\frac{1}{l! k!}\sum_{i_1, \dots , i_k=1}^{\nparams} \sum_{j_1, \dots , j_l=1}^{\nparams} \Ebb_{\dlv}\left[ \delta_{i_1} ... \delta_{i_k} \delta_{j_1} ... \delta_{j_l} \right]\norm{O}_\infty^2\gamma^{k+l}\label{eq:put_back_assum} \\
    & =  \norm{O}_\infty^2\sum_{k,l = \expansionorder+1}^{\infty}\frac{\gamma^{k+l}}{l! k!} \Ebb_{\dlv}\left[ \sum_{i_1=1}^{\nparams} \delta_{i_1} \cdots \sum_{i_k = 1}^{\nparams} \delta_{i_k} \sum_{j_1=1}^{\nparams} \delta_{j_1} \cdots \sum_{j_l=1}^{\nparams} \delta_{j_l}\right] \\
    & = \norm{O}_\infty^2\sum_{k,l = \expansionorder+1}^{\infty}\frac{\gamma^{k+l}}{l! k!} \Ebb_{\dlv}\left[ \left( \sum_{i=1}^{\nparams} \delta_{i}\right)^{k+l}\right] \label{eq:contraction_sums}
\end{align}
where the inequality is by the assumption of the bounded derivatives as specified in Eq.~\eqref{eq:assumption_gamma_wc_2},  
\begin{align}
     |\partial_{i_1, i_2, ... i_k}\LC(\alv)| \leq  \norm{O}_\infty\gamma^k \;\; \;,\;\; \forall \alv \in \vol(\alv^*, r_{\nparams})\;, \label{eq:assumption_gamma_wc_proof_mse}
\end{align}
together with the fact that the averages of the perturbations are always non-negative due to the symmetry in the perturbations. In particular, since $\delta_k \in [-r_{\nparams},r_{\nparams}]$, we have that the uniform average results in $\Ebb_{\delta_k}[\delta^{2a+1}_k] = 0$ and $\Ebb_{\delta_k}[\delta^{2a}_k] \geq 0$ for any integer $a \in \mathbb{N}$. Hence, this justified Eq.~\eqref{eq:put_back_assum}. To reach the last equality, we simply massage the expression by pulling out the average and re-writing the sums.

From here onwards we focus on bounding the average term $\Ebb_{\dlv}\left[\left( \sum_{i=1}^{\nparams} \delta_i \right)^{k+l}\right]$. Particularly, one of our key ingredients is the central limit theorem. To see this, we begin with denoting
\begin{align}
    S_{\nparams} = \sum_{i=1}^{\nparams} \delta_i \;,
\end{align}
together with its variance 
\begin{align}\label{eq:variance_sm}
    \sigma_{\nparams}^2 := \Var_{\dlv}[S_{\nparams}] & = \frac{\nparams r_{\nparams}^2}{3} \;,
\end{align}
where we use the fact that the variables $\{ \delta_i \}$ are independent and $\Var_{\delta_i} [\delta_i] = r^2_{\nparams}/3$ for a uniform distribution. Furthermore, we can define a new random variable 
\begin{align}
    X_{\nparams} = \frac{S_{\nparams}}{\sigma_{\nparams}} \;,
\end{align}
together with its variance which is equal to $1$ i.e.,  
\begin{align}
    \Var_{\dlv} [X_{\nparams}] = 1 \;.
\end{align}

The central limit theorem can now be applied to $X_{\nparams}$ (see \supt~\ref{theorem:clt}), which essentially states that when ${\nparams}$ approaches infinity, the probability distribution of the random variable $X_{\nparams}$ approaches a normal distribution with the zero mean and  variance equal to one, i.e., $\mathcal{N}(0,1)$. As stated in Appendix~\ref{sec:central_limit_th}, the moments of a random variable under $\mathcal{N}(0,1)$ are 
\begin{equation}\label{eq:highmoments}
    \Ebb_{\dlv}[X^a_{{\nparams}\to\infty}] = 
    \begin{cases}
        0 &{\,\, \rm if}\,\, a \,\, {\rm is} \,\,{\rm odd}\\
        (a-1)!! &\,\, { \rm if} \,\, a \,\, {\rm is}\,\,  {\rm even}
    \end{cases}
\end{equation}
for any integer $a\in\mathbb{N}$ and $x!!$ being the double factorial function, defined for $x\in \mathbb{N}$ as $x!! = x(x-2)(x-4)(x-6)...1$. 

Next we will upper-bound $\Ebb[X^a_{{\nparams}}]$ with $\Ebb[X^a_{{\nparams}\to\infty}]$ for any ${\nparams}$. That is, the expected value of $X^a_{\nparams}$ tends to a maximum with increasing ${\nparams}$. To see this, first notice that $X_{{\nparams}+1}$ can be written in terms of $X_{\nparams}$ as
\begin{equation}\label{eq:xmp1}
    X_{{\nparams}+1} = \frac{S_{\nparams} + \delta_{{\nparams}+1}}{\sigma_{{\nparams}+1}}= \frac{\sigma_{\nparams}}{\sigma_{{\nparams}+1}}X_{\nparams} + \frac{\delta_{{\nparams}+1}}{\sigma_{{\nparams}+1}} \;.
\end{equation}

We can bound the $a$-th moment of $X_{{\nparams}+1}$ with the $a$-th moment of $X_{\nparams}$ as
\begin{align}
    \Ebb_{\dlv}\left[X^a_{{\nparams}+1}\right] & = \Ebb_{\dlv}\left[\left(\frac{\sigma_{\nparams}}{\sigma_{{\nparams}+1}}X_{\nparams} + \frac{\delta_{{\nparams}+1}}{\sigma_{{\nparams}+1}}  \right)^{a}\right] \\
    & =  \sum_{i=0}^a \binom{a}{i} \Ebb_{\dlv} \left[ X_{\nparams}^{i} \right]\left(\frac{\sigma_{\nparams}}{\sigma_{{\nparams}+1}}\right)^{i}\frac{\Ebb_{\delta_{{\nparams}+1}}\big[ \delta_{{\nparams}+1}^{a-i} \big]}{\sigma_{{\nparams}+1}^{a-i}} \\
    & \geq \left(\frac{\sigma_{\nparams}}{\sigma_{{\nparams}+1}}\right)^{a}\Ebb_{\dlv}\left[ X_{\nparams}^a \right] \;,
\end{align}
where in the second line we use the binomial expansion $(x+y)^a = \sum_{i=0}^a \binom{a}{i} x^i y^{a -i}$ and the linearity of the expectation, while in the inequality we dropped all the terms (which are all non-negative) except for $i = a$.

To proceed, we consider the scenario where $\frac{\sigma_{\nparams}}{\sigma_{{\nparams}+1}}\geq 1$ which implies that
\begin{align}\label{eq:bounding-moments}
    \Ebb_{\dlv}\left[X^a_{{\nparams}+1}\right] & \geq \Ebb_{\dlv}\left[X^a_{{\nparams}}\right] \;.
\end{align}

To determine when $\frac{\sigma_{\nparams}}{\sigma_{{\nparams}+1}}\geq 1$, we further consider $r_{\nparams}$ to be of the form $r_{\nparams} = c\nparams^{x}$ with some constant $c$. Then, we can see that $\frac{\sigma_{\nparams}}{\sigma_{{\nparams}+1}}\geq 1$ is valid only if $x \leq -1/2$. More precisely, by using Eq.~\eqref{eq:variance_sm} we have
\begin{equation}\label{eq:condition_r}
    \frac{\sigma_{\nparams}}{\sigma_{{\nparams}+1}} = \left(\frac{{\nparams}}{{\nparams}+1}\right)^{2x +1}\geq 1 \;\; {\rm iff} \;\; x\leq-\frac{1}{2} \;.
\end{equation}
Hence, from here onwards, we consider that 
\begin{align}
    r_{\nparams} \in \OC\left( \frac{1}{\sqrt{{\nparams}}}\right) \;,
\end{align}
justifying Eq.~\eqref{eq:bounding-moments} and leading to
\begin{equation}\label{eq:bound_clt}
    \Ebb_{\dlv}[X^a_{{\nparams}}]\leq \Ebb_{\dlv}[ X_{{\nparams}+1}^a ]\leq \dots \leq \Ebb_{\dlv}[ X_{{\nparams}\to \infty}^a ] = (a-1)!! \;\;\;\; {\rm if} \,\,  a \,\,  {\rm is}\,\,  {\rm even} \;,
\end{equation}
where in the last equality we used Eq.~\eqref{eq:highmoments}. We remind that only even $a$ terms do not vanish, since for odd $a$ we have $\Ebb_{\dlv}[X^a_{{\nparams}}] = 0$. 

We are now ready to revisit Eq.~\eqref{eq:contraction_sums} by applying the previously obtained bound. In particular, we have
\begin{align}
    \left( \epsilon_{\expansionorder,\alv^*}^{{\rm MSE}}\right)^2 &\leq \norm{O}_\infty^2\sum_{k,l = \expansionorder+1}^\infty\frac{\gamma^{k+l}}{k!l!}\Ebb_{\dlv}\left[\left( \sum_{i=1}^{\nparams} \delta_i \right)^{k+l}\right] \\
    &=\norm{O}_\infty^2\sum_{k,l= \expansionorder+1}^\infty\frac{\gamma^{k+l}}{k!l!}\sigma_{\nparams}^{k+l}\Ebb_{\dlv}\left[X_{\nparams}^{k+l}\right]\\
    & \leq \norm{O}_\infty^2\sum_{k,l = \expansionorder+1}^\infty\frac{(\gamma\sigma_{\nparams})^{k+l}}{k!l!}\Ebb_{\dlv}[X_{{\nparams}\to\infty}^{k+l}] \\
    & = \norm{O}_\infty^2\sum_{\substack{k,l = \expansionorder+1 \\ ;   k+l \,\, {\rm even}}}^\infty \frac{(\gamma\sigma_{\nparams})^{k+l}}{k!l!} (k + l -1)!! \,,\label{eq:boundeps_withclt}
\end{align}
where in the fist equality we used that $S_{\nparams} = X_{\nparams}\sigma_{\nparams}$. In the second inequality we used the bound derived in Eq.~\eqref{eq:bound_clt}.

In the next few steps, we will do a series of changes of indices, which leads to a more manageable expression. First, we denote $k =  \expansionorder+1 +i$, $l = \expansionorder+1+j$ and explicitly add a condition to force that $i+j$ is even. Particularly, we have
\begin{equation}
    \sum_{\substack{k,l = \expansionorder+1 \\ ;   k+l \,\, {\rm even}}}^\infty \frac{(\gamma\sigma_{\nparams})^{k+l}}{k!l!} (k + l -1)!!  = \sum_{i,j = 0}^\infty \frac{(2\expansionorder + 1 +i +j)!! }{(\expansionorder+1+i)!(\expansionorder+1+j)!}(\gamma\sigma_{\nparams})^{2\expansionorder + 2 + i + j} \left(\frac{1+(-1)^{i+j}}{2} \right)\label{eq:pre-change}
\end{equation}
where the last fraction ensures $i+j$ is even.

We can now make another change in the variables that will help with the computation. More precisely, for any $h_{ij}$ we will use
\begin{equation}
    \sum_{i,j = 0}^{\infty} h_{i,j} =\sum_{y' = 0}^{\infty}\sum_{ \substack{ i,j =0 \\  ;i+j = y'}}^{y'} h_{i,j} = \sum_{y' = 0}^{\infty}\sum_{x = 0}^{y'} h_{x,y'-x} \,,
\end{equation}
where in the first equality we are adding all the combinations of $i + j$ that add up to $y'$, and then in the second equality we substitute $i=x$ and $j = y'-x$ (and thus we keep $y' = i + j$). This last change in indices will allow us to write the sum more simply. Indeed if we apply it Eq.~\eqref{eq:pre-change} we find 
\begin{align}
      &\sum_{i,j = 0}^\infty \frac{(2\expansionorder + 1 +i +j)!! }{(\expansionorder+1+i)!(\expansionorder+1+j)!}(\gamma\sigma_{\nparams})^{2\expansionorder + 2 + i + j} \frac{1+(-1)^{i+j}}{2} \nonumber\\
      &=\sum_{y' = 0}^\infty\sum_{x  = 0}^{y'} \frac{(2\expansionorder + 1 + y')!! }{(\expansionorder+1+x)!(\expansionorder+1+y'-x)!}(\gamma\sigma_{\nparams})^{2\expansionorder + 2 +y'} \frac{1+(-1)^{y'}}{2}
\end{align}
It is clear from the previous expression that $y'$ has to be even, thus we only keep the terms of $y'$ that are divisible by 2. Enforcing this condition is equivalent to (i) replacing $y' = 2y$ in the terms and in the limit of the second sum, and (ii) keeping the limit of the first sum from $y = 0$ to $\infty$. Implementing (i) and (ii) leads to
\begin{align}
    \sum_{y' = 0}^\infty\sum_{x  = 0}^{y'} \frac{(2\expansionorder + 1 + y')!! }{(\expansionorder+1+x)!(\expansionorder+1+y'-x)!}(\gamma\sigma_{\nparams})^{2\expansionorder + 2 +y'} \frac{1+(-1)^{y'}}{2} & =
    \sum_{y = 0}^\infty\sum_{x  = 0}^{2y} \frac{(2\expansionorder + 1 + 2y)!! }{(\expansionorder+1+x)!(\expansionorder+1+2y-x)!}(\gamma\sigma_{\nparams})^{2\expansionorder + 2 +2y}  \,.
\end{align}

We have finished the process of changing the variables and will proceed to bound the factorial terms. We start by recalling the identity $(2k + 1)!! = \frac{(2k+2)!}{2^{k+1} (k+1)!}$ with $k = \expansionorder+y$, which leads to
\begin{align}
    \sum_{y = 0}^\infty\sum_{x  = 0}^{2y} \frac{(2\expansionorder + 1 + 2y)!! }{(\expansionorder+1+x)!(\expansionorder+1+2y-x)!}(\gamma\sigma_{\nparams})^{2\expansionorder + 2 +2y} 
    & =  \sum_{y = 0}^\infty \frac{(\gamma\sigma_{\nparams})^{2\expansionorder + 2 +2y}}{2^{\expansionorder + y +1}(\expansionorder + y +1 )!}  \sum_{x  = 0}^{2y} \frac{(2\expansionorder + 2 + 2y)!}{(\expansionorder+1+x)!(\expansionorder+1+2y-x)!} \,.\label{eq:change-variable-done}
\end{align}
Let us focus now on upper bounding the second sum involving $x$. We write the factorial terms in a binomial form as
\begin{align}
     \sum_{x  = 0}^{2y} \frac{(2\expansionorder + 2 + 2y)!}{(\expansionorder+1+x)!(\expansionorder+1+2y-x)!} = \sum_{x  = 0}^{2y}\binom{2\expansionorder + 2 + 2y}{\expansionorder+1+x} \leq \sum_{x  = -\expansionorder -1}^{2y + \expansionorder +1}\binom{2\expansionorder + 2 + 2y}{\expansionorder+1+x} = 2^{2\expansionorder + 2 + 2y}\,,
\end{align}
where in the inequality we use that all the terms in the summation are positive, meaning that we can increase the limits of the sum. Then, in the last equality we use that the sum of all this terms equal the exponential (as per the Binomial Theorem/Binomial expansion). Replacing this bound in Eq.~\eqref{eq:change-variable-done} leads to
\begin{align}
    \sum_{y = 0}^\infty \frac{(\gamma\sigma_{\nparams})^{2\expansionorder + 2 +2y}}{2^{\expansionorder + y +1}(\expansionorder + y +1 )!}  \sum_{x  = 0}^{2y} \frac{(2\expansionorder + 2 + 2y)!}{(\expansionorder+1+x)!(\expansionorder+1+2y-x)!} \leq & \sum_{y = 0}^\infty \frac{\left(2\gamma^2 \sigma_{\nparams}^{2} \right)^{\expansionorder + 1 +y}}{(\expansionorder + y +1 )!} \\
    \leq &  \sum_{y = 0}^\infty \frac{\left(2\gamma^2 \sigma_{\nparams}^{2} \right)^{\expansionorder + 1 +y}}{(\expansionorder +1 )! y!}\\
    =&\frac{\left(2\gamma^2 \sigma_{\nparams}^{2} \right)^{\expansionorder + 1}}{(\expansionorder +1 )!} e^{2\gamma^2\sigma_{\nparams}^{2} } \label{eq:bound-double-factorial-done}
\end{align}
where in the second inequality we use that $(\expansionorder+1+y)!\geq (\expansionorder+1)!y!$, and we use the fact that $\sum_{k=0}^{\infty} \frac{x^k}{k!} = e^x$.

Finally, we revisit Eq.~\eqref{eq:boundeps_withclt} by writing in terms of the new variables and bounding with Eq.~\eqref{eq:bound-double-factorial-done}, which leads to
\begin{align}
       \left( \epsilon_{\expansionorder,\alv^*}^{{\rm MSE}}\right)^2 & \leq \norm{O}_\infty^2\frac{\left(2 \gamma^2\sigma_{\nparams}^{2} \right)^{\expansionorder + 1}}{(\expansionorder +1 )!} e^{2\gamma^2\sigma_{\nparams}^{2} } \\
       & = \frac{\norm{O}_\infty^2}{(\expansionorder +1 )!}\left( \frac{2 {\nparams} \gamma^2 r^2_{\nparams}}{3} \right)^{\expansionorder+1} e^{\frac{2 {\nparams} \gamma^2r^2_{\nparams}}{3} } \;,
\end{align}
where to reach the last equality we use Eq.~\eqref{eq:variance_sm}. We remind again that we explicitly use the condition that $r_{\nparams} \in \OC\left(\frac{1}{\sqrt{{\nparams}}}\right)$ in Eq.~\eqref{eq:condition_r} for the bound to be valid. Moreover, assuming $\gamma\in\OC(1)$ implies $\nparams\gamma^2r_{\nparams}^2\in\OC(1)$. Therefore, we can see that the MSE above can be made arbitrarily small by increasing $\expansionorder$. Especially, if $\norm{O}_\infty\in\OC(1)$ we have $\epsilon_{\expansionorder,\alv^*}^{{\rm MSE}}\in\OC(1)$.
\end{proof}

\subsection{Scaling of the number of expectation function terms for given surrogation order: Proof of Proposition~\ref{prop:number_of_terms}}\label{app:number_of_terms}\label{app:generalsurrogatetheorem-5}
\begin{proof}
   We start this proof by recalling the expression of the surrogate in Eq.~\eqref{eq:surrogate_taylor} 
    \begin{equation}
        \fsur_{\expansionorder}(\alv) =\sum_{k = 0}^{\expansionorder} \frac{1}{k!}\sum_{i_1, i_2, ..,i_k}^{\nparams} \gradlk \dellk \;. 
    \end{equation}
Now, let us rewrite this expression as a sum over unique partial derivatives. Indeed, we can group terms using the symmetry of the partial derivatives (e.g. $\partial_{i,j}f = \partial_{j,i}f$) to give

\begin{align}
\fsur_{\expansionorder}(\alv) & = \sum_{k = 0}^{\expansionorder} \frac{1}{k!}\sum_{i_1, i_2, ..,i_k=1}^{\nparams} \partial_{i_1,..,i_k}f(\alv^*)
 \dellk \\
 &= \sum_{k = 0}^{\expansionorder} \frac{1}{k!}\sum_{i_1, i_2, ..,i_k=1}^{\nparams} \prod_{i_l=i_1}^{i_k}\left(\delta_{i_l}\partial_{i_l}\right)f(\alv^*)
\end{align}
where we noted $\dellk = \prod_{i_s=i_1}^{i_k}\delta_{i_s}$, with $\delta_{i_l} = \alpha_{i_l} - \alpha^*_{i_l}$. We can recognize that this expression contains all the terms of a multinomial expansion. Indeed, the sum of products $
    \sum_{i_1, i_2, ..,i_k=1}^{\nparams} \prod_{i_l=i_1}^{i_k}\left(\delta_{i_l}\partial_{i_l}\right) $
contains all the possible combinations of products $\left(\delta_{i_l}\partial_{i_l}\right)$ up to a maximum power of $k$. Therefore, we can just rewrite this expression in a more compact form making use of the multinomial expression. That is, as
\begin{align}
\fsur_{\expansionorder}(\alv) & = \sum_{k = 0}^{\expansionorder} \frac{\left(\sum_{i=1}^\nparams \delta_i\frac{\partial}{\partial \alpha_i}\right)^k} {k!}f(\alv^*)\\
 &=\sum_{k = 0}^{\expansionorder}\frac{1}{k!}\sum_{k_1+...+k_\nparams=k} \binom{k}{k_1,...,k_\nparams} \prod_{l=1}^\nparams\left(\delta_l^{k_l}\partial_{l}^{k_l}\right) f(\alv^*)\\
 &=\sum_{k = 0}^{\expansionorder}\sum_{k_1+...+k_\nparams=k} \prod_{l=1}^\nparams\left(\frac{\delta_l^{k_l}}{k_l!}\partial_{l}^{k_l}\right) f(\alv^*)\;, \label{eq:Pauli-surrogate-loss-unique-partial-sum}
\end{align}
where the second equality is obtained by expanding the multinomial. In this expression we have introduced $k_l\in\mathbb{N}$ which is the order of the derivative with respect to parameter $\alpha_l$ for each $l\in\{1,2,...,\nparams\}$. The last equality is obtained by simplifying the expression of the multinomial coefficient and the factorial term:
\begin{equation}
    \binom{k}{k_1,...,k_\nparams}\frac{1}{k!} = \prod_{l=1}^\nparams \frac{1}{k_l!}\, .
\end{equation}

The number of different terms $N$ that need to be evaluated in order to construct the surrogate/simulate the expectation function landscape corresponds to the number of different derivatives. Indeed, for example in the case of surrogation, each derivative would be a different term that would need to be estimated using a quantum computer. In other words, we assume that each partial derivative can be obtained by at worst measuring $N_{d} $ different parameter settings of the expectation function on a quantum computer. For the $k^{\rm th}$ order in the surrogate, there exists $\binom{{\nparams}+k-1}{k}$ unique partial derivatives (which is the number of possible $\vec{k}=(k_1,k_2,...,k_\nparams)\in\mathbb{N}^\nparams$ such that $\sum_{l=1}^\nparams k_l=k$), leading to at most $\binom{{\nparams}+k-1}{k}N_{d} $ expectation values to be measured in a quantum computer. Then we can readily compute an upper-bound in the total number of terms $N$ we will have to measure by summing all the orders in the surrogate.
\begin{equation} \label{eq:number-of-loss-term-taylor-surrogate}
 N = \sum_{k=0}^\expansionorder\binom{{\nparams}+k-1}{k}N_{d} \leq  N_{d} \left[\frac{ e({\nparams}+\expansionorder-1)}{\expansionorder}\right]^\expansionorder \;,
\end{equation}
where the inequality is obtained using $\binom{n}{k}\leq \frac{n^k}{k!}$ as follows
\begin{align}
    \sum_{k=0}^\expansionorder\binom{{\nparams}+k-1}{k} &\leq \sum_{k=0}^\expansionorder\frac{({\nparams}+k-1)^k}{k!} \\
    &=\sum_{k=0}^\expansionorder\frac{\expansionorder^k}{k!}\left(\frac{{\nparams}+k-1}{\expansionorder}\right)^k \\
    &\leq \left(\frac{{\nparams}+\expansionorder-1}{\expansionorder}\right)^\expansionorder \sum_{k=0}^\expansionorder\frac{\expansionorder^k}{k!} \\
    &\leq \left(\frac{e({\nparams}+\expansionorder-1)}{\expansionorder}\right)^\expansionorder \;,
\end{align}
where the second inequality is obtained by noticing that the term $\left(\frac{{\nparams}+k-1}{\expansionorder}\right)^k$ reaches its largest value when $k=\expansionorder$ (as it increases with $k$ for all $\expansionorder\leq {\nparams}-1$), and the last inequality is obtained by adding terms in the sum up to $k\to \infty$ and recognising the exponential series i.e. $\sum_{k=0}^\infty \frac{\expansionorder^k}{k!}=e^\expansionorder$. 
\end{proof}

\subsection{General surrogation guarantee: Proof of \supt~\ref{coro:surrogate_quantum_quantum}}\label{app:generalsurrogatetheorem-7}
\begin{proof}
Assumption~\ref{assumption:bounded-loss-evaluations} and Assumption~\ref{assumption:loss-evaluations-expectations} implies that the $\expansionorder$-order surrogate $\fsur_{\expansionorder}(\alv)$ can be expressed as a linear combination of expectation functions evaluated at different parameters. In particular, from Assumption~\ref{assumption:loss-evaluations-expectations}, we can express each partial derivatives in Eq.~\eqref{eq:Pauli-surrogate-loss-unique-partial-sum} as
\begin{align}
    \prod_{l}^k\frac{\partial\fsur_{\expansionorder}(\alv^*)}{\partial\alpha_l^{k_l}}=\sum_{j=1}^{N_d}b_j^{(\vec{k})}f(\alv_j^{(\vec{k})}) \;,
\end{align}
where $\vec{k}=(k_1,k_2,...,k_m)\in \mathbb{N}^\nparams$. Therefore, Eq.~\eqref{eq:Pauli-surrogate-loss-unique-partial-sum} can be rewritten as follows

\begin{align}
    \fsur_{\expansionorder}(\alv) & =\sum_{k = 0}^{\expansionorder}\sum_{k_1+...+k_\nparams=k}\prod_l\frac{\delta_l^{k_l}}{k_l!}\sum_{j=1}^{N_d}b_j^{(\vec{k})}f(\alv_j^{(\vec{k})}) \\
    &=\sum_{k = 0}^{\expansionorder}\sum_{|\vec{k}|=k}\sum_{j=1}^{N_d} \left(b_j^{(\vec{k})}\prod_{l=1}^{\nparams}\frac{\delta_l^{k_l}}{k_l!} \right) f(\alv_j^{(\vec{k})}) \\
    &=\sum_{|\vec{k}|\leq\expansionorder} \sum_{j=1}^{N_d} \left(b_j^{(\vec{k})}\prod_{l=1}^{\nparams}\frac{\delta_l^{k_l}}{k_l!} \right) f(\alv_j^{(\vec{k})}) \\
    &=\sum_{l=1}^{N} c_l(\alv) f(\alv_l), \label{eq:surrogate_taylor_expressed_linear_combination}
\end{align}
where in the last equality we use an index $l$ for each possible choice of $(\vec{k},j)$ such that $\vec{k}\in\mathbb{N}^\nparams$  with $|\vec{k}|=k_1+...+k_\nparams\leq\expansionorder$ and $j\in\{1,...,N_d\}$ (notice that there are $N$ possible choice of $(\vec{k},j)$). So, we defined $f(\alv_l)=f(\alv_j^{(\vec{k})})$ and $c_l(\alv)=b_j^{(\vec{k})}\prod_{j=1}^{\nparams}\frac{\delta_j^{k_j}}{k_j!} $. The total number of terms in the sum, $N$ is given in Eq.~\eqref{eq:number-of-loss-term-taylor-surrogate} as well as its upper bound (see Proposition~\ref{prop:number_of_terms}).

\noindent\underline{(i). The worst-case error with the perturbation $r \in \OC(1/\nparams)$.}

The total error can be upper bounded by the sum of both truncation error and shot noise error using the triangle inequality:
\begin{align} \label{eq:proof-total-error-presense-shotnoise}
    \left|\widehat{\fsur_\expansionorder}(\alv) - f(\alv) \right| \leq \left|\widehat{\fsur_\expansionorder}(\alv) - \fsur_{\expansionorder}(\alv)\right| + \left|\fsur_{\expansionorder}(\alv) - f(\alv)\right| \;.
\end{align}
The term $|\fsur_{\expansionorder}(\alv) - f(\alv)|$ can be bounded using Proposition~\ref{th:wors_case_error_app}. We now focus on the term $\left|\widehat{\fsur_\expansionorder}(\alv) - \fsur_{\expansionorder}(\alv)\right|$. 
Lemma~\ref{lemma:multiple-observable-max-error} can be applied to the surrogate of the expectation function expressed in the form of Eq.~\eqref{eq:surrogate_taylor_expressed_linear_combination} by noticing that for each $l$ we can define operator $O(\alv_l)$ such that $f(\alv_l)=\Tr[\rho O(\alv_l)]$ and $\norm{O(\alv_l)}_\infty \leq \norm{O}_\infty$ with associate coefficients $c_l(\alv)$. Let $\vec{c}_{\alv}$ be the vector of length $N$ with the $l^{\rm th}$ component $c_l (\alv)$. Therefore, the measurement strategy presented in Lemma~\ref{lemma:multiple-observable-max-error} can be applied to output an estimator $\widehat{\fsur_\expansionorder}(\alv)$ of $\fsur_\expansionorder(\alv)$ such that $N_s$ measurements on different copies of $\rho$ are sufficient to ensure 
\begin{equation} \label{eq:shot-noise-error-Taylor-surrogate-numerical}
    \left|\widehat{\fsur_\expansionorder}(\alv)-\fsur_\expansionorder(\alv)\right|< \norm{\vec{c}_{\alv}}_{1,{\rm worst}} \norm{O}_\infty \sqrt{\frac{2 \log(2/\delta)}{N_s}}\;,
\end{equation}
 with probability at least $1-\delta$. Now, using the assumption that the coefficients $b_j^{(\vec{k})}$ are bounded by $b_0$, i.e. $|b_j^{(\vec{k})}|\leq b_0$, we have 
\begin{align}
    \norm{\vec{c}_{\alv}}_1 &= \sum_{l=1}^N |c_l(\alv)| \\
    &=\sum_{|\vec{k}|\leq \expansionorder}\sum_{j=1}^{N_d} |b_j^{(\vec{k})}|\prod_{l=1}^\nparams \frac{|\delta_l|^{k_l}}{k_l!} \\
    &\leq b_0 N_d  \sum_{|\vec{k}|\leq \expansionorder}\prod_{l=1}^\nparams\frac{|\delta_l|^{k_l}}{k_l!} \\
    &\leq b_0 N_d  \sum_{|\vec{k}|\leq \expansionorder}\prod_{l=1}^\nparams\frac{r^{k_l}}{k_l!} \\ 
    &=b_0 N_d \sum_{k=0}^{\expansionorder} \frac{(r\nparams)^k}{k!} \;, 
    \label{eq:upper-bound-1-norm-b_alv-general}
\end{align}
where the second inequality follows from $|\delta_l|\leq r$ and the last equality is obtained by recognising the multinomial sum $\sum_{|\vec{k}|=k} \binom{k}{k_1,...,k_\nparams} \prod_{l=1}^\nparams r^{k_l}=(\nparams r)^k$ where $\binom{k}{k_1,...,k_\nparams}=\frac{k!}{\prod_{l=1}^\nparams k_l!}$. Notice that this bound is valid for any $\alv\in\vol(\alv^*,r)$ due to the second inequality, so this is also an upper bound of $\norm{\vec{c}_{\alv}}_{1,{\rm worst}}$.  Additionally, we can upper bound Eq.~\eqref{eq:upper-bound-1-norm-b_alv-general} by 
\begin{align}
    b_0 N_d \sum_{k=0}^{\expansionorder} \frac{(r\nparams)^k}{k!} &\leq b_0 N_d e^{r\nparams}\;,
\end{align}
which is obtained by extending the limit of the sum to $\infty$ and recognising the exponential series. Assuming that $r\nparams \in \OC(1)$ ensures that $e^{r\nparams}\sim\OC(1)$. By inserting previous bound on the effective 1-norm in  Eq.~\eqref{eq:shot-noise-error-Taylor-surrogate-numerical} to bound the shot noise error, and using Proposition~\ref{th:wors_case_error_app} to bound the truncation error, the total error is bounded as follows
\begin{align}
     \left|\widehat{\fsur_\expansionorder}(\alv) - f(\alv) \right| &\leq b_0N_d e^{r\nparams} \norm{O}_\infty \sqrt{\frac{2 \log(2/\delta)}{N_s}}  + \frac{\norm{O}_\infty}{(\expansionorder+1)!}(\gamma r {\nparams})^{\expansionorder+1} \;,
\end{align}
which happens with probability at least $1 - \delta$. Assuming that $b_0,N_d\in\OC({\rm poly}(\nparams))$, then we can see that using polynomial number of shots is enough to decreases the empirical error while the truncation error can be reduced by increasing $\expansionorder$.

\medskip

\noindent\underline{(ii). The average-case error with the perturbation $r \in \OC(1/\sqrt{\nparams})$.} 

First, let us bound the total mean square error by both contribution of the empirical error and the truncation error using that $(x+y)^2\leq 2(x^2+y^2)$ for any reals $x$ and $y$. This leads to
\begin{align}
\Ebb_{\vec{\alpha} ,\MC}\left[\left(\widehat{\fsur_{\expansionorder}}(\alv)-f(\alv)\right)^2\right]&\leq 2\Ebb_{\vec{\alpha}, \MC}\left[\left(\widehat{\fsur_{\expansionorder}}(\alv)-\fsur_{\expansionorder}(\alv)\right)^2\right] + 2\Ebb_{\vec{\alpha} }\left[\left(\fsur_\expansionorder(\alv)-f(\alv)\right)^2\right]\;.
\end{align}
Notice that the average here is over both the random parameters $\alv$ and the output of the measurement denoted by $\MC$ (except for the truncation error which depends only on $\alv$).

Now, the estimate $\widehat{\fsur_\expansionorder}(\alv)$ obtained by following the measurement protocol described in Lemma~\ref{lemma:multiple-observable-mean-squared-error} leads to the following upper bound on the empirical mean square error
\begin{align}
    \Ebb_{\alv,\MC}\left[\left(\widehat{\fsur_\expansionorder}(\alv) - \fsur_\expansionorder(\alv)\right)^2\right] &\leq \frac{\norm{O}_\infty^2 \norm{\vec{c}}_{1,{\rm avg}}^2}{N_s} \;,
\end{align}
where $\norm{\vec{c}}_{1,{\rm avg}}$ is the averaged effective 1-norm (see Definition~\ref{def:effective-1-norms}) of the vector $\vec{c}_{\alv}$.

Finally, combining this result together with the MSE upper bound of the surrogate from Proposition~\ref{th:mse_surrogated} leads to

\begin{align}
    \Ebb_{\vec{\alpha} ,\MC}\left[\left(\widehat{\fsur_{\expansionorder}}(\alv)-f(\alv)\right)^2\right] &\leq  2\frac{\norm{O}_\infty^2 \norm{\vec{c}}_{1,{\rm avg}}^2}{N_s} + 2\left( \frac{2 \gamma^2 {\nparams} r^2}{3} \right)^{\expansionorder+1}\frac{\norm{O}_\infty^2}{(\expansionorder +1 )!} e^{\frac{ 2\gamma^2{\nparams} r^2}{3} } \;.
\end{align}

Notice that Eq.~\eqref{eq:upper-bound-1-norm-b_alv-general} is also a valid upper bound for $\norm{\vec{c}}_{1,{\rm avg}}$ as it correspond to a worst case 1-norm obtained by upper bounding each $|c_l(\alv)|$. Therefore, the sum in Eq.~\eqref{eq:upper-bound-1-norm-b_alv-general} is bounded as follows
\begin{align}
    \sum_{k=0}^\expansionorder \frac{(rm)^k}{k!} \leq (\expansionorder+1)\frac{(rm)^\expansionorder}{\expansionorder !} \;,
\end{align}
under the condition that $r\nparams > \expansionorder$ from Eq.~\eqref{eq:binom_sum_first_powers_inequ} (since $r\in\OC(1/\sqrt{\nparams})$ , we have that $r\nparams\in\OC(\sqrt{\nparams})$). From the assumptions, the bound on the effective norm scales polynomially with the number of parameters. Finally, we have
\begin{align}
    \Ebb_{\vec{\alpha} ,\MC}\left[\left(\widehat{\fsur_{\expansionorder}}(\alv)-f(\alv)\right)^2\right] &\leq 2\left(\frac{b_0 N_d(\expansionorder+1)(r\nparams)^{\expansionorder}\norm{O}_\infty}{\sqrt{N_s} \expansionorder !}\right)^2 + 2\left( \frac{2 \gamma^2 {\nparams} r^2}{3} \right)^{\expansionorder+1}\frac{\norm{O}_\infty^2}{(\expansionorder +1 )!} e^{\frac{ 2 \gamma^2 {\nparams} r^2}{3} } \;.
\end{align}

Notice that both measurement strategies from Lemma~\ref{lemma:multiple-observable-mean-squared-error} and Lemma~\ref{lemma:multiple-observable-max-error} can be used to get previous bound as we upper-bounded the average-case 1-norm by the worst-case one.
\end{proof}

\subsection{Parametrised quantum circuits that satisfy assumptions: Proof of Proposition~\ref{prop:derivatives-upperbound-unitary} and Corollary~\ref{cor:surrogate-taylor-unitary}}
Corollary~\ref{cor:surrogate-taylor-unitary} is a direct consequence of \supt~\ref{coro:surrogate_quantum_quantum} and Proposition~\ref{prop:derivatives-upperbound-unitary}. Let us first derive Proposition~\ref{prop:derivatives-upperbound-unitary}.
\begin{proof}
     The expectation function considered here is of the form of Eq.~\eqref{eq:loss_app} with unitary channel, i.e. $f(\alv)=\Tr[\rho U^\dagger(\alv)OU(\alv)]$, where $U(\alv)=\prod_{l=1}^\nparams V_l U_l(\alpha_l)$ such that $V_l$ are arbitrary non-parametrized unitaries and each $U_l(\alpha_l)=\exp(-iH_l\alpha_l)$ is a parametrised gate with Hermitian generator $H_l$. Now, let us bound the $k$-th order derivatives of the form $\frac{\partial^k f(\alv)}{\prod_{l=1}^\nparams \partial\alpha_l^{k_l}}:=\partial^{\vec{k}}f(\alv)$ where $\sum_{l=1}^\nparams k_l=k$ and  $k_l\in\mathbb{N}$ for each $l$.
    First, let us use Hölder inequality as follows
    \begin{align}
      |\partial^{\vec{k}}f(\alv)|&=|\Tr[\rho\partial^{\vec{k}}U^\dagger(\alv)OU(\alv) ]| \\
      &\leq \norm{\rho}_1\norm{\partial^{\vec{k}}U^\dagger(\alv)OU(\alv)}_\infty \\
      &=\norm{\partial^{\vec{k}}(U^\dagger(\alv)OU(\alv))}_\infty\;,
    \end{align}
    where the last equality is given by the normalisation of quantum states.
    Now, the goal is to show that 
    \begin{align}
        \norm{\partial^{\vec{k}}(U^\dagger(\alv)OU(\alv))}_\infty \leq \norm{O}_\infty \prod_{l=1}^\nparams \left(2\norm{H_l}_\infty\right)^{k_l}\;,
    \end{align}
    by induction on $l$.
    First notice that the operator we want to bound can be rewritten as follows
    \begin{align}
        \label{eq:split-derivatives-reccursive-proofthm1}
        \partial^{\vec{k}}\left((U^\dagger(\alv)OU(\alv)\right) &= \frac{\partial^{k_1}}{\partial\alpha_1^{k_1}}\left(U_1^\dagger(\alpha_1) V_1^\dagger\left(...\frac{\partial^{k_\nparams}}{\partial\alpha_\nparams^{k_\nparams}}\left(U_\nparams^\dagger(\alpha_\nparams)V_\nparams^\dagger O V_\nparams U_\nparams(\alpha_\nparams)\right)...\right)V_1U_1(\alpha_1)\right)\;.
    \end{align}
     Then, each derivative consists in applying a commutator with the generator or more formally for any operator $A_{l+1}$ (which will be specify later), we have 
     \begin{align}
         \frac{\partial^{k_l}}{\partial\alpha_l^{k_l}}\left(U_l^\dagger(\alpha_l)A_{l+1} U_l(\alpha_l)\right) &= \frac{\partial^{k_l-1}}{\partial\alpha_l^{k_l-1}}\left(U_l^\dagger(\alpha_l)i[H_l,A_{l+1}] U_l(\alpha_l)\right) \\
         &= \frac{\partial^{k_l-2}}{\partial\alpha_l^{k_l-2}}\left(U_l^\dagger(\alpha_l)i[H_l,i[H_l,A_{l+1}]] U_l(\alpha_l)\right) \\
         &=...\\
         &=i^{k_l}U_l^\dagger(\alpha_l)\underbrace{[H_l,...,[H_l}_{k_l  \text{ times}},A_{l+1}]...] U_l(\alpha_l)\;.
     \end{align}
     From this result and the unitarily invariance properties of the norm, we have
     \begin{align}
     \norm{\frac{\partial^{k_l}}{\partial\alpha_l^{k_l}}\left(U_l^\dagger(\alpha_l)V_l^\dagger A_{l+1} V_l U_l(\alpha_l)\right)}_\infty &= \norm{[H_l,...,[H_l,V_l^\dagger A_{l+1} V_l]...]}_\infty \\
     &\leq (2\norm{H_l}_\infty)^{k_l}\norm{V_l^\dagger A_{l+1} V_l}_\infty \\
     &= (2\norm{H_l}_\infty)^{k_l}\norm{ A_{l+1} }_\infty \label{eq:upper-bound-derivative-single-parameter-proofthm1} \;,
     \end{align}
     where the last equality follows from unitarily invariance of the norm, and the inequality is obtained by induction using the property that for any operators $A$ and $B$, we have \begin{align}
         \norm{[A,B]}_\infty &\leq  \norm{AB}_\infty+\norm{BA}_\infty \\
         &\leq 2\norm{A}_\infty\norm{B}_\infty \;,
     \end{align}
     where we first use triangle inequality, and then sub-multiplicativity of the spectral norm. Now, using Eq.~\eqref{eq:upper-bound-derivative-single-parameter-proofthm1} with $A_l=\frac{\partial^{k_l}}{\partial\alpha_l^{k_l}}\left(U_l^\dagger(\alpha_l) V_l^\dagger\left(...\frac{\partial^{k_\nparams}}{\partial\alpha_\nparams^{k_\nparams}}\left(U_\nparams^\dagger(\alpha_\nparams)V_\nparams^\dagger O V_\nparams U_\nparams(\alpha_\nparams)\right)...\right)V_lU_l(\alpha_l)\right)$ gives $\norm{A_l}_\infty\leq (2\norm{H_l}_\infty)^{k_l}\norm{ A_{l+1} }_\infty$. Finally, using it recursively starting from Eq.~\eqref{eq:split-derivatives-reccursive-proofthm1}  leads to the desired result.
     \begin{align}
      \norm{\partial^{\vec{k}}\left((U^\dagger(\alv)OU(\alv)\right)}_\infty &=\norm{A_1}_\infty \\
      &\leq (2\norm{H_1}_\infty)^{k_1}\norm{ A_{2}}_\infty \\
      &\leq ...\\
      &\leq \norm{O}_\infty \prod_{l=1}^\nparams \left(2\norm{H_l}_\infty\right)^{k_l} \;,
     \end{align}
     where we used the convention that $A_{\nparams+1}=O$ in the last step. Now, we can simply upper bound the norm of each generator by the largest one, i.e., by $\max_{l}\norm{H_l}_\infty$ which finally leads to 

     \begin{align}
        |\partial^{\vec{k}}f(\alv)| &\leq \norm{O}_\infty \prod_{l=1}^\nparams \left(2\norm{H_l}_\infty\right)^{k_l} \\
        &\leq \norm{O}_\infty \left(2\max_l\norm{H_l}_\infty\right)^k\;.
     \end{align}
     \end{proof}

     Finally, we can prove Corollary~\ref{cor:surrogate-taylor-unitary} as follows.
     \begin{proof}
         Now, from the assumption that each generator has a spectral norm of order $\OC(1)$, then there exist a constant $\gamma\in\OC(1)$ such that $2\norm{H_l}_\infty\leq \gamma$ for each $l$. Therefore, we recover Assumption~\ref{assumption:bounded-derivatives} using Proposition~\ref{prop:derivatives-upperbound-unitary}, which allows us to apply \supt~\ref{coro:surrogate_quantum_quantum}. Therefore, the bounded assumption for the observable i.e. $\norm{O}_\infty\in\OC(1)$ together with \supt~\ref{coro:surrogate_quantum_quantum} completes the proof of the corollary.
     \end{proof}

\subsection{Application to circuits with Pauli generators}
\label{app:sine-cosine-surrogate-non-Clifford}
In this appendix we apply our surrogation technique to the special case when the circuit generators are Pauli strings. Furthermore, we will also explore the connection between the Taylor expansion based surrogate at the heart of Theorem~\ref{thm:surrogate-average-general} and the Pauli propagation surrogation we will discuss in Appendix~\ref{app:smallanglePP}. Concretely, we will focus on a circuit with the following general structure
\begin{equation}\label{eq:circuit_pauli}
	U\left( \alv \right) = \prod_{j=1}^\nparams V_j U_j(\al_j)\,,
\end{equation}
where $V_j $ are a set of fixed unitary matrices (but not necessarily Clifford) and $U_j(\al_j) = e^{-i\frac{\al_j}{2} P_j}$ are parametrized rotations with $P_j $ Pauli strings. We will further suppose that the parameters $\al_j$ are uncorrelated. 

Under the previous constraints, we first can compute the Fourier expansion of the expectation function. Indeed, as it will be explicitly shown below (Appendix~\ref{app:fourier_exp}), the expectation function can be naturally expressed in terms of sine/cosine functions  as 
\begin{align}\label{eq:fourier_exp}
    f(\alv)&= \sum_{\vec{\omega}\in\{0,1,2\}^{\nparams}}\Apath_{\vec{\omega}}(\vec{\delta}) \Tr[O_{\vec{\omega}}(\alv^*)\rho_0] \;,
\end{align}
where $\vec{\delta}=\alv-\alv^*$ with $\alv \in\vol(\alv^{*}, r)$ a hypercube of the form of Eq.~\eqref{eq:hypercube} sampled uniformly as in Eq.~\eqref{eq:uniformdist}.
Here $\omv = (\om_1, \om_2, ...,\om_{\nparams} )$ is an index vector with each $\om_j$ taking a value of $0$, $1$ or $2$ and the Fourier coefficients are given by
\begin{align}
    &\Apath_{\vec{\omega}}(\vec{\delta})=\prod_{i=1}^{\nparams}\apath_{\omega_j}(\delta_j) \;\;\;\;\;\; \text{where,} \\
    &\apath_{\omega_j}(\delta_j) = \left\{
    \begin{array}{ll}
        \cos^2\left(\frac{\delta_j}{2}\right) & \mbox{if } \om_j =0 \;, \\
        2\sin\left(\frac{\delta_j}{2}\right)\cos\left(\frac{\delta_j}{2}\right) & \mbox{if } \om_j=1 \;, \\
        \sin^2\left(\frac{\delta_j}{2}\right) & \mbox{if } \om_j=2 \;.
    \end{array}
\right.\label{eq:decompositionapath}
\end{align}
The operators $O_{\omv}(\thv^*)$ are defined recursively and can be computed efficiently. The recursive method to obtain them is detailed in Eq.~\eqref{eq:operator-transformation-Fourrier-path} in Appendix~\ref{app:fourier_exp}.

For an individual term in the sum associated with $\vec{\omega}$, denote $n_{1}(\vec{\omega})$ and $n_{2}(\vec{\omega})$ as the number of times $\om_j = 1$ and $\om_j = 2$ appear in a given $\omega$. Then we can define the order of each term as
\begin{align}
    K(\vec{\omega}) = n_{1}(\vec{\omega}) + 2 n_{2} (\vec{\omega}) \;.
\end{align}
Notice that this count the number of sine contributions to the corresponding coefficient.
We can then consider an order $\expansionorder$ surrogate of the expectation function that is built by keeping only terms with $K(\vec{\omega}) \leq \expansionorder$. That is, we have
\begin{align}\label{eq:surrogate-sine-cosine}
    \fsur_{\expansionorder}(\alv) := \sum_{\vec{\omega} \;;\; K(\vec{\omega}) \leq \expansionorder}  \Apath_{\vec{\omega}}(\vec{\delta}) \Tr[O_{\vec{\omega}}(\alv^*)\rho_0] \, .
\end{align}

This surrogate is roughly equivalent to  the Taylor surrogate (i.e., Eq.~\eqref{eq:surrogate_taylor}). Indeed, for a given order, the error of the two techniques. Trivially, we can compare the surrogate via Taylor and the Taylor expansion of the $\sin/\cos$ surrogate. We see that these are equivalent, as the terms dropped in the $\sin/\cos$ surrogate are also dropped in the Taylor one (this comes from the fact that $\cos\alpha\sim1,\, \sin\alpha\sim \alpha$). Therefore, the only difference between the two surrogate methods lies in the difference between the term $\Apath(\vec{\delta})$ and its Taylor expansion at order $\expansionorder$. This difference is of order $\expansionorder+1$. \\

Now that we have established the connection between the Taylor and the $\sin/\cos$ surrogation techniques, we devote the rest of this section to explaining how the latter connects with the Pauli Propagation technique. To do so, let us start by assuming that the measurement operator is given by a single Pauli, i.e., $O = P$ (we will later generalize to the case of $O = \sum_{i}a_i P_i$.). Furthermore, let us now assume that the $V_i$ gates are Clifford gates. 

Pauli Propagation (which we will detail further in Appendix~\ref{app:smallanglePP}) is understood in the Heisenberg picture. Here, we  study the backwards evolved measurement operator
\begin{equation}\label{eq:propagated_pauli}
    U^\dagger (\alv) P U(\alv) = \prod_{j=1,l = 1}^m V_j^\dagger U_j^\dagger(\alpha_j) P U_l(\alpha_l) V_l \,  .
\end{equation}
Let us start by using Eq.~\eqref{eq:exp_pauli_sincos} to compute the effect of applying the layer closest to the Pauli observable. For a general observable this would look like
\begin{align}\label{eq:pauli_toprop}
    e^{i\frac{\al_{\nparams}}{2} P_m} P e^{-i\frac{\al_{\nparams}}{2} P_m} = v_0(\al_m) P +  v_1(\al_m)\frac{i}{2}[P_m,P] + v_2(\al)P_m P P_m\, .
\end{align}
However, Pauli strings either commute or anti-commute. Therefore, we can use this to our advantage. Indeed, if $P_m P = P P_m$ we can trivially simplify the previous expression to 
\begin{align}
    e^{i\frac{\al_{\nparams}}{2} P_m} P e^{-i\frac{\al_{\nparams}}{2} P_m} = \left[v_0(\al_m)+v_2(\al_m)\right]P=P\;.
\end{align}
 Furthermore, in the case that $P_m P = - P P_m$, we can also simplify the terms in Eq.~\eqref{eq:pauli_toprop}  using $\frac{i}{2}[P_m,P]  =  iP_m P$ and $  P_m P P_m  = -P$ to give
\begin{align}
    e^{i\frac{\al_{\nparams}}{2} P_m} P e^{-i\frac{\al_{\nparams}}{2} P_m} = \left[v_0(\al_m)-v_2(\al_m)\right] P + v_1(\al_m)iP_m P \, ,
\end{align}
which by simple trigonometric identities (and by recalling Eq.~\eqref{eq:decompositionapath}) can be simplified to  $v_0(\al_m)-v_2(\al_m) = \cos(\al_m)$, and $v_1(\al_m) = \sin(\al_m)$. This means that Eq.~\eqref{eq:pauli_toprop} can be rewritten as 
\begin{align}\label{eq:almost_pauli}
    e^{i\frac{\al_\nparams}{2} P_m} P e^{-i\frac{\al_\nparams}{2} P_m} = 
    \begin{cases}
        P \, &{\rm iff} \, [P_m,P] = 0\\
        \cos(\al_m) P + \sin(\al_m) P' \, &{\rm iff} \, \{P_m,P\}
    \end{cases}
\end{align}
where $P' := iP_m P$ is a Pauli matrix (up to a sign). We can now define the Pauli Path coefficients defined in Eq.~\eqref{eq:pauli_path} that we rewrite here:
\begin{equation}
    \phi_{\omega_i}(\alpha_i)= \begin{cases}
    1,                      & \text{if } \omega_i = 0\\
    \cos(\alpha_i),        & \text{if } \omega_i = 1\\
    \sin(\alpha_i),        & \text{if } \omega_i = -1 \ .
\end{cases}
\end{equation}
Thus we obtain
\begin{align}
    e^{i\frac{\al_{\nparams}}{2} P_m} P e^{-i\frac{\al_{\nparams}}{2} P_m} = 
    \begin{cases}
        \phi_{0}(\al_\nparams)P = P \, &{\rm iff} \, [P_m,P] = 0\\
        \phi_{1}(\al_m) P +  \phi_{2}(\al_m)  P' \, &{\rm iff} \, \{P_m,P\}
    \end{cases}
\end{align}
Now it is worth noting that when we apply any Clifford gate into any Pauli matrix, it returns a Pauli matrix. Therefore, we can repeat this procedure for all the circuit in Eq.~\eqref{eq:propagated_pauli} to obtain
\begin{align}
    U^\dagger (\alv) P U(\alv) = \sum_{\vec{\omega}} \Phi_{\omega}(\alv)P_{\vec{\omega}}\,,
\end{align}
where $P_{\vec{\omega}}$ are the propagated Pauli matrices. In other words, the final Pauli matrices obtained after applying the commutation rules for all of the gates in the circuit.

Finally, it is worth stressing that if $O = \sum_i a_i P_i$ instead of a single Pauli, due to the linearity of the expectation function the results still apply. Indeed, 
if we start from such observable, the expectation function can be written as
\begin{equation}
    f(\alv) = \Tr[\rho U^\dagger(\alv) O U(\alv)] = \sum_i a_i \Tr[\rho U^\dagger(\alv) P_i U(\alv)]\,,
\end{equation}
and we can apply the propagation technique for each $P_i$. This previous leads to
\begin{equation}
    f(\alv) = \sum_{i}\sum_{\vec{\omega}}a_i \Phi_{\vec{\omega}}(\alv)\Tr[\rho P_{\vec{\omega}}] = \sum_{\vec{\omega}}c_{\vec{\omega}}(\alv)\Tr[\rho P_{\omv}]\,,
\end{equation}
where we have introduced the notation $c_{\vec{\omega}}(\alv) = \sum_{i}a_i \Phi_{\omv}(\alv)$. 

As the error associated with this Pauli propagation surrogation strategy are of the same order as the Taylor surrogate (i.e., Eq.~\eqref{eq:surrogate_taylor}), Theorem~\ref{thm:surrogate-average-general} can be used to indirectly bound the efficiency of Pauli Propagation. However, substantially tighter bounds can be obtained by analysing Pauli propagation directly as we will proceed to do in the next full appendix (Appendix~\ref{app:smallanglePP}). 

\subsubsection{Fourier expansion}\label{app:fourier_exp}
Let us now  study the Fourier expansion in Eq.~\eqref{eq:fourier_exp}. We will explicitly show that the expectation function can naturally be expressed in terms of sine and cosine functions. In particular, since  $P_j^2 = \1$, we find 
\begin{align}\label{eq:exp_pauli_sincos}
    e^{-i \frac{\al_j}{2} P_j} = \cos\left(\frac{\al_j}{2}\right) \1 - i \sin\left(\frac{\al_j}{2}\right) P_j \;.
\end{align}
In addition, the rotation gates can be re-expressed as perturbations $\vec{\alpha}$ around any point $\alv \in\vol(\alv^{*}, r)$ for all $j$. This can be trivially seen by just writing $\alv = \alv^* + \vec{\delta}$
\begin{align}
     U\left( \alv \right) 
     & = \prod_{j=1}^\nparams  U_j (\delta_j) U_j(\al^*_j) V_j  \\
     & = \prod_{j=1}^\nparams  U_j (\delta_j) \widetilde{V}_j \;,
\end{align}
where the first equality holds due to $e^{-i \frac{\al_j}{2} P_j} = e^{- i \frac{\al^*_j}{2} P_j} e^{- i \frac{\delta_j}{2} P_j}$ and in the second equality we denote $\widetilde{V}_j := V_j e^{- i \frac{\al_j^*}{2} P_j }$. Also, we introduce the notation
\begin{align}
    \rho_{l} = \prod_{j = 1}^l  U_j (\alpha_j) \widetilde{V}_j \; \rho_0 \left( \prod_{j = 1}^l  U_j (\alpha_j) \widetilde{V}_j \right)^\dagger
\end{align}

Now, the expectation function as defined in Eq.~\eqref{eq:expectation} can be expanded as
\begin{align}
    \LC(\alv) = \;& \Tr\left[\rho_{\nparams-1} \widetilde{V}^\dagger_{\nparams} U^\dagger_\nparams(\delta_\nparams) O U_\nparams(\delta_\nparams) \widetilde{V}_{\nparams} \right] \\
    =\;& \Tr\left[ \rho_{\nparams-1} \widetilde{V}^\dagger_{\nparams} \left( \cos\left(\frac{\delta_\nparams}{2}\right) \1 + i\sin\left(\frac{\delta_\nparams}{2}\right) P_\nparams \right)  O \left( \cos\left(\frac{\delta_\nparams}{2}\right) \1 - i\sin\left(\frac{\delta_\nparams}{2}\right) P_\nparams \right) \widetilde{V}_{\nparams} \right] \\
    =\;& \cos^2\left(\frac{\delta_\nparams}{2}\right)\Tr\left[\rho_{\nparams-1}\left(\widetilde{V}^\dagger_{\nparams} O \widetilde{V}_{\nparams} \right)\right] + 2 \sin\left(\frac{\delta_\nparams}{2}\right) \cos\left(\frac{\delta_\nparams}{2}\right) \Tr\left[\rho_{\nparams-1}\left(\widetilde{V}^\dagger_{\nparams}\frac{i}{2} \left[ P_\nparams, O \right] \widetilde{V}_{\nparams} \right)\right] \\
    & + \sin^2\left(\frac{\delta_\nparams}{2}\right) \Tr\left[  \rho
_{\nparams-1} \left( \widetilde{V}^\dagger_{\nparams} P_\nparams O P_\nparams \widetilde{V}_{\nparams}\right) \right] \\
=\;& \apath_{0}(\delta_\nparams) \Tr\left[\rho_{\nparams-1} O_0 \right] + \apath_{1}(\delta_\nparams) \Tr\left[\rho_{\nparams-1} O_1 \right] + \apath_{2}(\delta_\nparams) \Tr\left[ \rho_{\nparams-1} O_2\right] \;,
\end{align}
where we defined 
\begin{align}
    O_0 =& \widetilde{V}^\dagger_{\nparams} O \widetilde{V}_{\nparams} \\
    O_1 =&\widetilde{V}^\dagger_{\nparams}\frac{i}{2} \left[ P_\nparams, O \right] \widetilde{V}_{\nparams}\\
    O_3 =&  \widetilde{V}^\dagger_{\nparams} P_\nparams O P_\nparams \widetilde{V}_{\nparams}
\end{align}
and 
\begin{align}
    \apath_0(\delta_\nparams) = &  \cos^2\left(\frac{\delta_\nparams}{2}\right)  \\
    \apath_1(\delta_\nparams) = &   2 \sin\left(\frac{\delta_\nparams}{2}\right) \cos\left(\frac{\delta_\nparams}{2}\right) \\
    \apath_2(\delta_\nparams) =&\sin^2\left(\frac{\delta_\nparams}{2}\right)\,.
\end{align}
By iteratively repeating the previous steps, we obtain the form of the expectation function by Fourier expansion as
\begin{align}
    f(\alv)&= \sum_{\vec{\omega}\in\{0,1,2\}^{\nparams}}\Apath_{\vec{\omega}}(\vec{\delta}) \Tr[O_{\vec{\omega}}(\alv^*)\rho_0] \;,
\end{align}
where $\omv = (\om_1, \om_2, ...,\om_{\nparams})$ is an index vector with each $\om_j$ taking a value of $0$, $1$ or $2$. We also introduce a notation $\idx{i}{j} = (\om_i, \om_{i+1}, ..., \om_{j})$ with $i\leq j$ and $\omv = \idx{1}{\nparams}$. Now, the operators $O_{\idx{1}{\nparams}} := O_{\omv}(\thv^*)$ are defined recursively by 
\begin{align}
\label{eq:operator-transformation-Fourrier-path}
    O_{\idx{j}{\nparams}} = \left\{
    \begin{array}{ll}
        \tilde{V}_j^\dagger  O_{\idx{j+1}{\nparams }} \tilde{V}_j & \mbox{if } \om_j=0 \;, \\
        \frac{i}{2} \tilde{V}_j^\dagger \left[P_l, O_{\idx{j+1}{\nparams}}\right]  \tilde{V}_j& \mbox{if } \om_j =1 \;, \\
        \tilde{V}_j^\dagger P_j O_{\idx{j+1}{\nparams}} P_j  \tilde{V}_j & \mbox{if } \om_j =2 \;,
    \end{array}
\right.
\end{align}
with $\tilde{V}_j = e^{-i\frac{\al_j^*}{2} P_j} V_j$, as defined above,
and the Fourier coefficients are given by $\Apath_{\vec{\omega}}(\vec{\delta})=\prod_{j=1}^\nparams \apath_{\omega_j}(\delta_j)$, where
\begin{align}
    \apath_{\omega_j}(\delta_j) = \left\{
    \begin{array}{ll}
        \cos^2\left(\frac{\delta_j}{2}\right) & \mbox{if } \om_j =0 \;, \\
        2\sin\left(\frac{\delta_j}{2}\right)\cos\left(\frac{\delta_j}{2}\right) & \mbox{if } \om_j=1 \;, \\
        \sin^2\left(\frac{\delta_j}{2}\right) & \mbox{if } \om_j=2 \;.
    \end{array}
\right.
\end{align}

\section{Iterative training strategy with patch surrogate}
\label{sec:iterative-training-taylor-surrogate}
In this section, we study an iterative training procedure for parameterized quantum circuits based on patch surrogate. Here, instead of employing a standard hybrid classical-quantum optimization strategy where a loss gradient is estimated at each training iteration, we construct a surrogate of loss gradients around a patch and then perform several classical gradient steps within this patch. After few training iterations, we need to reset the surrogate around a new patch. We investigate whether this iterative strategy can be advantageous. 
The section is structured as follows:
\begin{itemize}
    \item In Appendix~~\ref{sec:surrogate-approach-resource}, we provide a detailed description of the iterative surrogate approach.
    \item In Appendix~\ref{sec:summary-results-taylor-gradient}, we numerically compare the surrogate with the standard gradient descent.
    \item In Appendix~\ref{sec:taylor-iterative-numerical-details}, further numerical details are provided including a detailed description of our shot allocation strategy.
    \item Lastly, in Appendix~\ref{sec:analytical-guarantee-second-order}, the analytical results relating shot noise and shot allocation are derived.
\end{itemize}

\subsection{Description of a surrogate approach}
\label{sec:surrogate-approach-resource}

We consider a quantum surrogate based strategy where the loss gradients are surrogated and then trained purely classically. More specifically, we consider a Taylor surrogate of the loss gradients \( \nabla f(\alv) \) around a point \( \alv^* \) and train this surrogate classically using gradient descent for a small number of steps. To ensure the faithfulness of the training, the process of re-surrogating the loss gradients is then repeated using the newly obtained point as the next expansion center. 

For simplicity, we consider a second-order Taylor expansion of the gradient of the form
\begin{equation} \label{eq:grad-surrogate-order-2}
    \nabla \fsur(\alv) = \nabla f(\alv^*) + H(\alv^*)(\alv - \alv^*)\;,
\end{equation}
where \( H(\alv^*) \) is the Hessian of the loss at the expansion point. The individual components read:
\begin{equation}
    \partial_l \fsur(\alv) = \partial_l f(\alv^*) + \sum_{k=1}^{m} \partial_{l,k} f(\alv^*) \, \delta_k, \quad \delta_k = \alpha_k - \alpha_k^*.
\end{equation}
We note that it is straightforward to extend this construction to higher-order surrogate models. Algorithm~\ref{alg:iter} summarizes the overall protocol. Moreover, whereas other parts of our work primarily focus on surrogating the loss itself, in practical gradient-based training it is often more appropriate to surrogate the loss gradient directly. A related strategy, termed Quantum Analytic Descent, was introduced in Ref.~\cite{koczor2020quantum}; it is based on a second-order truncation of the sine coefficients in the Fourier expansion (see Appendix~\ref{app:fourier_exp}). 

Since standard GD is unbiased as opposed to the surrogate strategy which suffers from Taylor error, we are interested in comparing the performance on both methods by fixing an equivalent total shot budget. This is obtained by choosing
\begin{equation}
    N_{\rm grad}\approx \frac{N_{\rm surrogate}}{N_{\rm iters}}\;\;,
\end{equation}
where $N_{\rm grad}$ is the number of measurement shots for one iteration of standard GD, $N_{\rm iters}$ is the total number of training iterations for standard GD and $N_{\rm surrogate}$ is the total number of measurement shots used in the surrogate training strategy. In this case, the shot noise with standard GD will increase by a factor $N_{\rm iters}$, the number of iterations between surrogate updates. We aim to study how, in the task of estimating the ground state of a given Hamiltonian, the additional shot noise incurred by standard gradient descent balances against the Taylor truncation error introduced by the surrogate, and how this trade-off impacts overall optimization performance.

\begin{algorithm}[H]
\SetAlgoLined
\KwIn{Initial parameters \( \alv_1 \), number of surrogates \( N_{\mathrm{surr}} \), steps per surrogate \( N_{\mathrm{iters}} \)}
\For{$i = 1$ \KwTo $N_{\mathrm{surr}}$}{
    \uIf{$i = 1$}{
        Set \( \alv^* = \alv_1 \)
    }
    \Else{
        Set \( \alv^* = \alv_i \)
    }
    Construct Taylor surrogate \(\grad  \fsur(\alv) \) around \( \alv^* \) by estimating \(\grad  f(\alv^*) \) and \(H(\alv^*)\)\\
    Perform $N_{\mathrm{iters}}$ steps of classical gradient descent with \( \grad\fsur(\alv) \) starting from \( \alv^* \) \\
    Let \( \alv_{i+1} \) be the final iterate
}

\caption{Iterative Training via Taylor Surrogate }
\label{alg:iter}
\end{algorithm}

\subsection{Summary of the results: Comparison study with standard gradient descent}
\label{sec:summary-results-taylor-gradient}

In this section, we empirically investigate the performance of the surrogate approach and compare its resource consumption with that of the standard VQA approach.

\subsubsection{Set-up}
We consider the task of estimating the ground state of a Hamiltonian. The target Hamiltonian is a rescaled Heisenberg chain given by
\begin{equation}
    O = \frac{1}{3(n-1)} \sum_{i=1}^{n-1} \left( X_i X_{i+1} + Y_i Y_{i+1} + Z_i Z_{i+1} \right).
\end{equation}
The variational ansatz consists of \(L\) layers of single-qubit rotations followed by two-qubit \(ZZ\) gates of the form
\begin{equation}
    U(\alv) = \prod_{l=1}^L \left( \prod_{i=1}^n R_{ZZ}(\alpha_{i+2n,l}) \prod_{i=1}^n R_Z(\alpha_{i+n,l}) R_X(\alpha_{i,l}) \right),
\end{equation}
with \(\ket{\psi(\alv)} = U(\alv)\ket{0}^{\otimes n}\). The total number of parameters is \(m = 3 n L\).

We explore different strategies for allocating measurement shots for the surrogate; details are provided in Appendix~\ref{sec:iterative-taylor-order2-shot-allocation}. One approach is straightforward and easily generalizable to higher-order expansions, while the second is slightly more optimized for practical implementation. The key difference is that the optimized strategy evaluates each loss only once for repeated Hessian diagonal elements, whereas the naive approach recomputes the same loss separately for each diagonal element. In Appendix~\ref{sec:iterative-taylor-order2-shot-allocation}, we also compare these methods with standard gradient descent under an equivalent total shot budget. 
Additionally, we employ a two-phase training schedule for the surrogate, where the number of iterations between surrogate updates, \(N_{\rm iters}\), is adjusted once during the simulation (see Appendix~\ref{sec:two-phase-surrogate-schedule}).

\subsubsection{Results discussion}

\paragraph*{Numerical details in simulations:} 
Let \emph{Surrogate 1} and \emph{Surrogate 2} denote the na\"ive and the optimized surrogate strategies, respectively, and let \emph{standard GD} denote standard gradient descent with an equivalent quantum resource budget.
As detailed in Appendix~\ref{sec:iterative-taylor-order2-shot-allocation}, our shot allocation is fully determined by fixing the number of shots \textit{per loss evaluation}, \(F_2\) - that is, the total number of shots to compute the gradient is given by $2m F_2$. In our experiments, we set  \(F_2 = 600\). During the first training phase, we use \(N_{\rm iters}^{(1)} = 10\) and \(N_{\rm surr}^{(1)} = 20\). In the second phase, we set \(N_{\rm iters}^{(2)} = 8\) and \(N_{\rm surr}^{(2)} = 200\), corresponding to a total of \(1800\) GD steps. At each GD step, the test loss is computed as the difference between the exact expectation value of the observable and the ground state energy obtained numerically using \textit{scipy.sparse.linalg.eigsh}.
All methods are run independently 10 times with different random initial parameters. Importantly, the three methods share the same initial parameters in each run. The learning rate is fixed at \(\eta = 0.1\).

\medskip \paragraph*{Results and Discussion:}

In Fig.~\ref{fig:iterative-taylor-10qubits-l3-6}, we show the test loss as a function of the number of shots (at depths (a) \(L=3\) and (c) \(L=6\)) and the number of iterations ((b) \(L=3\) and (d) \(L=6\)) for a 10-qubit Heisenberg Hamiltonian. We selected examples where the best final test loss is achieved either by standard GD (\(L=6\)) or by Surrogate 2 (\(L=3\)); however, the differences are small, particularly in terms of the interquartile range (IQR).  
We also compute the number of shots required to reduce the test loss below \(10^{-1}\) for \(L=2,3,4,5,6\). The results are shown in Fig.~\ref{fig:shots_to_threshold} for \(n=8\) (a) and \(n=10\) (b) qubits. Again, there is no clear distinction in performance between the strategies.

\begin{figure*}[h!]
    \includegraphics[width=0.8\linewidth]{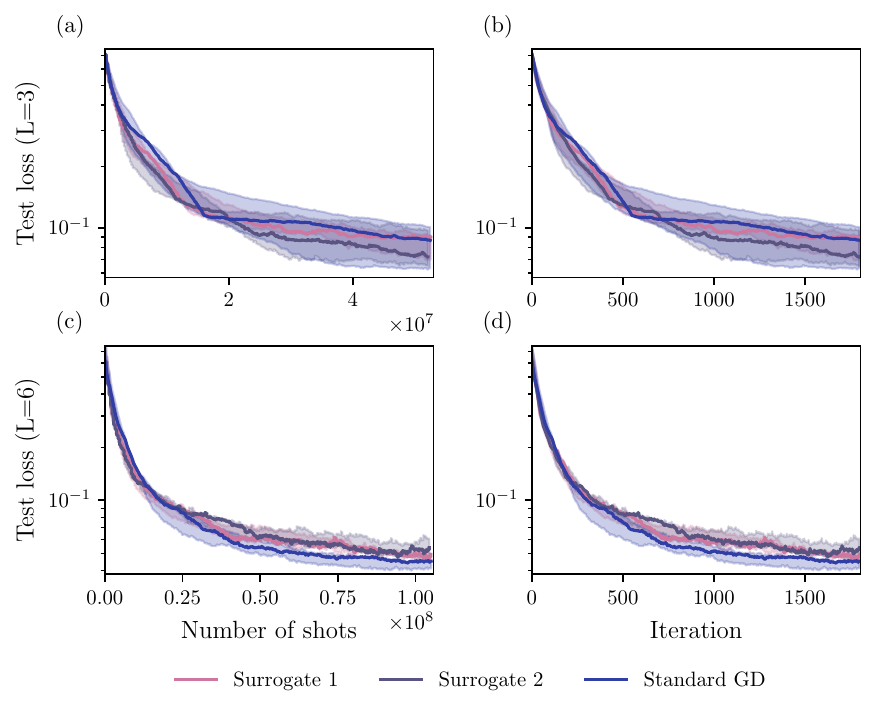}
    \caption{Test loss during training for the 10-qubit Heisenberg Hamiltonian as a function of quantum shots (left) and iterations (right) for \(L=3\) (top) and \(L=6\) (bottom). Surrogate 1 and 2 correspond to the na\"ive and optimized shot allocation, respectively, while standard GD uses an equivalent shot budget. Solid lines show the median test loss over 10 runs, and the shaded area indicates the interquartile range.}
    \label{fig:iterative-taylor-10qubits-l3-6}
\end{figure*}

\begin{figure*}[h!]
    \includegraphics[width=0.8\linewidth]{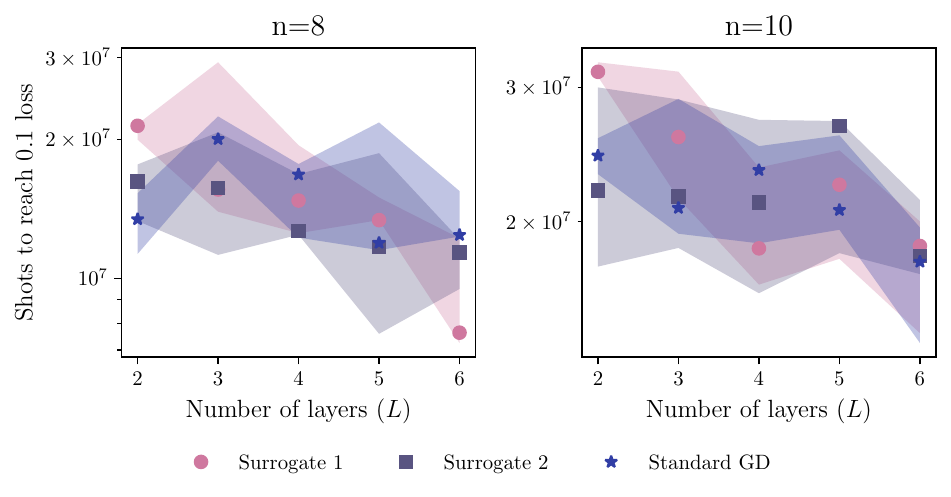}
    \caption{Number of shots required to reach \(f(\alv)\leq 0.1\) as a function of \(L\) for (a) \(n=8\) and (b) \(n=10\) qubits in the Heisenberg Hamiltonian. Surrogate 1 and 2 correspond to the na\"ive and optimized shot allocation, respectively. Standard GD uses an equivalent shot budget. Shaded areas indicate the interquartile range.}
    \label{fig:shots_to_threshold}
\end{figure*}

The empirical results observed from Fig.~\ref{fig:iterative-taylor-10qubits-l3-6} and Fig.~\ref{fig:shots_to_threshold} show that gradient descent (GD) and the surrogate-based approaches exhibit comparable performance in terms of quantum resource usage, although both methods achieve good accuracy for the tasks considered. 
Intuitively, while standard GD suffers from larger shot noise, its stochastically unbiased nature compensates for this drawback. 
Nevertheless, this does not rule out the possibility that more sophisticated surrogate approaches could yield a distinct performance advantage. Indeed, early evidence for a Fourier-based surrogate method has been reported in Ref.~\cite{koczor2020quantum}. A plausible explanation is that such approaches preserve the periodicity of the loss function and retain partial contributions from higher-order perturbative terms, although this remains a hypothesis. 

From a practical perspective, the key distinction is that GD consumes quantum resources at every step, whereas the surrogate expends a larger quantum budget up front and then reuses the resulting information across multiple iterations. This can be advantageous in two separate ways. First, it can better tolerate noise drift on real devices: by amortizing quantum evaluations into fewer, more substantial data-collection phases, it reduces the need to repeatedly query the hardware under slowly changing calibration conditions. As a result, successive optimization updates are less exposed to time-dependent variations in gate errors and measurement bias that can otherwise introduce inconsistent gradient estimates. Second, on cloud-based quantum platforms, it can mitigate the impact of variable queueing delays and inflated wall-clock runtimes that arise when jobs are repeatedly submitted.

\subsubsection{Further numerical experiment: Training inside a single patch}
\label{sec:training-inside-single-patch}

We now focus on a single surrogate and investigate how many GD steps can be performed from random initial parameters before the test loss begins to increase. The goal is to identify practically relevant values for $N_{\rm iters}$.

\paragraph*{Simulation details.} 
We fix \(F_2 = 600\) and \(\eta = 0.1\) as in previous simulations and consider the na\"ive surrogate method. The gradient is surrogated around random initial parameters \(\alv^*\), and GD is performed classically until the test loss starts to increase. We record the maximal value of \(N_{\rm iters}\) for each run. 
We set \(n = L\) and perform \(n_{\rm run} = 20\) independent runs for each \(n = L \in \{2,3,\dots,10\}\). 

\paragraph*{Results and discussion.} 
The results are shown in Fig.~\ref{fig:classical_training_one_patch}, which plots the optimal number of iterations \(N_{\rm iters}\) as a function of the number of parameters \(m\). The shaded area corresponds to the minimum–maximum range over all \(n_{\rm run} = 20\) random initializations, the diamond markers indicate the mean, and the solid line shows a power-law fit (linear fit in log–log scale).  
The minimal value observed is \(N_{\rm iters} = 9\), suggesting that the choice of \(N_{\rm iters} = 8\text{--}10\) in previous simulations was reasonable. These values were set heuristically based on preliminary tests. In practice, the test loss would not be accessible and so how to obtain the optimal \(N_{\rm iters}\) remains an open problem.
The corresponding distances to the center of the patch are shown in Fig.~\ref{fig:training-distance-patch} in the main text. One practical approach for selecting the number of iterations could be to restrict parameters within the patch of radius \(r \in \OC(1/\sqrt{m})\), where theoretical guarantees exist. However, Fig.~\ref{fig:training-distance-patch} indicates that the surrogate can be trained beyond this patch, meaning such a restriction would be suboptimal.

\begin{figure*}[h!]
    \includegraphics[width=0.6\linewidth]{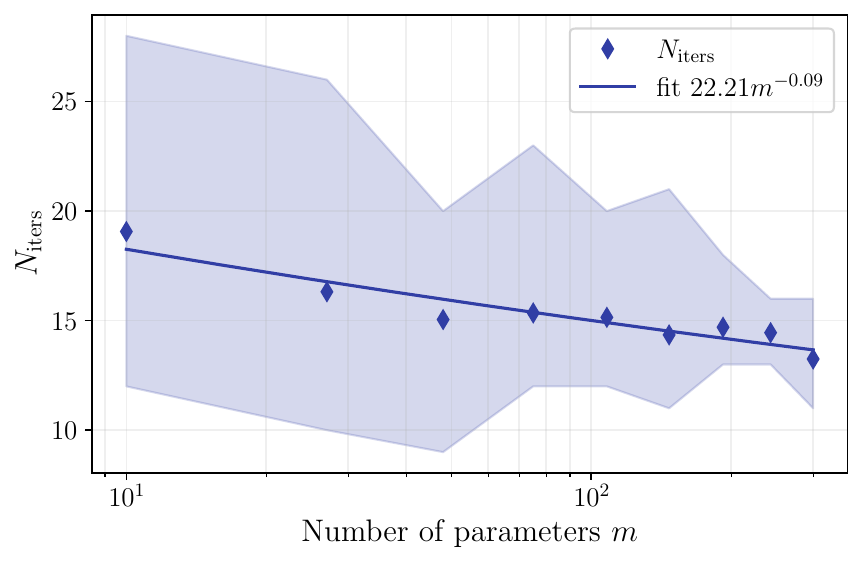}
    \caption{Maximal number of iterations that can be  performed using the surrogate before the test loss increases (while the surrogate loss continues to decrease). Shaded area shows the min–max range over \(n_{\rm run} = 20\) different initial random parameters, diamond markers indicate the mean, and the solid line is a power-law fitting curve (linear fit in log–log scale).}
    \label{fig:classical_training_one_patch}
\end{figure*}

\subsection{Further numerical details for Taylor surrogate}
\label{sec:taylor-iterative-numerical-details}

\subsubsection{Shot allocation in simulation}
\label{sec:iterative-taylor-order2-shot-allocation}

To account for finite sampling noise, we assign a number of measurement shots \(N(\delta\alv)\) to each loss evaluation \(f(\alv^*+\delta\alv)\). 
 Instead of explicitly sampling loss evaluations according to the optimal randomized measurement distribution of Lemma~\ref{lemma:multiple-observable-max-error}, we adopt a deterministic shot-allocation scheme. Specifically, we approximate the optimal randomized measurement process by fixing the shot allocation to its average over the optimal sampling distribution,
which preserves the same asymptotic scaling where the local patch size scales as \(r \sim \OC(1/m)\) while being marginally simpler to implement numerically. We fix a reference number of shots \(F_2\) for gradient evaluations (i.e., for shifts of the form $\pm \frac{\pi}{2} \vec{e}_l$) and define the remaining allocations proportionally:
\begin{align}
F_3 &\equiv \frac{F_2}{m}, \quad \text{(Hessian-related shifts)} \\
F_1 &\equiv \frac{F_2}{\sqrt{m}}, \quad \text{(central point evaluation)}.
\end{align}
This reflects the intuition that larger offsets or higher-order terms can be estimated more coarsely.

\paragraph*{(1) Na\"ive strategy (without periodicity)}
In this straightforward approach, all loss evaluations required by the parameter-shift rules are treated independently. In particular, the Hessian is estimated using four separate evaluations for each element:
\begin{equation}
    f\left(\alv^* \pm \frac{\pi}{2} (\vec{e}_l \pm \vec{e}_k)\right), \quad k,l=1,\dots,m.
\end{equation}
For \(k=l\), this includes multiple evaluations of \(f(\alv^*)\). Denoting \(N_{l,i}(0)\) as the number of shots used for the two independent central-point evaluations corresponding to the \(l\)-th diagonal term (\(i=1,2\)), we set \(N_{l,i}(0) = F_3\) for simplicity. The total number of shots required to construct one surrogate is then:
\begin{align}
N_{\text{surrogate}}^{\text{naive}} 
&= 2 m F_2 + 4 \cdot \frac{m(m+1)}{2} F_3 \\
&= 2 m F_2 + 2 m (m+1) F_3.
\end{align}
Substituting \(F_3 = F_2/m\) gives
\begin{equation}
    N_{\text{surrogate}}^{\text{naive}} = (4 m + 2) F_2.
\end{equation}

\paragraph*{(2) Optimized strategy (with periodicity)}
By exploiting the periodicity property
\begin{equation}
    f(\alv + 2 \pi \vec{e}_l) = f(\alv),
\end{equation}
the diagonal Hessian elements can be computed more efficiently:
\begin{equation}
    \partial_l^2 f(\alv^*) = \frac{f(\alv^* + \pi \vec{e}_l) - f(\alv^*)}{2}.
\end{equation}
Each diagonal element requires only one additional loss evaluation, while the central point \(f(\alv^*)\) is reused across all terms. The total number of shots per surrogate is then
\begin{align}
N_{\text{surrogate}}^{\text{opt}}
&= F_1 + 2 m F_2 + F_3 \left[m + 4 \cdot \frac{m(m-1)}{2} \right] \\
&= F_1 + 2 m F_2 + F_3 m (2 m - 1).
\end{align}
Substituting \(F_1 = F_2/\sqrt{m}\) and \(F_3 = F_2/m\) gives
\begin{equation}
    N_{\text{surrogate}}^{\text{opt}} = \left(4 m - 1 + \frac{1}{\sqrt{m}}\right) F_2.
\end{equation}

\medskip \paragraph*{Summary}
Both strategies scale similarly as \(\OC(m F_2)\), but the optimized strategy is slightly more shot-efficient due to reuse of the central evaluation and compression of diagonal Hessian evaluations. Nevertheless, the naive approach can slightly reduce shot noise in practice, resulting in comparable performance in terms of quantum resources.

\medskip \paragraph*{Standard gradient descent}
Standard gradient descent requires estimating the gradient at each iteration, which involves \(2 m\) loss evaluations via the parameter-shift rule. Let \(F_2^{\rm GD}\) denote the number of shots per evaluation. To ensure a fair comparison with the surrogate methods under an equivalent total shot budget, we set
\begin{equation} \label{eq:F2GD-general}
    F_2^{\rm GD} \approx \frac{N_{\text{surrogate}}^{\text{naive}}}{2 m N_{\rm iters}} 
    = \frac{\left(2 + \frac{1}{m}\right) F_2}{N_{\rm iters}}.
\end{equation}

\subsubsection{Two-phase surrogate schedule}
\label{sec:two-phase-surrogate-schedule}

We adopt a two-phase training schedule in which the frequency of surrogate updates changes over the course of optimization. Empirically, this heuristic improves the performance of the surrogate. Concretely, we define:
\begin{itemize}
    \item \( N_{\mathrm{surr}}^{(1)} \): number of surrogate refreshes in the first phase,
    \item \( N_{\mathrm{iters}}^{(1)} \): number of classical gradient steps between surrogates in the first phase,
    \item \( N_{\mathrm{surr}}^{(2)} \): number of surrogate refreshes in the second phase,
    \item \( N_{\mathrm{iters}}^{(2)} \): number of classical gradient steps between surrogates in the second phase,
\end{itemize}
with the typical choice
\[
N_{\mathrm{iters}}^{(1)} > N_{\mathrm{iters}}^{(2)}, \quad N_{\mathrm{surr}}^{(1)} < N_{\mathrm{surr}}^{(2)}.
\]
This adaptive schedule balances exploration and local accuracy: larger patches in the early phase encourage movement out of flat regions, while smaller patches in the later phase ensure the Taylor surrogate remains accurate near convergence.

When comparing the surrogate strategy to standard gradient descent under an equivalent total shot budget, we need to adapt Eq.~\eqref{eq:F2GD-general} to account for the two different values of \(N_{\rm iters}\). Specifically, we enforce
\begin{equation}
 N_{\text{surrogate}}^{\text{naive}}\,(N_{\text{surr}}^{(1)}+N_{\text{surr}}^{(2)}) 
 = 2 m F_2^{\rm GD1} \left( N_{\text{surr}}^{(1)} N_{\text{iters}}^{(1)} + N_{\text{surr}}^{(2)} N_{\text{iters}}^{(2)} \right).
\end{equation}
Since this equation may yield a non-integer value for \(F_2^{\rm GD}\), we round it as
\begin{equation}
     F_2^{\rm GD} = \left\lfloor (2 + 1/m) F_2 \frac{N_{\text{surr}}^{(1)} + N_{\text{surr}}^{(2)}}{N_{\text{surr}}^{(1)} N_{\text{iters}}^{(1)} + N_{\text{surr}}^{(2)} N_{\text{iters}}^{(2)}} \right\rfloor.
\end{equation}
This is equivalent to Eq.~\eqref{eq:F2GD-general} with \(N_{\rm iters}\) replaced by an effective average over the simulation.

\subsection{Analytical guarantees}
\label{sec:analytical-guarantee-second-order}
We now analyze how shot noise propagates through the surrogate-based \textit{gradient} estimation (not the shot noise analysis for the Taylor surrogate of the loss directly was covered in Appendix~\ref{sec:summary-results-taylor-gradient}). To understand the impact of statistical uncertainty, we construct a Taylor expansion of the surrogate gradient $\nabla \fsur(\alv)$ around a chosen expansion point $\alv^*$:
\begin{equation}
    \nabla \fsur(\alv) = \nabla f(\alv^*) + H(\alv^*)(\alv - \alv^*)\;,
\end{equation}
where $H(\alv^*)$ represents the Hessian of the loss function. Decomposing this into individual components for each parameter $l$, we obtain:
\begin{equation}
    \partial_l \fsur(\alv) = \partial_l f(\alv^*) + \sum_{k=1}^{m} \partial_{l,k} f(\alv^*) \, \delta_k, \quad \delta_k = \alpha_k - \alpha_k^*.
\end{equation}

To estimate these derivatives on a quantum processor, we employ the parameter-shift rule. This approach allows us to compute exact gradients and second-order derivatives by evaluating the circuit at specific shifts:
\begin{align}
    \partial_l f(\alv^*) &= \frac{f(\alv^* + \frac{\pi}{2}\vec{e}_l) - f(\alv^* - \frac{\pi}{2}\vec{e}_l)}{2}, \\
    \partial_{l,k} f(\alv^*) &= \frac{f(\alv^* + \frac{\pi}{2}(\vec{e}_l + \vec{e}_k)) - f(\alv^* + \frac{\pi}{2}(\vec{e}_l - \vec{e}_k)) - f(\alv^* - \frac{\pi}{2}(\vec{e}_l - \vec{e}_k)) + f(\alv^* - \frac{\pi}{2}(\vec{e}_l + \vec{e}_k))}{4}.
\end{align}
We define the shot noise error on the full gradient vector as the expected $L_2$-norm distance between the empirical estimate and the exact surrogate gradient (its mean):
\begin{equation}
    \Ebb\left[\norm{\grad\tilde{f}(\alv)-\Ebb\left[\grad\tilde{f}(\alv)\right]}_2^2\right]=\sum_{l=1}^m \Var[\partial_l\fsur(\alv)]\;.
\end{equation}

Assuming independent sampling noise for each evaluation, we propagate the variance through the parameter-shift expressions. Summing the contributions from the first-order shifts, the Hessian diagonals, and the off-diagonal cross-terms, we obtain 

\begin{align}
\label{eq:variance-shot-noise-2-nd-taylor-optimized}
    \sum_{l=1}^m \Var[\partial_l\fsur(\alv)]=&\sum_{l=1}^m\frac{\Var[f(\alv^*+\frac{\pi}{2}\vec{e}_l)]+\Var[f(\alv^*-\frac{\pi}{2}\vec{e}_l)]}{4}+\sum_{l=1}^m\sum_{k=1}^m\delta_k^2\frac{\sum_{\pm, \pm} \Var[f(\alv^* \pm \frac{\pi}{2}\vec{e}_l \pm \frac{\pi}{2}\vec{e}_k)]}{16}\;.
\end{align}

\noindent In the following, we perform the shot-noise analysis over these different scenarios:
\begin{itemize}
    \item Periodicity shot allocation for patch surrogate
    \item Naive shot allocation for patch surrogate, without periodicity,
    \item Comparison to standard gradient descents.
\end{itemize}

\subsubsection{Periodicity shot allocation for patch surrogate}

Here, we analyze the shot allocation in Eq.~\eqref{eq:variance-shot-noise-2-nd-taylor-optimized} using the periodicity condition. That is, by utilizing the periodicity of the loss, $f(\alv + 2\pi \vec{e}_k) = f(\alv)$, we can simplify the diagonal elements of the Hessian, reducing the number of required quantum circuit evaluations:
\begin{equation}
    \partial^2_l f(\alv^*) = \frac{f(\alv^* + \pi \vec{e}_l) - f(\alv^*)}{2}. \label{eq:2-nd-order-tylor-diagonal-hessian-periodicity}
\end{equation}

Let $N(\delta\alv)$ be the number of shots allocated for measuring $f(\alv^*+\delta\alv)$. The variance of the estimate is bounded by the single-shot variance:
\begin{equation} \label{eq:single-shot-variance-upperbound-taylor}
\Var[f(\alv^*+\delta\alv)]=\frac{\Var_1[f(\alv^*+\delta\alv)]}{N(\delta\alv)}\leq \frac{\norm{O}_\infty^2}{N(\delta\alv)}\;,
\end{equation}
where $\Var_1[\cdot]$ is upper-bounded by the squared norm of the observable $O$:
\begin{align}
\Var_1[f(\alv)]&=\bra{\psi(\alv)}O^2\ket{\psi(\alv)} - \bra{\psi(\alv)}O\ket{\psi(\alv)}^2 \\
&\leq\bra{\psi(\alv)}O^2\ket{\psi(\alv)} \\
& \leq\norm{O}_\infty^2\;.
\end{align}

By inserting Eq.~\eqref{eq:single-shot-variance-upperbound-taylor} and assuming a maximum patch radius $|\delta_l|\leq r$ into Eq.~\eqref{eq:variance-shot-noise-2-nd-taylor-optimized}, we find:
\begin{align}
    \frac{\sum_{l=1}^m \Var[\partial_l\fsur(\alv)]}{\norm{O}_\infty^2}&\leq \sum_{l=1}^m\left(\frac{1}{4N(\frac{\pi}{2}\vec{e}_l)}+\frac{1}{4N(-\frac{\pi}{2}\vec{e}_l)}+\frac{r^2}{4N(\pi\vec{e}_l)}\right) + \frac{mr^2}{4N(0)} \nonumber \\
    &+\sum_{l=1}^m\sum_{k=l+1}^m\frac{r^2}{8}\left(\sum_{\pm,\pm} \frac{1}{N(\pm\frac{\pi}{2}(\vec{e}_l\pm\vec{e}_k)))}\right)\;. 
\end{align}

Applying the shot allocation $(F_1, F_2, F_3)$ denoted in Appendix~\ref{sec:iterative-taylor-order2-shot-allocation}, the expression simplifies to:
\begin{align}
    \frac{\sum_{l=1}^m \Var[\partial_l\fsur(\alv)]}{\norm{O}_\infty^2}&\leq \frac{m}{2F_2}+\frac{m^2 r^2}{4 F_3}+\frac{mr^2}{4F_1} = \frac{m}{2F_2}\left(1+\frac{m^2r^2}{2}+\frac{\sqrt{m}r^2}{2}\right) \;. 
\end{align}
Considering a scaling of $r\in \OC(1/m)$, the shot noise scales as:
\begin{equation}
    \sum_{l=1}^m \Var[\partial_l\fsur(\alv)] \in\OC\left(\frac{\norm{O}_\infty^2\nparams}{F_2}\right)=\OC\left(\frac{\norm{O}_\infty^2\nparams^2}{N_{\rm surrogate}^{\rm opti}}\right)\;.
\end{equation}

\subsubsection{Naive shot allocation for patch surrogate without periodicity}
If we ignore periodicity and recompute the loss at the fixed point $\alv^*$ whenever it is required, the diagonal Hessian elements require 4 independent evaluations each. In this naive case, the total variance becomes:
\begin{align}
    \sum_{l=1}^m \Var[\partial_l \fsur(\alv)] &= \sum_{l=1}^m \frac{\Var[f(\alv^* + \frac{\pi}{2} \vec{e}_l)] + \Var[f(\alv^* - \frac{\pi}{2} \vec{e}_l)]}{4} + \sum_{l=1}^m \sum_{k=1}^m \delta_k^2 \cdot\frac{\sum_\pm\sum_\pm\Var[f(\alv^* \pm \frac{\pi}{2}\vec{e}_l \pm \frac{\pi}{2}\vec{e}_k)]}{16}.
\end{align}
Applying the same bounds and the allocation from Appendix~\ref{sec:iterative-taylor-order2-shot-allocation}, the error bound is:
\begin{align}
       \frac{\sum_{l=1}^m \Var[\partial_l\fsur(\alv)]}{\norm{O}_\infty^2} \leq \frac{m}{2F_2}+\frac{m^2 r^2}{4 F_3} = \frac{m}{2F_2}\left(1+\frac{m^2r^2}{2}\right) \;. 
\end{align}
While the bound is slightly lower, the total number of shots required is significantly higher. Thus, for $r\sim 1/m$, the shot noise maintains the same asymptotic scaling:
\begin{equation}
    \sum_{l=1}^m \Var[\partial_l\fsur(\alv)] \in\OC\left(\frac{\norm{O}_\infty^2\nparams^2}{N_{\rm surrogate}^{\rm naive}}\right)\;.
\end{equation}

\subsubsection{Comparison to standard GD}
For standard gradient descent (GD), we assume that each of the $2m$ loss evaluations for the gradient is estimated using $F_2^{\rm GD}$ shots. The total variance is upper-bounded by:
\begin{equation}
    \sum_{l=1}^m \Var[\partial_l f(\alv)] \leq \frac{\norm{O}_\infty^2m}{2F^{\rm GD}_2}\;.
\end{equation}
Recalling that $F_2^{\rm GD}\approx \frac{2}{N_{\rm iters}}F_2$, we see that the variance for a standard gradient evaluation is approximately $N_{\rm iters}/2$ times larger than that of the surrogate gradient:
\begin{equation}
    \sum_{l=1}^m \Var[\partial_l f(\alv)]\in\OC\left(\frac{\norm{O}_\infty^2\nparams N_{\rm iters}}{F_2}\right)\;.
\end{equation}

\section{Small angle Pauli Propagation}\label{app:smallanglePP}
This section is devoted to the classical simulation and surrogation of quantum circuits in the near-Clifford regime.

\subsection{The Pauli Propagation framework}
\label{sec:def-Pauli-path-expectation-function}

We start by briefly reviewing the Pauli-path formalism and how it can employed for classically estimating expectation values of quantum circuits. For further details, see also Refs.~\cite{fontana2023classical, rudolph2023classical, aharonov2022polynomial, angrisani2024classically} and references therein.

\medskip

Let $U(\vec{\alpha})$ be a quantum circuit written as parameterized perturbations from Clifford circuits. That is, of the form
\begin{equation}\label{eq:CliffordVQAcircuit-Appendix}
	U(\alv) = \prod_{i=1}^\nparams C_i U_i(\alpha_i)\,,
\end{equation}
where $\left\{ C_i \right\}_{i=1}^\nparams$ are a set of fixed Clifford operators and $\left\{ U_i(\alpha_i) = e^{-i\alpha_i P_i/2} \right\}_{i=1}^\nparams$ are parameterized rotations generated by the Pauli operators $\{ P_i \}_{i=1}^\nparams$. 

Given an observable $O = \sum_{P\in \mathcal{P}_n} a_P P$  and an initial state $\rho$, we want to approximate the following expectation value:
\begin{align}
    f(\vec{\alpha}) = \Tr[OU(\alv)\rho U^\dag( \alv )]\,.
\end{align}

For the sake of simplicity,  we will consider observables consisting of a single Pauli term. In the following lemma, we show that this assumption comes with no loss of generality.

\begin{lemma}[Transferring guarantees from Pauli observables]
\label{lem:trans}
Let $\mathcal{S}\subseteq\{I,X,Y,Z\}^{\otimes n}$ be a subset of Pauli operators and let $O = \sum_{P\in \mathcal{S}} a_P P$ be an observable.
\begin{enumerate}
\item \textbf{Worst-case error.} For all $P \in \mathcal{S}$ let $\fsur^{(P)}(\vec{\alpha})$ be a function satisfying
\begin{align}
    \abs{\fsur^{(P)}(\vec{\alpha}) - \Tr[P U(\vec{\alpha})\rho U(\vec{\alpha})^\dag]}\leq \epsilon.
\end{align}
Then the function $\fsur(\vec{\alpha}) = \sum_{P \in \mathcal{S}} a_P  \fsur^{(P)}(\vec{\alpha}) $ satisfies the following
\begin{align}
     \abs{\fsur(\vec{\alpha}) - f(\vec{\alpha})} \leq \norm{\boldsymbol{a}}_1 \epsilon.
\end{align}
\item \textbf{Mean Square Error.} Given a distribution $\mathcal{D}$ over $[-\pi, \pi]^m$, for all $P \in \mathcal{S}$ let $\fsur^{(P)}(\vec{\alpha})$ be a function satisfying
\begin{align}
    \Ebb_{\vec{\alpha}\sim \mathcal{D}} \abs{\fsur^{(P)}(\vec{\alpha})- \Tr[P U(\vec{\alpha})\rho U(\vec{\alpha})^\dag]}^2 \leq \epsilon^2.
\end{align}
Then the function $\fsur(\vec{\alpha}) = \sum_{P \in \mathcal{S}} a_P  \fsur^{(P)}(\vec{\alpha})$ satisfies the following
\begin{align}
    \Ebb_{\vec{\alpha} \sim \mathcal{D}} \abs{\fsur(\vec{\alpha}) - f(\vec{\alpha})}^2 \leq \norm{\boldsymbol{a}}_1^2 \epsilon^2.
\end{align}
\end{enumerate}
\end{lemma}
\begin{proof}
    The first part of the lemma follows by a simple application of the triangle inequality.
    \begin{align}
         &\abs{\fsur(\vec{\alpha}) - f(\vec{\alpha})}  = \abs{\sum_{P \in \mathcal{S}} a_P  \left(\fsur^{(P)}(\vec{\alpha})- \Tr[P U(\vec{\alpha})\rho U(\vec{\alpha})^\dag] \right)}
         \\\leq &\sum_{P \in \mathcal{S}} \abs{a_P} \abs{ \fsur^{(P)}(\vec{\alpha})- \Tr[P U(\vec{\alpha})\rho U(\vec{\alpha})^\dag] } \leq  \norm{\boldsymbol{a}}_1 \epsilon.
    \end{align}
As for the second part, we have
\begin{align}
     \abs{\fsur(\vec{\alpha}) - f(\vec{\alpha})}^2 
    =&\abs{\sum_{P \in \mathcal{S}} a_P  \left(\fsur^{(P)}(\vec{\alpha})- \Tr[P U(\vec{\alpha})\rho U(\vec{\alpha})^\dag] \right)}^2
    \\\leq &\left(\sum_{P \in \mathcal{S}} {\abs{a_P}}  \abs{\fsur^{(P)}(\vec{\alpha})- \Tr[P U(\vec{\alpha})\rho U(\vec{\alpha})^\dag]}\right)^2
    \\= &\left(\sum_{P \in \mathcal{S}} \sqrt{\abs{a_P}} \sqrt{\abs{a_P}}  \abs{\fsur^{(P)}(\vec{\alpha})- \Tr[P U(\vec{\alpha})\rho U(\vec{\alpha})^\dag]}\right)^2
    \\ \leq &\left(\sum_{P \in \mathcal{S}} \abs{a_P} \right) \left(\sum_{P \in \mathcal{S}}  \abs{a_P} \abs{\fsur^{(P)}(\vec{\alpha})- \Tr[P U(\vec{\alpha})\rho U(\vec{\alpha})^\dag]}^2\right)
    \\ = & \norm{\boldsymbol{a}}_1 \left(\sum_{P \in \mathcal{S}}  \abs{a_P} \abs{\fsur^{(P)}(\vec{\alpha})- \Tr[P U(\vec{\alpha})\rho U(\vec{\alpha})^\dag]}^2\right),
\end{align}
where in the second step we used the triangle inequality and in the fourth step we used the Cauchy-Schwarz inequality -- i.e. $( \sum_i x_i y_i)^2 \leq (\sum_i x_i^2)(\sum_i y_i^2)$.
We conclude the proof by taking the expectation over $\boldsymbol{\alpha}$.
\begin{align}
   \Ebb_{\vec{\alpha}\sim \mathcal{D}}  \abs{\fsur(\vec{\alpha}) - f(\vec{\alpha})}^2 = &
   \Ebb_{\vec{\alpha}\sim \mathcal{D}} \norm{\boldsymbol{a}}_1 \left(\sum_{P \in \mathcal{S}}  \abs{a_P} \abs{\fsur^{(P)}(\vec{\alpha})- \Tr[P U(\vec{\alpha})\rho U(\vec{\alpha})^\dag]}^2\right)
   \\= &\norm{\boldsymbol{a}}_1 \left(\sum_{P \in \mathcal{S}}  \abs{a_P} \underbrace{\Ebb_{\vec{\alpha}\sim \mathcal{D}}\abs{\fsur^{(P)}(\vec{\alpha})- \Tr[P U(\vec{\alpha})\rho U(\vec{\alpha})^\dag]}^2}_{\leq \epsilon^2}\right)
   \\\leq &  \norm{\boldsymbol{a}}_1^2 \epsilon^2,
\end{align}
where in the second step we used the fact that the expectation is linear.
\end{proof}

Thus, given a Pauli operator $P$, we consider the back-propagated (i.e., Heisenberg-evolved) observable $U(\vec{\alpha})^\dag P U(\vec{\alpha})$ and express it in the Pauli-path formalism
\begin{align}
    U(\vec{\alpha})^\dag P U(\vec{\alpha}) = \sum_{\omv} \Phi_{\omv}(\vec{\alpha}) P_{\omv},
\end{align}
where $\omv\in\{-1, 0, 1\}^m$ is the so-called \textit{frequency vector} or \textit{path vector} that uniquely enumerates a Pauli path, and $P_{\omv}$ is the backpropagated Pauli along the path $\omv$ .
If at the Pauli gate $U_i$  in Eq.~\eqref{eq:CliffordVQAcircuit} the backpropagating Pauli operator commutes with with $P_i$, then $\omega_i=0$. If it anticommutes, $\omega_i= 1$ or $\omega_i= -1$ depending on whether the Pauli path picked up a cosine or sine coefficient, respectively. $\Phi_{\omv}(\thv)$ thus represents the coefficient of the backpropagated Pauli path and is defined as 
\begin{equation}
    \Phi_{\omv}(\thv) = \prod_{i=1}^\nparams \phi_{\omega_i}(\alpha_i),
\end{equation}
where
\begin{equation}\label{eq:pauli_path}
    \phi_{\omega_i}(\alpha_i)= \begin{cases}
    1,                      & \text{if } \omega_i = 0\\
    \cos(\alpha_i),        & \text{if } \omega_i = 1\\
    \sin(\alpha_i),        & \text{if } \omega_i = -1 
\end{cases}\,.
\end{equation}
Furthermore, we define
\begin{equation}\label{eq:definition_d_omega}
    d_{\omv} := \Tr[\rho P_{\omv}]\in[-1, 1]
\end{equation}
as the overlap of the backpropagated Pauli along the path $\omv$ with the initial state, and we denote the expectation function as
\begin{align}
    f^{(P)}(\thv) = \Tr[P U(\thv)\rho U^\dag(\thv)] = \sum_{\omv} \Phi_{\omv}(\vec{\alpha}) d_{\omv}.
\end{align}
We additionally define
\begin{equation}\label{eq:morethanksins}\mathcal{S}_\expansionorder^{\nparams}:=\{\text{ paths } \omv |\ \omv \text{ contains more than \expansionorder negative entries} \} \subseteq \{-1,0,1\}^{\nparams}
\end{equation}
as the set of paths $\omv$ with more than $\expansionorder$ entries of $-1$, i.e., the set of paths $\omv$ with more than $\expansionorder$ sine coefficients. For simplicity, we define 
\begin{equation} \label{eq:num_sins}
    \#\sin(\omv)\; := \text{number of negative entries of } \omv\;,
\end{equation}
which is also the number of sine coefficients of the path $\omv$.
Finally, we introduce the truncated function $\fsur_{\expansionorder}^{(P)}(\thv)$ which contains only the terms with at most $\expansionorder$ sines
\begin{align}
    \fsur_{\expansionorder}^{(P)}(\thv) \coloneqq \sum_{\omv \not\in \mathcal{S}_\expansionorder^{\nparams}} \Phi_{\omv}(\vec{\alpha}) d_{\omv}\;.
\end{align}
Similarly, we will use the notation $\widetilde{P}(\alv)$ to denote the truncated version of the back-propagated operator $P$ given by
\begin{equation}
    \widetilde{P}(\alv) := \sum_{\omv \not\in \mathcal{S}_\expansionorder^{\nparams}} \Phi_{\omv}(\vec{\alpha}) P_{\omv}\;.
\end{equation}
In the following, we will consider two regimes:
\begin{itemize}
    \item \textbf{Average-case.} When the parameter vector $\thv$ is equipped with a suitable distribution $\mathcal{D}$ over the hypercube $\vol(\mathbf{0}, r) $, we consider the mean square error
    \begin{align}\label{eq:l2error}
    \Ebb_{\thv \sim \mathcal{D}}\left[ \left(f^{(P)}(\thv) - \fsur_{\expansionorder}^{(P)}(\thv)\right)^2 \right]
    =\Ebb_{\thv \sim \mathcal{D}} \left[\left(\sum_{\omv \in \mathcal{S}_\expansionorder^{\nparams}} \Phi_{\omv}(\vec{\alpha}) d_{\omv}\right)^2\right].
\end{align}
\item \textbf{Worst-case.} When the parameter vector $\thv$ is an arbitrarily chosen point in the hypercube $\vol(\mathbf{0}, r) $, we consider the worst-case error
\begin{align}
    \sup_{\vol(\bold{0}, r) } \abs{f^{(P)}(\thv) - \fsur_{\expansionorder}^{(P)}(\thv)} 
    =\sup_{\vol(\bold{0}, r) } \abs{\sum_{\omv \in \mathcal{S}_\expansionorder^{\nparams}}  \Phi_{\omv}(\vec{\alpha}) d_{\omv}}.
\end{align}
\end{itemize}

\subsection{Uncorrelated parameters}\label{apx:average_uncorr}
In this section we consider the case where all the parameters $\alpha_i$ are sampled independently from some probability distributions $\mathcal{D}_i$ supported in $[-\pi, \pi]$, and therefore  the parameter vector $\thv = (\alpha_1,\alpha_2,\dots, \alpha_m)$ is sampled from the product distribution $\mathcal{D}:=\mathcal{D}_1 \times \mathcal{D}_2 \times \dots \mathcal{D}_m$.
For all $i$, we make following key assumptions on the distribution $\mathcal{D}_i$.
\begin{enumerate}
    \item \textbf{Symmetric distribution.} 
    We assume $\mathcal{D}_i$ to be symmetric and centered around zero, i.e. for all $0\leq a\leq b \leq \pi$
    \begin{align}
        \Pr_{\alpha_i \sim \mathcal{D}_i}\bigg[\alpha_i \in [a,b]\bigg] =\Pr_{\alpha_i \sim \mathcal{D}_i}\bigg[\alpha_i \in [-b,-a]\bigg] 
    \end{align}
    This assumption ensures that the odd function $\cos(\alpha_i)\sin(\alpha_i)$ has mean zero
    \begin{align}
        \Ebb_{_{\alpha_i \sim \mathcal{D}_i}} \cos(\alpha_i)\sin(\alpha_i) = 0.
    \end{align}
    \item \textbf{Variance bound.} The variance of the function $\sin(\alpha_i)$ is bounded by an appropriately chosen value $\sigma^2$
    \begin{align}
        &\Ebb_{_{\alpha_i \sim \mathcal{D}_i}} \sin^2(\alpha_i)\leq \sigma^2.
    \end{align}
\end{enumerate}
In particular, these two assumptions are satisfied by a truncated Gaussian distribution with mean $0$ and variance $\sigma^2$ supported in $[-\pi, \pi]$, and by any symmetric distribution supported entirely in the range $[-\sqrt{3}\sigma, \sqrt{3}\sigma]$.

Crucially, the symmetry assumption ensures that the different paths are ``orthogonal'', in the sense that the associated coefficients are uncorrelated
\begin{align}
    \omv'\neq \omv \implies \Ebb_{\thv \sim \mathcal{D}} \, \Phi_{\omv}(\thv) \Phi_{\omv'}(\thv) = 0.
\end{align}
This occurs when for any $i$ one has $\omega_i = 1, \omega^\prime_i = -1$, or vice versa. Given that all paths stem from the same Pauli operator in the Pauli decomposition of the observable $O$, two different paths must have split off from another, collecting opposing coefficients. This implies that $ \Ebb_{\thv \sim \mathcal{D}} \, \Phi_{\omv}(\thv) \Phi_{\omv'}(\thv) = C_{\omv \setminus \omv_i} \cdot \Ebb_{_{\alpha_i \sim \mathcal{D}_i}} \cos(\alpha_i)\sin(\alpha_i)=0$, where $C_{\omv \setminus \omv_i}$ are the integrals over all other independent parameters $\alpha_j, j\neq i$.

The above observation allows us to drastically simplify the expression of the mean square error:
\begin{align}
    \Ebb_{\thv \sim \mathcal{D}} \left(\sum_{\omv \in \mathcal{S}_\expansionorder^{\nparams}} \Phi_{\omv}(\vec{\alpha}) d_{\omv}\right)^2
    & = \sum_{\omv \in \mathcal{S}_\expansionorder^{\nparams}} \Ebb_{\thv \sim \mathcal{D}} \Phi_{\omv}^2(\vec{\alpha}) d_{\omv}^2 \\
    &\leq \sum_{\omv \in \mathcal{S}_\expansionorder^{\nparams}} \Ebb_{\thv \sim \mathcal{D}} \Phi_{\omv}^2(\vec{\alpha}),
\end{align}
where we used the fact that  $d_{\omv}^2 = \Tr[P_{\omv}\rho]^2\leq 1$ for all $\omv$.
Leveraging the variance bound, we can upper bound the mean square error as follows.

\begin{theoremsup}[Mean square error for uncorrelated parameters]
\label{thm:mse-avg-supp}
Let $\mathcal{D}$ satisfy the assumptions stated above with $\sigma^2 \leq \expansionorder/\nparams$. Then we have
\begin{align}
    \Ebb_{\thv \sim \mathcal{D}} \left[\left(f^{(P)}(\thv) - \fsur_{\expansionorder}^{(P)}(\thv)\right)^2\right]\leq \left(\frac{e\nparams \sigma^2}{\expansionorder}\right)^{\expansionorder}.
\end{align}
\end{theoremsup}
\begin{proof}
  We start by observing that
\begin{align}
    \Ebb_{\thv \sim \mathcal{D}} \Phi_{\omv}^2(\vec{\alpha})  =\prod_{i=1}^{\nparams} \begin{cases}
    1                      & \text{if } \omega_i = 0\\
    \Ebb_{\alpha_i \sim \mathcal{D}_i} \cos(\alpha_i)^2 \leq 1        & \text{if } \omega_i = 1\\
   \Ebb_{\alpha_i \sim \mathcal{D}_i} \sin(\alpha_i)^2 \leq \sigma^2       & \text{if } \omega_i = -1 
\end{cases}\,.
\end{align}

In order to upper bound the mean square error, we need to estimate how many sine terms are contained in each path $\omv$. Unfortunately, this is highly problem-dependent. For a generic upper bound, we thus choose a maximally-splitting problem, where any truncation induces the largest loss in probability mass. That is, we assume that the splitting tree yields $2^{\nparams}$ paths with all combinations of $\sin$ and $\cos$ prefactors. We can thus write
\begin{align}
    \Ebb_{\thv \sim \mathcal{D}} \left[\left(f^{(P)}(\thv) - \fsur_{\expansionorder}^{(P)}(\thv)\right)^2\right] & \leq \sum_{\omv \in \mathcal{S}_\expansionorder^{\nparams}} \Ebb_{\thv \sim \mathcal{D}} \Phi_{\omv}^2(\vec{\alpha}) 
    \leq \sum_{\substack{\omv\in \{1, -1\}^{\nparams}\\ \#\sin(\omv) \geq \expansionorder}} \Ebb_{\thv \sim \mathcal{D}} \Phi_{\omv}^2(\vec{\alpha}) 
    \\&\leq  \sum_{i=\expansionorder}^{\nparams} \binom{\nparams}{i} \left(\sigma^{2}\right)^i
    \leq  \left(\frac{e\nparams \sigma^2}{\expansionorder}\right)^{\expansionorder},
\end{align}
where $\#\sin(\omv)$ is defined in Eq.~\eqref{eq:num_sins} and the last step follows from Lemma~\ref{lem:bin-up} since by assumption $\sigma^2\leq \expansionorder/\nparams$.
  
\end{proof}

\paragraph*{Uniformly sampled parameters.} We now consider the special case of particular interest when all the parameters $\alpha_i$ are sampled independently from the uniform distribution over the symmetric interval $[-r,r]$, where $r>0$ is a suitable small constant.

Averaging over the small angle range $[-r, r]$ leads to 
\begin{align}
\Ebb_{\alpha_i \sim \mathcal{D}_i} \sin(\alpha_i)^2  =  \frac{1}{2r}\int_{-r}^r \sin(\alpha_i)^2  d\alpha_i = 
\frac{1}{2}\left(1 - \frac{\sin(2r)}{2r}\right) \leq \frac{r^2}{3}\,.
\end{align}
From the previous, we can upper bound the mean square error by setting $\sigma^2 = r^2/3$ and invoking \supt~\ref{thm:mse-avg-supp}.
\begin{corollary}[Mean square error for the uniform distribution]
\label{cor:MSE-uniform-uncorrelated}
Let $\mathcal{D}$ be the uniform distribution over $[-r,r]^m$, where $r^2 \leq 3\expansionorder/\nparams$. Then we have
\begin{align}
    \Ebb_{\thv \sim \mathcal{D}} \left(f^{(P)}(\thv) - \fsur_{\expansionorder}^{(P)}(\thv)\right)^2\leq \left(\frac{e\nparams r^2}{3\expansionorder}\right)^{\expansionorder}.
\end{align}    
\end{corollary}
In particular, for any $r \leq \sqrt{{1}/{\nparams}}$, this bound decreases faster than exponentially with $\expansionorder$.

\subsection{Worst-case analysis}\label{apx:worst}
In this section, we consider the parameter vector $\vec{\alpha}=(\alpha_1, \alpha_2,\dots, \alpha_m)$ to be arbitrarily chosen -- possibly adversarially -- in the hypercube $[-r, r]^m$.  Thus, we drop the ``symmetry assumption'' and we don't assume the different $\alpha_i$ to be independently sampled.

We consider the maximum error over the hypercube:

\begin{align}
     \max_{\vec{\alpha}\in[-r,r]^\nparams}\left|f^{(P)}(\thv) - \fsur_{\expansionorder}^{(P)}(\thv)\right|.
\end{align}

We provide the following upper bound.
\begin{theoremsup}[Worst-case error]
\label{thm:worst-case-error-single-Pauli}
Let  $r \leq \expansionorder/\nparams$. Then we have
\begin{align}
    \max_{\vec{\alpha}\in[-r,r]^\nparams}\left|f^{(P)}(\thv) - \fsur_{\expansionorder}^{(P)}(\thv)\right| \leq \left(\frac{e\nparams r}{\expansionorder}\right)^{\expansionorder}.
\end{align}
Thus, for any $r\leq 1/\nparams$, the absolute error decreases faster than exponentially in $\expansionorder$.
\end{theoremsup}

\begin{proof}
  We start by upper bounding the error as follows
  \begin{align}
     \left|f^{(P)}(\thv) - \fsur_{\expansionorder}^{(P)}(\thv)\right| =  \left| \sum_{\omv\in \mathcal{S}_{\expansionorder}^{\nparams}} \Phi_{\omv}(\thv)d_{\omv} \right| 
    \leq \sum_{\omv\in \mathcal{S}_{\expansionorder}^{\nparams}} \left|\Phi_{\omv}(\thv)\right| \,,
\end{align}
where we used the triangle inequality and the fact $|d_{\omv}| =\Tr[P_{\omv}\rho] \leq 1$.

Proceeding as in the proof of \supt~\ref{thm:mse-avg-supp}, we obtain
\begin{align}
    &\sum_{\omv \in \mathcal{S}_\expansionorder^{\nparams}} \abs{\Phi_{\omv}(\vec{\alpha})} 
    \leq \sum_{\substack{\omv\in \{1, -1\}^{\nparams}\\ \#\sin(\omv) \geq \expansionorder}}\abs{\Phi_{\omv}(\vec{\alpha})} 
    \\\label{eq:upperbound-path-coefficient-sum}  \leq  &\sum_{\substack{\omv\in \{1, -1\}^{\nparams}\\ \#\sin(\omv) \geq \expansionorder}} r^{\#\sin(\omv)} =\sum_{i=\expansionorder}^{\nparams} \binom{\nparams}{i} r^i
    \leq  \left(\frac{e\nparams r}{\expansionorder}\right)^{\expansionorder},
\end{align}
where we used the fact that $\abs{\sin(\alpha_i)}\leq r$, $\abs{\cos(\alpha_i)}\leq 1$ and we applied Lemma~\ref{lem:bin-up} in the last step. Recall that $\#\sin(\omv)$ is defined in Eq.~\eqref{eq:num_sins}. Notice that the proof is also valid if the $\alpha_i$'s are correlated and provides also a bound for the $1$-norm of the coefficients of the paths.

\end{proof}

\subsection{Circuits with correlated parameters}
\label{apx:average_correlated}

Given the apparent difference in asymptotic scalings between the average case of parameters randomly sampled from a distribution and the worst case, we can ask the question: \textit{How hard is it to simulate fundamentally interesting circuits?} Here we consider the example of Trotter circuits for simulating quantum dynamics. The scenario is that of perfectly correlated angles, as the same  parameter  is used across all of the gates.  This is the case when all coefficients in a Hamiltonian are equal. Clearly, our earlier worst-case error also bounds this case, but we will now consider the \textit{average case of correlated parameters}.

\begin{proposition}[Improved bound for correlated parameters]
Let $\mathcal{D}$ the distribution over $\vol(\cenv, r)$ obtained by sampling a parameter $\alpha_1$ uniformly at random over $[-r,r]$ and setting $\alv = (\alpha_1,\alpha_1,\dots, \alpha_1)$.
Then the average error satisfies the following upper bound:
\begin{align}
    \Ebb_{\alv} \left|f^{(P)}(\alv) - \fsur_{\expansionorder}^{(P)}(\alv)\right| \leq \frac{\left(\frac{e\nparams r}{\expansionorder}\right)^{\expansionorder}}{\expansionorder+1}.
\end{align}
\end{proposition}
\begin{proof}
  Here we assume a quantum circuit consisting of Clifford and Pauli gates, where the Pauli gates are all parametrized by the identical angle $\alpha_1$. The average error is
\begin{align}
     \frac{1}{2r}\int_{-r}^r \left|f^{(P)}(\alv) - \fsur_{\expansionorder}^{(P)}(\alv)\right| d\alpha_1 &= \frac{1}{2r}\int_{-r}^r \left| \sum_{\omv\in \mathcal{S}_{\expansionorder}^{\nparams}} \Phi_{\omv}(\alv)d_{\omv} \right| d\alpha_1 \\
    &\leq \frac{1}{2r}\int_{-r}^r\sum_{\omv\in \mathcal{S}_{\expansionorder}^{\nparams}}  \left|\Phi_{\omv}(\alv)d_{\omv} \right| d\alpha_1\\
    &\leq \sum_{\omv\in \mathcal{S}_{\expansionorder}^{\nparams}} \frac{1}{2r}\int_{-r}^r |\alpha_1|^{\#\text{sin}(\omv)} d\alpha_1\\
    &= \sum_{\omv\in \mathcal{S}_{\expansionorder}^{\nparams}} \frac{1}{r}\int_{0}^r  \alpha_1^{\#\text{sin}(\omv)} d\alpha_1\,,
\end{align}
where we used the triangle inequality in the first inequality, and $|\sin(\alpha_1)| \leq |\alpha_1|, |\cos(\alpha_1)| \leq 1, d^2_{\omv}\leq 1$ for the second inequality with the same notation for the number of sins per path as in Eq.~\eqref{eq:num_sins}. In the last equality we use the fact that the integral is symmetric over the absolute value $\alpha_1$, to compute the integral over a positive value range.
Again assuming a maximally splitting tree behavior as the worst case yields
\begin{align}
    \sum_{\substack{\omv\in \{1, -1\}^{\nparams}\\ \#\sin(\omv) \geq \expansionorder}} \frac{1}{r}\int_{0}^r  \alpha_1^{\#\text{sin}(\omv)} d\alpha_1
    &\leq \sum_{i=\expansionorder}^\nparams \binom{\nparams}{i} \frac{1}{r}\int_{0}^r \alpha_1^i d\alpha_1\\
    &= \sum_{i=\expansionorder}^\nparams \binom{\nparams}{i} \frac{r^i}{i+1} \\
    &\leq \sum_{i=\expansionorder}^\nparams \binom{\nparams}{i} \frac{r^i}{\expansionorder+1} \;,
\end{align}
where in the end we used that $\frac{1}{i+1}\leq \frac{1}{\expansionorder+1}$ for $i\geq \expansionorder$. Finally, we can use Lemma~\ref{lem:bin-up} to get 
\begin{align}
    \frac{1}{2r}\int_{-r}^r \left|f^{(P)}(\alv) - \fsur_{\expansionorder}^{(P)}(\alv)\right| d\alpha_1 &\leq \frac{\left(\frac{e\nparams r}{\expansionorder}\right)^{\expansionorder}}{\expansionorder+1}\;,
\end{align}
which corresponds to the upper bound of the worst-case multiplied by a factor $\frac{1}{\expansionorder+1}$ (see \supt~\ref{thm:worst-case-error-single-Pauli}).
\end{proof}

Looking ahead to Appendix~\ref{apx:worst}, we see that the error scaling is the same for the worst case as in the average case with fully correlated angles. Whether this is a feature of our bounds or indicative of something more fundamental is left as an open question.

\subsection{Time complexity of the Pauli Propagation algorithm}\label{apx:runtim}
Recall that the circuit ansatz defined in Eq.~\eqref{eq:CliffordVQAcircuit} is of the form
\begin{align}
    U(\alv) = \prod_{i=1}^m C_i U_i(\alpha_i).
\end{align}
We start by upper bounding the time complexity of the Pauli Propagation algorithm as a function of the system size $n$, the number of Pauli rotations (and Clifford gates) $m$ and the truncation order $\expansionorder$.

\medskip

\noindent\underline{Upper bound on the number of paths}: In a quantum circuit with $\nparams$ Pauli rotations as in Eq.~\eqref{eq:CliffordVQAcircuit}, the total number of Pauli paths $n_{pth}^{(\nparams )}$ is upper bounded by $n_{pth}^{(\nparams)} \leq 2^\nparams$. This is achieved when the backpropagating Pauli operators split into two at every gate. As used throughout this appendix, the distribution of the number of sines in the $2^\nparams$ paths is given via
\begin{equation}
    n_{pth}^{(\nparams)} \leq \sum_{i=0}^\nparams \binom{\nparams}{i} = 2^\nparams  \,,
\end{equation}
where the index $i$ represents the number of sines. 

Given that we truncate paths with more than $\expansionorder$ sines, the number of paths $n_{pth}^{(\expansionorder)}$ that we need to compute is upper-bounded like
\begin{align}
    n_{pth}^{(\expansionorder)} \leq \sum_{i=0}^\expansionorder \binom{\nparams}{i} \leq \left( \frac{e \nparams}{\expansionorder}\right)^\expansionorder \, \label{eq:number-paths},
\end{align}
as stated in Eq.~\eqref{eq:cumbinom_ineq}. 

\medskip

\noindent\underline{Runtime $T$ for updating the Pauli operators:}
The runtime of the Pauli propagation algorithm is necessarily upper-bounded by a number proportional to the number of paths $n_{pth}^{(\expansionorder)}$. We call this proportionality factor $T$, and it includes the time complexity of updating a propagating Pauli operator when a gate is applied to it. More concretely, let $V$ be either a Clifford $C_i$ or a Pauli rotation $U_i(\alpha_i)$.
Given a Pauli operator $P\in \mathcal{P}_n$
we need to upper bound the time complexity for computing $V^\dag P V$.

Throughout the main text, we have assumed $T\in \OC(1)$ for simplicity, but  we here include the explicit dependence on $T$ for preciseness. Interestingly, the case of $T\in \OC(1)$ can be accurate under certain implementation and simulation assumptions. For example, $T\in \OC(1)$ is possible if $V$ acts on $\OC(1)$ qubits and the Pauli operator $P$ can be updated in-place on a subset of qubits. If in-place manipulation is not done, then the full operator needs to be copied which takes $\OC(n)$ time. In general, global Clifford gates can be decomposed into $\OC(n^2/\log(n))$ local Clifford gates~\cite{aaronson2004improved,gottesman1998heisenberg}, each taking $\OC(1)$ or $\OC(n)$ time (in-place or not). Furthermore, one could define $C_i$ as an arbitrary long concatenation of local Clifford gates without re-compilation, which can take an arbitrary large amount of time to execute. However, even global Pauli rotations can be applied in time $\OC(n)$. Thus, in summary, in the best case we have $T \in \OC(1)$ and in the worst $T \in \OC(n^3)$. Given that we assumed $m\in \OC(\in \text{poly}(n))$, the extra factor of $n^3$ can be absorbed by big $\mathcal{O}$ into our final time complexity. Hence, in the main text we take $T\in \OC(1)$.

\medskip

\noindent\underline{Runtime for a truncation order $\expansionorder$}:
The total number of gates applied to all the paths is
\begin{equation}\label{eq:runtime-single-Pauli_1}
     \nparams n_{pth}^{(\expansionorder)} \leq \nparams \left( \frac{e\nparams}{\expansionorder}\right)^\expansionorder.
\end{equation}
As we assumed that the time for computing the new Pauli terms is $T$, this leads to a total runtime $t$ of 
\begin{equation}\label{eq:runtime-single-Pauli}
     t \leq T\nparams  n_{pth}^{(\expansionorder)} \leq T\nparams \left( \frac{e\nparams}{\expansionorder}\right)^\expansionorder.
\end{equation}

\medskip
In the following, we provide an upper bound on the runtime required to achieve an error of $\epsilon$, considering both the cases of uncorrelated and worst-case parameters. Moreover, we will now consider more general observable $O=\sum_{P\in\PC_n}a_P P$ by combining previous analysis for single Paulis together with Lemma~\ref{lem:trans}.

\begin{theoremsup}[Time complexity of small-angle Pauli propagation (uncorrelated angles), Formal)]\label{thm:time-complexity-PP-classical-avg}
Let $O= \sum_{P\in\mathcal{P}_n} a_P P$ be an observable containing $N_{\mathrm{Paulis}}$ different Pauli terms and satisfying $\norm{\vec{a}}_1 = \OC(1)$. Furthermore, assume that given some $P$, we can efficiently estimate  $\Tr[\rho P]$. Let $T$ be the time necessary to apply a gate to a propagating Pauli operator.
Then small-angle Pauli propagation can be used to simulate a quantum circuit of the form of Eq.~\eqref{eq:CliffordVQAcircuit} over uniformly sampled zero-centered hypercube of width $2r$, i.e. $\uni(\vec{0}, r)$ with $r\in\OC(1/\sqrt{\nparams})$.
To ensure a mean square error \begin{equation}
   \Ebb_{\vec{\alpha}}\left[\left(f(\alv)-\fsur_{\expansionorder}(\alv) \right)^2\right]\leq \norm{\vec{a}}_1^2\left(\frac{e\nparams r^2}{3\expansionorder}\right)^{\expansionorder} =\epsilon^2\;, 
\end{equation} 
it suffices to run the algorithm for time
\begin{equation}
    t=N_{\mathrm{Paulis}} T \nparams \left(\frac{e\nparams}{\expansionorder}\right)^\expansionorder \in  \OC\left(N_{\mathrm{Paulis}} T \, {\nparams}^{\log(1/\epsilon^2)}\right). \, 
\end{equation}

\end{theoremsup}

\begin{proof}
To approximate the expectation value of the observable $O$, we will run the Pauli Propagation algorithm for all the Pauli operators in the Pauli expansion of $O$.
Let $\tilde{t}$ be the time required to estimate the expectation value of a single Pauli observable $P$ up to mean square error $\tilde{\epsilon}^2 = \epsilon^2/\norm{\vec{a}}_1^2$. 
Then, invoking Lemma~\ref{lem:trans}, we find that the time required to  to estimate the expectation value of $O$ up to mean square error $\epsilon^2$ is 
\begin{align}
    \tilde{t}\cdot N_{\mathrm{Paulis}}. 
\end{align}

    Given $r\in\OC(1/\sqrt{\nparams})$,  we introduce a constant $c\in\OC(1)$ such that $r=c/\sqrt{\nparams}$ and consider the distribution $\vec{\alpha}$ randomly sampled from $ \uni(\bold{0}, c/\sqrt{\nparams})$. Then, as long as the truncation order satisfies
\begin{align}
    \expansionorder \geq \frac{c^2}{3},
\end{align}
Corollary~\ref{cor:MSE-uniform-uncorrelated} implies that, for any given Pauli operator $P$, the Pauli Propagation algorithm returns an operator $\widetilde{P}_{\thv}$, which is a weighted sum of Pauli operators, satisfying:
\begin{align}
    \Ebb_{\vec{\alpha}\sim \uni(\vec{0}, r)} \abs{\Tr[ P U(\vec{\alpha}) \rho U^\dag(\vec{\alpha}) - \widetilde{P}_{\thv}\rho]}^2
      \leq \left(\frac{ec^2}{3\expansionorder}\right)^{\expansionorder}.
\end{align}
In order to achieve the desired error, we set
\begin{align}
    &\tilde{\epsilon}^2:= \left(\frac{ec^2}{3\expansionorder}\right)^{\expansionorder}
    \\&\implies \expansionorder \log \expansionorder \in \Theta(\log(1/\tilde{\epsilon}^2)).
\end{align}

As observed in Eq.~\eqref{eq:number-paths}, the number of paths found by the Pauli Propagation algorithm is at most
\begin{align}
    n_{pth}\leq \left(\frac{e\nparams}{\expansionorder}\right)^\expansionorder = \nparams^\expansionorder \cdot \tilde{\epsilon}^2\cdot\left(\frac{3} {c^2}\right)^\expansionorder \in  \OC(\nparams^\expansionorder \tilde{\epsilon}^2).
\end{align}
Therefore, the runtime for a single Pauli operator in Eq.~\eqref{eq:runtime-single-Pauli} is at most
\begin{align}
    \tilde{t}\leq T \nparams \left(\frac{e\nparams}{\expansionorder}\right)^\expansionorder\in\OC(T\nparams^{\expansionorder+1} \tilde{\epsilon}^2) = T\nparams^{\OC(\expansionorder)}.
\end{align}
We further observe that, for sufficiently large constant $\expansionorder$, 
\begin{align}
    \nparams^{\OC(\expansionorder)} = \nparams^{\frac{1}{\log(\expansionorder)}\OC\left(\expansionorder \log (\expansionorder )\right) } 
= \nparams^\frac{\OC\left(\log(\norm{\vec{a}}_1^2/\epsilon^2)\right)}{\log(\expansionorder)} = \OC\left(\nparams^{\log(1/\epsilon^2)}\right)\,,
\end{align}
where we recall that $\norm{\vec{a}}_1^2\in\OC(1)$. Thus, we conclude that the runtime of the Pauli Propagation algorithm for the observable $O$ can be upper bounded by
\begin{align}
    t \leq N_{\mathrm{Paulis}} \cdot  \tilde{t} \in \OC\left(N_{\mathrm{Paulis}} T \, {m}^{\log(1/\epsilon^2)}\right).
\end{align}
\end{proof}

\begin{theoremsup}[Time complexity of small-angle Pauli propagation (worst-case)]\label{thm:time-complexity-worst-PP-classical}
Let $O= \sum_{P\in\mathcal{P}_n} a_P P$ be an observable containing $N_{\mathrm{Paulis}}$ different Pauli terms and satisfying $\norm{\vec{a}}_1 = \OC(1)$. Furthermore, assume that given some $P$, we can efficiently estimate  $\Tr[\rho P]$. Let $T$ be the time necessary to apply a gate to a propagating Pauli operator
Then small-angle Pauli propagation can be used to both classically simulate \textit{and} surrogate a quantum circuit of the form of Eq.~\eqref{eq:CliffordVQAcircuit} over sampled zero-centered hypercube of width $2r$, i.e. $\vol(\vec{0}, r)$ with $r\in\OC(1/\nparams)$.
In particular, running the algorithm for time
\begin{equation}
       t=N_{\mathrm{Paulis}} T \nparams \left(\frac{e\nparams}{\expansionorder}\right)^\expansionorder \in  \OC\left(N_{\mathrm{Paulis}} \,T\, {\nparams}^{\log(1/\epsilon)}\right)\,  ,
\end{equation}
we can ensure that
\begin{align}
   \max_{\vec{\alpha} \in \vol(\vec{0}, r)} \left|f(\alv)-\fsur_{\expansionorder}(\alv)\right| \leq \norm{\vec{a}}_1\left(\frac{er\nparams}{\expansionorder}\right)^\expansionorder =\epsilon\,.
\end{align}
\end{theoremsup}

\begin{proof}
This proof is similar to that of \supt~\ref{thm:time-complexity-PP-classical-avg}, but here we are considering the worst-case error upper bound. Let us rewrite some part of previous proof for completeness. First, we fix $c\in\OC(1)$ such that $r=c/\nparams$. Let $\tilde{t}$ be the time required to estimate the expectation value of a single Pauli observable $P$ up to error $\tilde\epsilon = \epsilon/\norm{\vec{a}}_1$ for any points $\alv\in\vol(\vec{0}, c/\nparams)$. 
Then, invoking Lemma~\ref{lem:trans}, we find that the time required to estimate the expectation value of $O$ up to error $\epsilon$ is 
\begin{align}
    \tilde{t}\cdot N_{\mathrm{Paulis}}. 
\end{align}

 Then, as long as the truncation order satisfies
\begin{align}
    \expansionorder \geq c,
\end{align}
\supt~\ref{thm:worst-case-error-single-Pauli} implies that, for any given Pauli operator $P$, the Pauli Propagation algorithm returns an operator $\widetilde{P}_{\thv}$ satisfying:
\begin{align}
    \abs{\Tr[ P U(\vec{\alpha}) \rho U^\dag(\vec{\alpha}) - \widetilde{P}_{\thv}\rho]}
      \leq \left(\frac{ec}{\expansionorder}\right)^{\expansionorder}.
\end{align}
In order to achieve the desired error, we set
\begin{align}
    &\widetilde\epsilon:= \left(\frac{ec}{\expansionorder}\right)^{\expansionorder}
    \\&\implies \expansionorder \log \expansionorder \in \Theta(\log(1/\widetilde\epsilon)).
\end{align}

As observed in Eq.~\eqref{eq:number-paths}, the number of paths found by the Pauli Propagation algorithm is at most
\begin{align}
    n_{pth}\leq \left(\frac{e\nparams}{\expansionorder}\right)^\expansionorder = \nparams^\expansionorder \cdot \tilde\epsilon\cdot\left(\frac{1} {c}\right)^\expansionorder \in  \OC(\nparams^\expansionorder \widetilde\epsilon).
\end{align}

And therefore the runtime for a single Pauli operator in Eq.~\eqref{eq:runtime-single-Pauli} is at most
\begin{align}
    \tilde{t}\leq  T \nparams\left(\frac{e\nparams}{\expansionorder}\right)^\expansionorder\in\OC(T \nparams^{\expansionorder+1} \widetilde\epsilon) = T \nparams^{\OC(\expansionorder)}.
\end{align}
We further observe that, for sufficiently large constant $\expansionorder$, 
\begin{align}
    \nparams^{\OC(\expansionorder)} = \nparams^{\frac{1}{\log(\expansionorder)}\OC(\expansionorder \log (\expansionorder )) } 
= \nparams^\frac{\OC\left(\log(\norm{\vec{a}}_1/\epsilon)\right)}{\log(\expansionorder)} = \OC(\nparams^{\log(1/\epsilon)})\,.
\end{align}
Thus, we conclude that the runtime of the Pauli Propagation algorithm for the observable $O$ can be upper bounded by
\begin{align}
    t \leq N_{\mathrm{Paulis}} \cdot  \tilde{t} \in \OC\left(\, N_{\mathrm{Paulis}} \, T\, {m}^{\log(1/\epsilon)}\right).
\end{align}
\end{proof}

An interesting corollary of our theorem is that in the cases where the worst-case results hold, our simulation method can reduce Trotter error in quantum dynamics simulations with ``only'' a polynomial slow-down. 
To reduce Trotter error, one needs to reduce $dt$ and in turn inversely proportionally increases the number of layers $L$ to leave the total simulation time $dt L$ unchanged. 
That is, if we reduce $dt$ by a factor $1/c$ we need to increase the number of layers by a factor of $c$, which in turn increases the number of parameters by a factor $c$. The maximum $d t$ we can handle with small-angle Pauli Propagation with worst case error guarantees is determined by $r$. As the error for small-angle Pauli propagation scales in the worse case simulation as $m r$ (this is easiest to see from our formal statement in Theorem~\ref{thm:time-complexity-worst-PP-classical}) we see that reducing Trotter error by reducing $dt$ by a factor $1/c$ and increasing the total number of layers leaves the final simulation error unchanged. However, the total simulation time increases polynomially in $m$ and so we obtain a polynomial overhead in time complexity of small-angle Pauli Propagation resulting from the total number of parameters increasing by a factor of $c$.  More crudely speaking, if our algorithm can accurately simulate to time $t$ with large time step $dt$ in polynomial time, then it can simulate to time $t$ with small time step $dt$ in polynomial time.

\subsection{Sample complexity of the Pauli Propagation algorithm}
\label{app:pauli-surrogate-quantum-measurement}
In the previous section, we demonstrated that Pauli Propagation enables efficient simulation of quantum circuits composed of Clifford gates and Pauli rotations within a small angle range. To upper bound the computational complexity, we crucially assumed that the input state \(\rho\) is classically simulable, meaning that for any Pauli operator \(P\), we can efficiently estimate \(\Tr[P\rho]\).

In the most general case, $\rho$ may be an unknown state that we can access through quantum measurements. 
In this section, we combine the Pauli Propagation guarantee from Appendix~\ref{app:smallanglePP} with the randomized measurement protocol from Appendix~\ref{app:estimating-observables}, achieving a polynomial sample complexity for arbitrarily small constant error, i.e., $\epsilon\in\Theta(1)$. We further show that, if the final back-propagated Pauli operators have low Pauli weight, then the number of required Pauli measurements scales logarithmically with the system size.

\subsubsection{The truncated Heisenberg-evolved observable}
\label{sec:truncated-propagated-observable}
\noindent\textbf{Pauli observables.} Given a Pauli operator $Q\in \mathcal{P}_n$ and a circuit $U(\vec{\alpha})$ as input, the Pauli-path simulation algorithm returns the observable
\begin{align}
    \widetilde{Q}(\vec{\alpha}) = \sum_{\omv \not\in \mathcal{S}_\expansionorder^{\nparams}} \Phi_{\omv}(\vec{\alpha}) Q_{\omv} = \sum_{P\in\mathcal{P}_n} c^{(Q)}_P(\vec{\alpha}) P,
\end{align}
where we merged different Pauli paths by introducing the parametrized coefficients
\begin{align}
\label{eq:definition-coeff-c_P(Q)}
    c^{(Q)}_P(\vec{\alpha}) = \sum_{\substack{\omv \not\in \mathcal{S}_\expansionorder^{\nparams} \\ \text{ s. t. } Q_{\omv} = P}}   \Phi_{\omv}(\vec{\alpha}).
\end{align}
Recall that $Q_{\omv}$ in previous expressions is the Pauli obtained after back-propagating $Q$ along the path $\omv$.
For our sample complexity analysis, we need to upper bound the effective 1-norms $\norm{\vec{c}^{(Q)}}_{1, \mathrm{avg}}$ and $\norm{\vec{c}^{(Q)}}_{1, \mathrm{worst}}$ introduced in Definition~\ref{def:effective-1-norms}. We start by providing a loose, yet polynomial in the number of gates, upper bound. 
First, we notice that the average-case effective 1-norm is always upper bounded by the worst-case effective 1-norm. Let $\mathcal{D}$ be a distribution over the set $\mathcal{A}$. We have
\begin{align}
   \norm{\vec{c}^{(Q)}}_{1, \mathrm{avg}}= \sum_{P}\sqrt{\mathbb{E}_{\boldsymbol{\alpha} \sim \mathcal{D}} [c_P^{(Q)}(\alv)^2]} \leq 
   \sum_P \max_{\alv \in \mathcal{A}} \abs{c_P^{(Q)}(\alv)} = \norm{\vec{c}^{(Q)}}_{1, \mathrm{worst}}.
\end{align}

Furthermore, we can upper bound $\norm{\vec{c}^{(Q)}}_{1, \mathrm{worst}}$ as done in the proof of \supt~\ref{thm:worst-case-error-single-Pauli}:
\begin{align}
\label{eq:effective-1-norm-upperbound-single-Pauli}
    \norm{\vec{c}^{(Q)}}_{1, \mathrm{avg}} \leq \norm{\vec{c}^{(Q)}}_{1, \mathrm{worst}} \leq \left(\frac{e\nparams r}{\expansionorder}\right)^\expansionorder
    = \Theta\left((r\nparams)^\expansionorder\right).
\end{align}
Indeed, this is obtained by applying triangle inequality on $\abs{c_P^{(Q)}(\alv)}$ in Eq.~\eqref{eq:definition-coeff-c_P(Q)} and upper bounding each $|\Phi_{\omv}(\vec{\alpha})|$ by $r^{\#\sin(\omv)}$ as done in Eq.~\eqref{eq:upperbound-path-coefficient-sum}.
\medskip

\noindent\textbf{General observables.}
Let us consider a general observable expressed in the Pauli basis 
\begin{align}
    O = \sum_{P\in\mathcal{P}_n} a_P P,
\end{align}
with $\norm{\boldsymbol{a}}_1 \in \OC(1)$.
The final back-propagated observables equals
\begin{align}
       \widetilde{O}(\vec{\alpha}) =  \sum_{Q\in\{I,X,Y,Z\}^{\otimes n}}a_Q\widetilde{Q}(\vec{\alpha})
       = \sum_{Q,P\in\{I,X,Y,Z\}^{\otimes n}} a_Q c^{(Q)}_P(\vec{\alpha}) P
       := \sum_{P\in\{I,X,Y,Z\}^{\otimes n}}  c_P(\vec{\alpha}) P , \label{eq:general-back-obs}
\end{align}
where we introduced the parametrized coefficients
\begin{align}\label{eq:c_P-coeff-multiple-Pauli-definition}
    c_P(\vec{\alpha}) = \sum_{Q\in\{I,X,Y,Z\}^{\otimes n}} a_Q c^{(Q)}_P(\vec{\alpha}).
\end{align}

Then, we can upper-bound the effective norms as follows:

\begin{equation} \label{eq:upperbound-effective-1-norms-multi-Paulis}
  \norm{\boldsymbol{c}}_{1,\mathrm{avg}} \leq   \norm{\boldsymbol{c}}_{1,\mathrm{worst}} \leq \norm{\boldsymbol{a}}_1 \max_Q \norm{\boldsymbol{c}^{(Q)}}_{1,\mathrm{worst}}
    \in \Theta((r\nparams)^\expansionorder). 
\end{equation}
\begin{proof}
    From Eq.~\eqref{eq:c_P-coeff-multiple-Pauli-definition}, we can use triangle inequality to get
\begin{align}
     \max_{\alv\in\AC}\left|c_P(\alv)\right| &\leq \max_{\alv\in\AC} \sum_{Q\in\PC_n} |a_Q| |c_P^{(Q)}(\alv)| \\
     &\leq  \sum_{Q\in\PC_n} |a_Q| \max_{\alv\in\AC} |c_P^{(Q)}(\alv)|\;.
\end{align}
Now, we can use this result to upper bound $\norm{\boldsymbol{c}}_{1,\mathrm{worst}}$ as follows
\begin{align}
    \norm{\boldsymbol{c}}_{1,\mathrm{worst}} &\leq \sum_{P\in\PC_n}\sum_{Q\in\PC_n} |a_Q| \max_{\alv\in\AC} |c_P^{(Q)}(\alv)| \\
    &=\sum_{Q\in\PC_n} |a_Q| \norm{\boldsymbol{c}^{(Q)}}_{1,\mathrm{worst}} \\
    &\leq \norm{\vec{a}}_1 \max_Q\norm{\boldsymbol{c}^{(Q)}}_{1,\mathrm{worst}}\;,
\end{align}
where the last inequality is simply obtained by upper bounding each $\norm{\boldsymbol{c}^{(Q)}}_{1,\mathrm{worst}}$ in the sum by $\max_Q\norm{\boldsymbol{c}^{(Q)}}_{1,\mathrm{worst}}$.
The average-case effective 1-norm is always upper bounded by the worst-case one. Therefore, we recover Eq.~\eqref{eq:upperbound-effective-1-norms-multi-Paulis} using Eq.~\eqref{eq:effective-1-norm-upperbound-single-Pauli} and the assumption that $\norm{\vec{a}}_1\in\OC(1)$.
\end{proof}

\subsubsection{Uncorrelated parameters}
\label{sec:uncorrelated-parameters-PP-surrogate-guarantee-quantum}
Here, we assume that the each rotation parameter $\alpha_i$ is sampled independently from the uniform distribution over the interval $[-r,r]$. However, we remind the reader that this assumption can be relaxed, as the same results hold for any symmetric distribution satisfying the variance bound presented in Appendix~\ref{apx:average_uncorr}.

\begin{theoremsup}[Average-case Pauli surrogation guarantee]
\label{thm:avg-guarantee-PP-uncorr-quantum}

Consider an observable $O = \sum_{P\in \mathcal{P}_n} a_P P$, satisfying  $\norm{\vec{a}}_1 \in\mathcal{O}(1)$ and containing $N_{\mathrm{Paulis}}$ distinct Pauli terms. Then running a classical Pauli Propagation algorithm in time 
\begin{align}
     t\leq  N_{\rm Paulis}T\nparams  \, \left(\frac{e\nparams}{\expansionorder}\right)^{\expansionorder} \in \OC\left(N_{\rm Paulis} T \nparams^{\log(1/\epsilon^2)}\right)\;,
\end{align} 
and performing 
\begin{align}
     N_s = \left(\frac{3e\nparams}{\expansionorder}\right)^{\expansionorder} \in \OC\left( \nparams^{\log(1/\epsilon^2)}\right)\;
\end{align} 
randomized Pauli measurements, 
it is possible to estimate a function $\widehat{\fsur_{\expansionorder}}(\vec{\alpha})$ such that for all $r \in \OC(1/\sqrt{m})$, we have 
\begin{align}
    \Ebb_{\MC,\vec{\alpha}}\left[\left(\widehat{\fsur_{\expansionorder}}(\vec{\alpha}) - {f}(\vec{\alpha})\right)^2\right] \leq \epsilon^2,
\end{align}
where the expectation is both over the randomness of the parameters $\vec{\alpha}\sim  \uni(\bold{0}, r)$ and the internal randomness of the measurement protocol denoted by $\MC$.
\end{theoremsup}
\begin{proof}

First, let us bound the total mean square error by both contribution of the empirical error and the truncation error using that $(x+y)^2\leq 2(x^2+y^2)$ for any reals $x$ and $y$. This leads to
\begin{align}
\Ebb_{\vec{\alpha} ,\MC}\left[\left(\widehat{\fsur_{\expansionorder}}(\alv)-f(\alv)\right)^2\right]&\leq 2\Ebb_{\vec{\alpha}, \MC}\left[\left(\widehat{\fsur_{\expansionorder}}(\alv)-\fsur_{\expansionorder}(\alv)\right)^2\right] + 2\Ebb_{\vec{\alpha} }\left[\left(\fsur_\expansionorder(\alv)-f(\alv)\right)^2\right]\;.
\end{align}
Notice that the truncation error has no dependence on the measurement output. Therefore, the proof consists in two steps: first, we upper bound the runtime of Pauli Propagation for a given truncation error, and then we upper bound the number of required Pauli measurements for a given empirical error.

\noindent\textbf{Pauli propagation step.} 
\medskip

\supt~\ref{thm:time-complexity-PP-classical-avg} ensures that the averaged truncation error in the Pauli Propagation step is at most
\begin{equation}
\label{eq:MSE-upperbound-for-truncation-PP-proof}
  \Ebb_{\vec{\alpha} \sim  \uni(\bold{0}, r)}\left[\left(\fsur_\expansionorder(\alv)-f(\alv)\right)^2\right] \leq  \norm{\vec{a}}_1^2\left(\frac{e\nparams r^2}{3\expansionorder}\right)^{\expansionorder}
\end{equation}
with time 
\begin{equation}
    t\leq N_{\rm Paulis}T\nparams  \left(\frac{e\nparams}{\expansionorder}\right)^{\expansionorder}\;.
\end{equation}
\noindent\textbf{Surrogation step.} It remains to upper bound the number of Pauli measurements. The observable of interest can be written as
\begin{align}
    \widetilde{O}(\vec{\alpha}) = \sum_{P} a_P \widetilde{P}(\vec{\alpha}).
\end{align}
The average-case effective 1-norm of each observable $\widetilde{P}(\vec{\alpha})$ is upper bounded in Eq.~\eqref{eq:effective-1-norm-upperbound-single-Pauli}, and therefore from Eq.~\eqref{eq:upperbound-effective-1-norms-multi-Paulis} 
 average-case effective 1-norm of $\widetilde{O}(\vec{\alpha})$ is at most
 \begin{align}
     \norm{\boldsymbol{a}}_1 \cdot \left(\frac{e\nparams r}{\expansionorder} \right)^\expansionorder\,.
 \end{align}
 Invoking Lemma~\ref{lemma:multiple-observable-mean-squared-error}, we find that
 \begin{align}
 \label{eq:MSE-empirical-1-norm-Pauli-surrogate-proof}
     \Ebb_{\vec{\alpha} \sim  \uni(\bold{0}, r)}\left[\left(\widehat{\fsur_{\expansionorder}}(\alv)-\fsur_{\expansionorder}(\alv)\right)^2\right] \leq \frac{\norm{\boldsymbol{a}}_1^2}{N_s}  \left(\frac{e\nparams r}{\expansionorder} \right)^{2\expansionorder}\;,
 \end{align}
where $N_s$ is the number of Pauli measurements (i.e., the number of measurement shots). Now, let us impose that mean square error upper bounds from both Eq.~\eqref{eq:MSE-upperbound-for-truncation-PP-proof} and Eq.~\eqref{eq:MSE-empirical-1-norm-Pauli-surrogate-proof} are equal to $(\epsilon/2)^2$ (to ensure a total mean square error $\epsilon^2$)
\begin{equation}
    \left(\frac{\epsilon}{2}\right)^2=\frac{\norm{\boldsymbol{a}}_1^2}{N_s}\left(\frac{e\nparams r}{\expansionorder} \right)^{2\expansionorder}=\norm{\vec{a}}_1^2\left(\frac{e\nparams r^2}{3\expansionorder}\right)^{\expansionorder}\;.
\end{equation}
From this equation, we can solve for $N_s$, leading to 
\begin{equation}
    N_s=\left(\frac{3e\nparams}{\expansionorder}\right)^{\expansionorder}\;.
\end{equation}
Finally, we follow the same  analysis as in the proof of \supt~\ref{thm:time-complexity-PP-classical-avg} concerning the scaling of the quantity of interest. Therefore, for the time complexity, we have 
\begin{equation}
    t\leq N_{\rm Paulis}T\nparams  \left(\frac{e\nparams }{\expansionorder}\right)^{\expansionorder} \in \OC\left(N_{\rm Paulis} T\nparams^{\log(1/\epsilon^2)}\right)\;.
\end{equation}

For the sample complexity, we notice that $\OC(N_s)\sim \OC(t/(N_{\rm Paulis} T))$ as
\begin{equation}
    N_s = \left(\frac{3e\nparams}{\expansionorder}\right)^{\expansionorder} \in \OC\left( \nparams^{\log(1/\epsilon^2)}\right)\;.
\end{equation}
\end{proof}

Next, we consider the  special case of particular interest when the truncated, back-propagated observable contains only Pauli terms with low weight, i.e., Pauli terms acting non-trivially on few qubits. In this case, we can employ the machinery of classical shadow tomography~\cite{huang2020predicting, elben2022randomized} to obtain an alternative bound on the number of randomized Pauli measurements.

\begin{theoremsup}[Classical shadow tomography using randomized Pauli measurements (Theorem 1 and Lemma 3 in Ref.\ \cite{huang2020predicting})] \label{thm:classical-shadow}
    Given an unknown $n$-qubit state $\rho$ and a Pauli operator $P$ of weight $\abs{P}=k$, after $N_s$ randomized Pauli measurements on copies of $\rho$ satisfying
    \begin{equation}
        N_s \in \OC\left( \frac{ 3^k\log(1/\delta) }{\epsilon^2} \right),
    \end{equation}
    we can estimate $\Tr(P \rho)$ to $\epsilon$ error with probability $1-\delta$.
\end{theoremsup}

\begin{theoremsup}[Shadow Surrogation Guarantee, Formal]
\label{thm:shadow-surrogation-guarantee-formal}
Under the same hypotheses of \supt~\ref{thm:avg-guarantee-PP-uncorr-quantum}, let us further assume that the truncated backpropagated observable contains only Pauli terms with constant Pauli weight, i.e.
\begin{align}
        \widetilde{O}(\vec{\alpha}) =  \sum_{\substack{P\in \mathcal{P}_n\\\abs{P}\in\OC(1)}}  c_P(\vec{\alpha}) P.
\end{align}
Then the number of required Pauli measurements improves to
\begin{align}
    N_s \in  \min \bigg( \log(n), \log(m \,N_{\mathrm{Paulis}})\bigg) \cdot \OC\left(\frac{\log(1/\delta)}{\epsilon}\right)
\end{align}
with probability $1-\delta$ over the internal randomness of the measurement protocol.
\end{theoremsup}

\begin{proof}
In the first part of this proof, we restrict our attention to the case where the initial observable is a Pauli $O=Q$ with Pauli weight $\abs{Q} \in \OC(1)$.
The final, back-propagated observable is
\begin{align}
    \widetilde{Q}(\vec{\alpha}) = \sum_{\omv \not\in \mathcal{S}_\expansionorder^{\nparams}} \Phi_{\omv}(\vec{\alpha}) Q_{\omv} ,
\end{align}
where for  all $\omv \not\in \mathcal{S}_\expansionorder^{\nparams}$, the operator $Q_{\omv} $ satisfies $\abs{Q_{\omv} }\in \OC(1)$.
We notice that the observable $\widetilde{Q}(\vec{\alpha}) $ contains at most 
\begin{align}
    \min\bigg\{\poly(n), \poly(m) \bigg\}
\end{align}
different Paulis. This follows from the fact that it contains $\OC(\poly(m))$ different Pauli paths, and that, given a constant $k$, there exists only $\OC(n^k)$ Paulis of weight at most $k$.

Then, we can use a classical shadows protocol for predicting some estimate $\hat{o}_{\omv}$ such that
\begin{align}
    \forall \omv \not\in \mathcal{S}_\expansionorder^{\nparams} : \abs{\Tr[Q_{\omv} \rho] - \hat{o}_{\omv}}^2\leq \epsilon,
\end{align}
with probability at least $1-\delta$. By \supt~\ref{thm:classical-shadow}, this requires a number of Pauli measurement scaling as
\begin{align}
    N_s \in \OC\left(\log(\frac{\log(\min(m,n))}{\delta}) \frac{1}{\epsilon} \right)
\end{align}
which in turn leads to
\begin{align}
    \Ebb_{\boldsymbol{\alpha}}\left(\Tr[ \widetilde{Q}(\vec{\alpha})\rho] - \sum_{\omv\in \not\in \mathcal{S}_\expansionorder^{\nparams} }  \Ebb_{\boldsymbol{\alpha}}\Phi_{\omv}(\boldsymbol{\alpha}) \hat{o}_{\omv} \right)^2
    = &\sum_{\omv\in\not\in \mathcal{S}_\expansionorder^{\nparams} } \Ebb_{\boldsymbol{\alpha}} \Phi_{\omv}^2(\boldsymbol{\alpha}) \left(\Tr[Q_{\omv}\rho]- \hat{o}_{\omv}\right)^2
    \\\leq &\max_{\omv\in\not\in \mathcal{S}_\expansionorder^{\nparams} } \left(\Tr[Q_{\omv}\rho]- \hat{o}_{\omv}\right)^2  \Ebb_{\boldsymbol{\alpha}}\norm{\Phi(\boldsymbol{\alpha})}_2^2 \leq \epsilon.
\end{align}
In the first step above we used the fact that different Fourier coefficients are uncorrelated, i. e. $ \Ebb_{\boldsymbol{\alpha}} \Phi_{\omv}(\boldsymbol{\alpha})\Phi_{\omv'}(\boldsymbol{\alpha})$ for $\omv \neq \omv'$, and in the  last step we used the fact that 
    \begin{equation}
\sum_{\substack{\omv \\ \#\sin(\omv)>\expansionorder}}\Ebb_{\alv}\left[\Phi_{\omv}(\alv)^2\right] \leq \left(\frac{e\nparams r^2}{3\expansionorder}\right)^\expansionorder\leq 1\;,
    \end{equation}
    which is provided by \supt~\ref{thm:mse-avg-supp} or more precisely Corollary~\ref{cor:MSE-uniform-uncorrelated}.
\medskip

We can upper bound the error for a general observable composed of $N_{\mathrm{Paulis}}$ different Pauli terms invoking Lemma~\ref{lem:trans}.
In particular, we need to rescale the error $\epsilon$ according to the 1-norm $\norm{\vec{a}}_1$, which we assumed to be $\OC(1)$. Moreover, the number of different Pauli terms is at most
\begin{align}
    \min\bigg\{\poly(n), N_{\mathrm{Paulis}}\poly(m) \bigg\},
\end{align}
and therefore the required number of measurements is 
\begin{align}
    N_s \in \OC\left(\log(\frac{\min(m\, N_{\mathrm{Paulis}},n)}{\delta}) \frac{1}{\epsilon} \right).
\end{align}
\end{proof}
\subsubsection{Worst-case parameters}
\label{sec:worst-case-PP-surrogate-guarantee-quantum}
\begin{theoremsup}[Worst-case Pauli surrogation guarantee]
\label{thm:worst-case-guarantee-PP-quantum}

Consider an observable $O = \sum_{P\in \{I,X,Y,Z\}^{\otimes n}} a_P P$, 
satisfying  $\norm{\vec{a}}_1 \in\mathcal{O}(1)$ and $\norm{\vec{a}}_0 = N_{\mathrm{Paulis}}$, i.e. $O$ contains $N_{\mathrm{Paulis}}$ Pauli terms.

Then running a classical Pauli Propagation algorithm in time 
\begin{align}
     t\leq N_{\rm Paulis}T\nparams  \left(\frac{e\nparams}{\expansionorder}\right)^{\expansionorder} \in \OC\left(N_{\rm Paulis} T \nparams^{\log(1/\epsilon)}\right)\;,
\end{align} 
and performing  
\begin{align}
     N_s=2\log(2/\delta) \in\OC(\log(1/\delta))\;
\end{align} 
randomized Pauli measurements, 
it is possible to estimate a function $\widehat{\fsur_{\expansionorder}}(\vec{\alpha})$ such that for all $r \in \OC(1/\nparams)$, we have 
\begin{align}
    \abs{\widehat{\fsur_{\expansionorder}}(\vec{\alpha}) - {f}(\vec{\alpha})} \leq \epsilon,
\end{align}
with probability at least $1-\delta$ over both the randomness of $\vec{\alpha}\sim \uni(\bold{0}, r)$ and the internal randomness of the measurement protocol.
\end{theoremsup}
\begin{proof}
 The proof is analogous to the one of \supt~\ref{thm:avg-guarantee-PP-uncorr-quantum}. We start by split the total error as both truncation and empirical error contributions using triangle inequality, leading to
 \begin{equation}
     \abs{\widehat{\fsur_{\expansionorder}}(\vec{\alpha}) - {f}(\vec{\alpha})}\leq \abs{\widehat{\fsur_{\expansionorder}}(\vec{\alpha}) - {\fsur_\expansionorder}(\vec{\alpha})}+\abs{\fsur_{\expansionorder}(\vec{\alpha}) - {f}(\vec{\alpha})}\;.
 \end{equation}
Next, the proof consists of two steps. First, we upper bound the runtime of Pauli Propagation for a given truncation error, and second,  we upper bound the number of required Pauli measurements for a given empirical error.

\medskip

\noindent\textbf{Pauli propagation step.} \supt~\ref{thm:time-complexity-worst-PP-classical} ensures that the truncation error in the Pauli Propagation step is at most
\begin{equation}
  \abs{\fsur_{\expansionorder}(\vec{\alpha}) - {f}(\vec{\alpha})} \leq \norm{\vec{a}}_1\left(\frac{e\nparams r}{\expansionorder}\right)^\expansionorder
\end{equation}
with time 
\begin{equation}
    t\leq N_{\rm Paulis}T\nparams \left(\frac{e\nparams }{\expansionorder}\right)^{\expansionorder}\;.
\end{equation}
\noindent\textbf{Surrogation step.} It remains to upper bound the number of Pauli measurements. The observable of interest can be written as
\begin{align}
    \widetilde{O}(\vec{\alpha}) = \sum_{P} a_P \widetilde{P}(\vec{\alpha}).
\end{align}
The worst-case effective 1-norm related to each observable $\widetilde{P}(\vec{\alpha})$ is upper bounded in Eq.~\eqref{eq:effective-1-norm-upperbound-single-Pauli}, and therefore from Eq.~\eqref{eq:upperbound-effective-1-norms-multi-Paulis} 
 worst-case effective 1-norm related to $\widetilde{O}(\vec{\alpha})$ is at most
 \begin{align}
     \norm{\boldsymbol{a}}_1 \cdot \left(\frac{e\nparams r}{\expansionorder} \right)^\expansionorder\,.
 \end{align}
 Invoking Lemma~\ref{lemma:multiple-observable-max-error}, ensures
 \begin{align}
 \label{eq:WC-empirical-1-norm-Pauli-surrogate-proof}
     \left|\widehat{\fsur_{\expansionorder}}(\alv)-\fsur_{\expansionorder}(\alv)\right| &\leq \sqrt{\frac{2\log(2/\delta)}{N_s}}\norm{\boldsymbol{a}}_1  \left(\frac{e\nparams r}{\expansionorder} \right)^\expansionorder\;,
 \end{align}
 with probability at least $1-\delta$ (over both parameters and measurement output) where $N_s$ is the number of Pauli measurements (i.e. shots).

Now, let us assume that both errors are upper bounded by $\epsilon/2$ (to ensure a total error $\epsilon$), so that
\begin{equation}
    \frac{\epsilon}{2}=\sqrt{\frac{2\log(2/\delta)}{N_s}}\norm{\boldsymbol{a}}_1  \left(\frac{e\nparams r}{\expansionorder} \right)^\expansionorder=\norm{\vec{a}}_1\left(\frac{e\nparams r}{\expansionorder}\right)^\expansionorder\;.
\end{equation}
From the previous,  we can solve for $N_s$ to get
\begin{equation}
    N_s=2\log(2/\delta)\;.
\end{equation}
Finally, we follow the same  analysis as in the proof of \supt~\ref{thm:time-complexity-worst-PP-classical} concerning the scaling of the quantity of interest. Therefore, for the time complexity, we have 
\begin{equation}
    t\leq N_{\rm Paulis}T\nparams  \left(\frac{e\nparams}{\expansionorder}\right)^{\expansionorder} \in \OC\left(N_{\rm Paulis} T \nparams^{\log(1/\epsilon)}\right)\;.
\end{equation}

Here, the sample complexity only depends on the probability $\delta$. Indeed, both errors are similar given that $N_s\sim 2\log(2/\delta)$, and $r\in\OC(1/\nparams)$ is small enough to ensure an arbitrary small error.  
\end{proof}

\subsubsection{Even allocation of shots}
\label{sec:even-allocation-PP-surogate-guarantee-quantum}
Here, we provide a proof of the surrogate guarantee based on \supt~\ref{thm:avg-guarantee-PP-uncorr-quantum} using even allocation of shots between the Paulis that needs to be measured instead of the effective $1$-norm allocation strategy used in \supt~\ref{thm:avg-guarantee-PP-uncorr-quantum}. As opposed to the warm-up example presented in Appendix~\ref{sec:warm-up}, here we fix the number of shots $n_s$ for each Pauli that needs to be measured in order to estimate the surrogate of the expectation function. Therefore, the total number of shots is at most $N_s=N_{\rm Paulis} n_{pth}^{(\expansionorder)} n_s$  recalling that $N_{\rm Paulis}$ is the number of Paulis contained by $O$ and $n_{pth}^{(\expansionorder)}$ is the number of path as mentioned previously.

\begin{theoremsup}[Surrogation via Direct Pauli Measurements, Formal]
\label{thm:avg-guarantee-PP-uncorr-quantum-even-allocation}

Consider an observable $O = \sum_{P\in \{I,X,Y,Z\}^{\otimes n}} a_P P$, 
satisfying  $\norm{\vec{a}}_1 \in\mathcal{O}(1)$ and $\norm{\vec{a}}_0 = N_{\mathrm{Paulis}}$, i.e. $O$ contains $N_{\mathrm{Paulis}}$ Pauli terms.

Then running a classical Pauli Propagation algorithm in time 
\begin{align}
     t\leq N_{\rm Paulis}T\nparams  \left(\frac{e\nparams}{\expansionorder}\right)^{\expansionorder} \in \OC\left(N_{\rm Paulis} T\nparams^{\log(1/\epsilon^2)}\right)\;,
\end{align} 
and performing 
\begin{align}
     N_s = N_{\rm Paulis} \left(\frac{e\nparams}{\expansionorder}\right)^{\expansionorder} \in \OC\left(N_{\rm Paulis} \nparams^{\log(1/\epsilon^2)}\right)\;
\end{align} 
randomized Pauli measurements, 
it is possible to estimate a function $\widehat{\fsur_{\expansionorder}}(\vec{\alpha})$ such that for all $r \in \OC(1/\sqrt{m})$, we have 
\begin{align}
    \label{eq:MSE-even-allocation-thm-app}
\Ebb_{\vec{\alpha},\MC}\left[\left(\widehat{\fsur_{\expansionorder}}(\vec{\alpha}) - {f}(\vec{\alpha})\right)^2\right] \leq \epsilon^2,
\end{align}
where the expectation is both over the randomness of $\vec{\alpha}\sim  \uni(\bold{0},r)$ and the internal randomness of the measurement protocol denoted by $\MC$.
\end{theoremsup}
Notice that applying Markov's inequality together with Eq.~\eqref{eq:MSE-even-allocation-thm-app}, we can ensure 
\begin{equation}
\abs{\widehat{\fsur_{\expansionorder}}(\vec{\alpha}) - {f}(\vec{\alpha})} < \tilde{\epsilon}\;,
\end{equation}
with probability at least $1-\delta$ (for any $\delta\in\Theta(1)$) over both $\alv\sim\uni(\bold{0},r)$ and the measurement output $\MC$ where $\tilde{\epsilon}=\frac{\epsilon}{\sqrt{\delta}}\in\Theta(1)$ can be made arbitrarily small by increasing  $\expansionorder\in\OC(1)$. This leads to the formal version of Theorem~\ref{thm:paulisurrogate-average} stated in the main text.
\begin{proof}
The proof here is the same as in \supt~\ref{thm:avg-guarantee-PP-uncorr-quantum} except the surrogation step (i.e., empirical error estimation). In the following, we show that for uniform allocation of shots Eq.~\eqref{eq:MSE-empirical-1-norm-Pauli-surrogate-proof} can be replaced by
        \begin{equation}
        \Ebb_{\alv,\MC}\left[\left(\widehat{\fsur_{\expansionorder}}(\alv)-\fsur_{\expansionorder}(\alv)\right)^2\right] \leq \frac{N_{\rm Paulis} \norm{\vec{a}}_2^2}{N_s}\left(\frac{e\nparams r}{\sqrt{3}\expansionorder}\right)^{2\expansionorder}\;.
    \end{equation}

    The propagated observable can be written as
    \begin{equation}
        \widetilde{O}(\alv) = \sum_{\substack{P\in\PC_n \\
        a_P\neq 0}}\sum_{\substack{\omv \\ \#\sin(\omv)>\expansionorder}} a_P\Phi_{\omv}(\alv) P_{\omv}\;,
    \end{equation}
    where $P_{\omv}$ 
 is the Pauli obtained after back-propagating $P$ along the path $\omv$. A worst-case scenario where each $P_{\omv}$ are different Paulis is considered here (that is, we don't group coefficients that corresponds to the same Pauli). So, for each Pauli $P$ and each path ${\omv}$ considered (i.e. each $P_{\omv}$), we will use $n_s$ shots to estimate $\Tr[\rho P_{\omv}]$. Let $\hat{x}_{P,{\omv}}^{(i)}$ denotes the $i$-th measurement outcome obtained by measuring $\Tr[\rho P_{\omv}]$ with the quantum computer. Then, our empirical estimate of $\Tr[\rho P_{\omv}]$ for each $P_{\omv}$ is simply the average over the $n_s$ corresponding measurement outcomes
    \begin{equation}
        \frac{1}{n_s}\sum_{i=1}^{n_s}\hat{x}_{P,\omv}^{(i)}\;.
    \end{equation}
    Therefore, the empirical estimate for the surrogate of the expectation function is given by
    \begin{equation}
        \widehat{\fsur_{\expansionorder}}(\alv)=\sum_{\substack{P\in\PC_n \\
        a_P\neq 0}}\sum_{\substack{\omv \\ \#\sin(\omv)>\expansionorder}}\sum_{i=1}^{n_s} \frac{a_P\Phi_{\omv}(\alv)}{n_s} \hat{x}_{P,\omv}^{(i)}\;,
    \end{equation}
    which is an unbiased estimator for $\fsur_{\expansionorder}(\alv)$. Moreover, notice that the $\hat{x}_{P,\omv}^{(i)}$ are uncorrelated (independent measurement). From the previous,  the mean square error (over both the parameters $\alv$ and the measurement process denoted by $\MC$) is given by
    \begin{align}
        \Ebb_{\MC,\alv}\left[\left(\widehat{\fsur_{\expansionorder}}(\alv)-\fsur_{\expansionorder}(\alv)\right)^2\right] &= \Ebb_{\alv}\left[\Var_{\MC}\left[\widehat{\fsur_{\expansionorder}}(\alv)\right]\right]\\
        &=\Ebb_{\alv}\left[\sum_{\substack{P\in\PC_n \\
        a_P\neq 0}}\sum_{\substack{\omv \\ \#\sin(\omv)>\expansionorder}}\sum_{i=1}^{n_s} \frac{a_P^2\Phi_{\omv}(\alv)^2}{n_s^2} \Var_{\MC}\left[\hat{x}_{P,\omv}^{(i)}\right] \right] \\
        &\leq \Ebb_{\alv}\left[\sum_{\substack{P\in\PC_n \\
        a_P\neq 0}}\sum_{\substack{\omv \\ \#\sin(\omv)>\expansionorder}}\sum_{i=1}^{n_s} \frac{a_P^2\Phi_{\omv}(\alv)^2}{n_s^2}\right] \\
        &=\frac{1}{n_s}\sum_{P\in\PC_n} a_P^2 \sum_{\substack{\omv \\ \#\sin(\omv)>\expansionorder}}\Ebb_{\alv}[\Phi_{\omv}(\alv)^2]\;,
    \end{align}
    where the inequality is obtained by noticing that $\Var_{\MC}\left[\hat{x}_{P,\omv}^{(i)}\right]=1-\Tr[\rho P_{\omv}]^2\leq 1$. Moreover, from \supt~\ref{thm:mse-avg-supp} or more precisely from Corollary~\ref{cor:MSE-uniform-uncorrelated}, we have
    \begin{equation}
\sum_{\substack{\omv \\ \#\sin(\omv)>\expansionorder}}\Ebb_{\alv}\left[\Phi_{\omv}(\alv)^2\right] \leq \left(\frac{e\nparams r^2}{3\expansionorder}\right)^\expansionorder\;.
    \end{equation}

    The total number of shots is given by $N_s=n_sN_{\rm Paulis}\; n_{pth}^{(\expansionorder)}$ where we recall that $N_{\rm Paulis}$ is the number of initial Paulis (i.e., the number of Pauli $P$ such that $a_P\neq 0$) and $n_{pth}^{(\expansionorder)}$ is the number of Pauli path in the surrogate (the number of path $\omv$ such that $\#\sin(\omv)>\expansionorder$). From Eq.~\eqref{eq:number-paths}, we have
    \begin{equation}
        n_{pth}^{(\expansionorder)}\leq \left(\frac{e\nparams}{\expansionorder}\right)^\expansionorder\;.
    \end{equation}
    Therefore, combining everything together leads to the following upper bound for the mean square empirical error
    \begin{equation}
        \Ebb\left[\left(\widehat{\fsur_{\expansionorder}}(\alv)-\fsur_{\expansionorder}(\alv)\right)^2\right] \leq \frac{N_{\rm Paulis}\; \norm{\vec{a}}_2^2}{N_s}\left(\frac{e\nparams r}{\sqrt{3}\expansionorder}\right)^{2\expansionorder}\;.
    \end{equation}

    Finally, substituting this in the proof of \supt~\ref{thm:avg-guarantee-PP-uncorr-quantum}, we get that the total number of shots required for even allocation is given by
    \begin{equation}
        N_s = \frac{\norm{\vec{a}}_2^2N_{\rm Paulis}}{\norm{\vec{a}}_1^2} \left(\frac{e\nparams}{\expansionorder}\right)^{\expansionorder}\;.
    \end{equation}
    Therefore, by taking at least 
    \begin{equation}
        N_s = N_{\rm Paulis} \left(\frac{e\nparams}{\expansionorder}\right)^{\expansionorder} \;,
    \end{equation}
    is sufficient to get the same error guarantee as in \supt~\ref{thm:avg-guarantee-PP-uncorr-quantum}.
\end{proof}

Notice here that with even allocation of shots we get a 2-norm instead of a 1-norm of the coefficients of the initial observable. However, we also get the additional factor $N_{\rm Paulis}$. Let us consider the warm-up example presented in Appendix~\ref{sec:warm-up} where the initial observable is given by
\begin{equation}
    O=Z^{\otimes n}+\frac{\tau}{4^n-1}\sum_{P\in \PC_n\backslash  \{Z^{\otimes n}\}}P\;,
\end{equation}
where $\tau\in\OC(1)$.
Therefore, we have $\norm{\vec{a}}_1^2=(1+\tau)^2\in\OC(1)$ and $\norm{\vec{a}}_1^2=1+\frac{\tau^2}{4^n-1}\in\OC(1)$, but $N_{\rm Paulis}=4^n$. Therefore, the number of shots required with even allocation (provided by the bound) is exponentially larger than the effective 1-norm allocation. In fact, this will highly depends on the number of Paulis in the initial observable. As we are mainly considering $\OC(\poly(n))$ Pauli in the observable, both bounds remains polynomial with $n$ and $\nparams$.

\subsection{Additional Numerical Results}\label{apx:heisenberg_PP}

\begin{figure}
    \centering
    \includegraphics[width=0.99\linewidth]{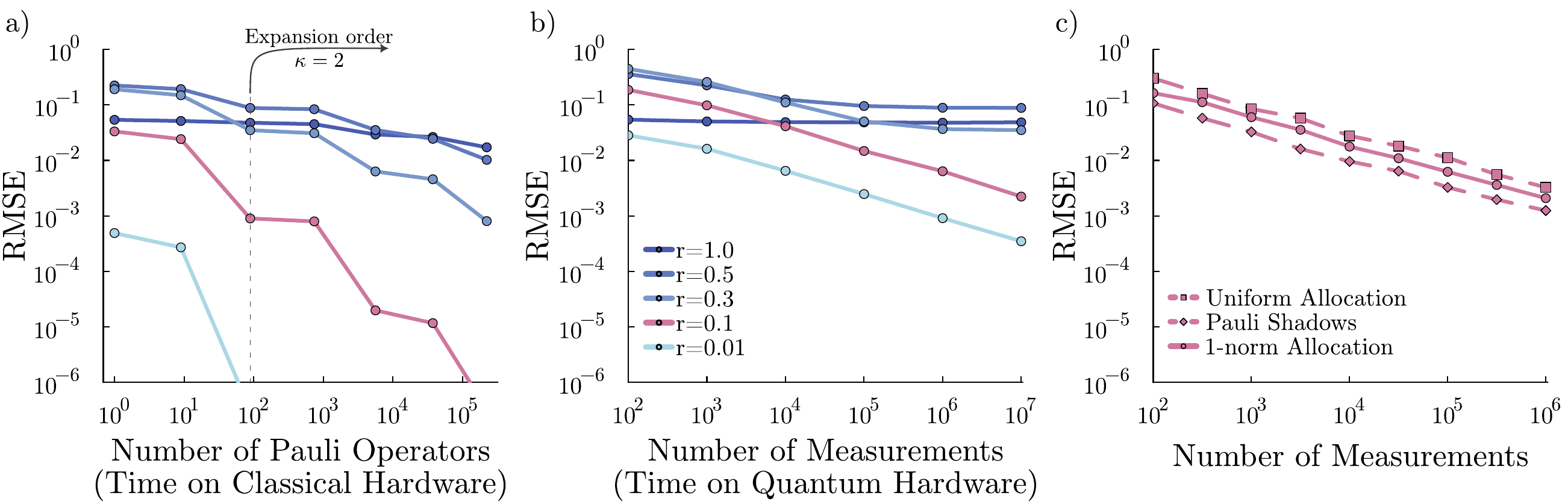}
    \caption{ \textbf{Landscape and surrogation RMSE as a function of classical and quantum resources with the Heisenberg Hamiltonian.} We adapt the 16-qubit setup of Fig.~\ref{fig:error-bounds} to quantum circuit representing a Heisenberg model HVA, and the results in panels b and c are for order $\expansionorder=2$. Other details remain the same. Overall, the Heisenberg HVA generates more Pauli strings than the TFI, which puts more strain on both the classical and quantum components. Furthermore, the Pauli shadow estimation strategy now outperforms the 1-norm allocation.}
    \label{fig:heisenberg_PP}
\end{figure}
To better elucidate the system-dependence of the PP surrogation in Fig.~\ref{fig:heisenberg_PP}, we here showcase the same numerical experiment with a Heisenberg model Hamiltonian variational ansatz (HVA) instead of a transverse-field Ising (TFI) HVA. This ansatz not only contains more Pauli rotations, but also more and larger groups of commuting operators. On the one hand, this leads to more Pauli strings being created during the classical simulation, in turn putting more strain on both the classical and quantum components. Generally speaking, the errors at the same resources are larger than in Fig.~\ref{fig:error-bounds}. However, we also see the Pauli shadow estimation strategy becoming superior to the 1-norm strategy, likely due to the increased structure and number of commuting terms. A notable difference between shadows and our 1-norm allocation is that the latter does not group commuting Pauli strings together for simultaneous estimation, which the shadow strategy implicitly does. This structure can be leveraged in sophisticated estimation strategies to design the best quantum-enhanced classical surrogate computation.

\end{document}